\documentclass[peerreview,a4paper,12pt]{IEEEtran}

\usepackage{amsthm}
\usepackage{amsmath}
\usepackage{amsfonts,amssymb,amsmath}
\usepackage{graphicx}
\usepackage{epsfig}
\usepackage[applemac]{inputenc}
\usepackage{latexsym}
\usepackage{dsfont}
\usepackage{xr}
\newtheorem{mydef}{Definition}
\newtheorem{mycor}{Corollary}
\newtheorem{myprop}{Proposition}
\newtheorem{mythe}{Theorem}
\newtheorem{mylem}{Lemma}
\newtheorem{mynot}{Notation}
\newtheorem{myrem}{Remark}

\newtheorem{myconj}{Conjecture}
\newtheorem{myex}{Example}

\DeclareMathOperator*{\scon}{scon}

\DeclareMathOperator{\proj}{Proj}

\begin{document}

\sloppy

\title{Ergodic Theory Meets Polarization. II: \\A Foundation of Polarization Theory} 

\author{
Rajai Nasser\\
School of Computer and Communication Sciences, EPFL\\
Lausanne, Switzerland\\
Email: rajai.nasser@epfl.ch
\thanks{This paper was presented in part at the IEEE International Symposium on Information Theory, Hong Kong, June 2015.}
}

\maketitle

\setcounter{page}{1}

\begin{abstract}
An open problem in polarization theory is to determine the binary operations that always lead to polarization (in the general multilevel sense) when they are used in Ar{\i}kan style constructions. This paper, which is presented in two parts, solves this problem by providing a necessary and sufficient condition for a binary operation to be polarizing. This (second) part provides a foundation of polarization theory based on the ergodic theory of binary operations which we developed in the first part \cite{RajErgI}. We show that a binary operation is polarizing if and only if it is uniformity preserving and its right-inverse is strongly ergodic. The rate of polarization of single user channels is studied. It is shown that the exponent of any polarizing operation cannot exceed $\frac{1}{2}$, which is the exponent of quasigroup operations. We also study the polarization of multiple access channels (MAC). In particular, we show that a sequence of binary operations is MAC-polarizing if and only if each binary operation in the sequence is polarizing. It is shown that the exponent of any MAC-polarizing sequence cannot exceed $\frac{1}{2}$, which is the exponent of sequences of quasigroup operations.
\end{abstract}

\section{Introduction}

The problem of finding a characterization for polarizing operations was discussed in the introduction of Part I of this paper \cite{RajErgI}. The first operation that was shown to be polarizing was the XOR operation in $\mathbb{F}_2$ (Ar{\i}kan \cite{Arikan}). \c{S}a\c{s}o\u{g}lu et al. generalized Ar{\i}kan's result and showed that if $q$ is prime, then the addition modulo $q$ in $\mathbb{F}_q$ is polarizing \cite{SasogluTelAri}. Park and Barg showed that if $q=2^r$ with $r>0$, then addition modulo $q$ in $\mathbb{Z}_q$ is polarizing \cite{ParkBarg}. Sahebi and Pradhan generalized these results and showed that all Abelian group operations are polarizing \cite{SahebiPradhan}. \c{S}a\c{s}o\u{g}lu showed that any alphabet can be endowed with a special quasigroup operation which is polarizing \cite{SasS}. The author and Telatar showed that all quasigroup operations are polarizing \cite{RajTel}.

In the context of multiple access channels (MAC), \c{S}a\c{s}o\u{g}lu et al. showed that if $q$ is prime, then addition modulo $q$ is MAC-polarizing for 2-user MACs, i.e., if $W$ is a 2-user MAC where the two users have $\mathbb{F}_q$ as the input alphabet, then using the addition modulo $q$ for the two users lead to a polarization phenomenon \cite{SasogluTelYeh}. Abbe and Telatar used Matroid theory to show that for binary input MACs with $m\geq 2$ users, using the XOR operation for each user is MAC-polarizing \cite{AbbeTelatar}. The author and Telatar showed that if $q_1,\ldots,q_m$ is a sequence of prime numbers and if $W$ is an $m$-user MAC with input alphabets $\mathbb{F}_{q_1}$,\ldots,$\mathbb{F}_{q_m}$, then using addition modulo $q_i$ for the $i^{th}$ user is MAC-polarizing \cite{RajTel}. This fact was used to construct polar codes for arbitrary MACs \cite{RajTelA}.

The ergodic theory of binary operations was developed in Part I \cite{RajErgI}. This part provides a foundation of polarization theory based on the results established therein. In section II we provide a formal definition of polarizing operations and MAC-polarizing sequences of binary operations. Section III proves that a binary operation is polarizing (in the general multilevel sense) if and only if it is uniformity preserving and its right-inverse is strongly ergodic. The exponent of polarizing operations is studied in section IV. It is shown that the exponent of every polarizing operation is at most $\frac{1}{2}$, which is the exponent of quasigroup operations. The polarization theory for MACs is studied in section V. We show that a sequence of binary operations is MAC-polarizing if and only if each operation in the sequence is polarizing. The exponent of every MAC-polarizing sequence is shown to be at most $\frac{1}{2}$ which is the exponent of sequences of quasigroup operations.

\section{Preliminaries}

Throughout this (second) part of the paper, we assume that the reader is familiar with the concepts of the ergodic theory of binary operations which were introduced in Part I \cite{RajErgI}.

All the sets that are considered in this paper are finite.

\subsection{Easy channels}

\begin{mynot}
A channel $W$ with input alphabet $\mathcal{X}$ and output alphabet $\mathcal{Y}$ is denoted by $W:\mathcal{X}\longrightarrow\mathcal{Y}$. The transition probabilities of $W$ are denoted by $W(y|x)$, where $x\in\mathcal{X}$ and $y\in\mathcal{Y}$. The probability of error of the ML decoder of $W$ for uniformly distributed input is denoted by $\mathbb{P}_e(W)$. The symmetric capacity of $W$, denoted $I(W)$, is the mutual information $I(X;Y)$, where $X$ and $Y$ are jointly distributed as $\mathbb{P}_{X,Y}(x,y)=\frac{1}{|\mathcal{X}|}W(y|x)$ (i.e., $X$ is uniform in $\mathcal{X}$ and it is used as input to the channel $W$ while $Y$ is the output).
\end{mynot}

\begin{mydef}
\label{defeasy}
A channel $W:\mathcal{X}\longrightarrow\mathcal{Y}$ is said to be $\delta$-\emph{easy} if there exists an integer $L\leq |\mathcal{X}|$ and a random code $\mathcal{B}$ of block length 1 and rate $\log L$ (i.e., $\mathcal{B}\in\mathcal{S}:=\{C\subset \mathcal{X}:\; |C|=L\}$), which satisfy the following:
\begin{enumerate}
\item $|I(W)-\log L|<\delta$.
\item For every $x\in\mathcal{X}$, we have $\displaystyle\sum_{C\in\mathcal{S}}\frac{1}{L}\mathbb{P}_{\mathcal{B}}(C)\mathds{1}_{x\in C}=\frac{1}{|\mathcal{X}|}$. In other words, if $C\in\mathcal{S}$ is chosen according to the distribution of $\mathcal{B}$ and $X$ is chosen uniformly in $C$, then the marginal distribution of $X$ as a random variable in $\mathcal{X}$ is uniform.
\item If for each $C\in\mathcal{S}$ we fix a bijection $f_{C}:\{1,...,L\}\rightarrow C$, then $I(W_{\mathcal{B}})>\log L - \delta$, where $W_{\mathcal{B}}: \{1,...,L\}\rightarrow \mathcal{Y}\times \mathcal{S}$ is the channel defined by:
$$W_{\mathcal{B}}(y,C|a)=W(y|f_C(a)).\mathbb{P}_{\mathcal{B}}(C).$$
Note that the value of $I(W_{\mathcal{B}})$ does not depend on the choice of the bijections $(f_C)_{C\in\mathcal{S}}$.
\end{enumerate}
If we also have $\mathbb{P}_e(W_{\mathcal{B}})<\epsilon$, we say that $W$ is $(\delta,\epsilon)$-\emph{easy}.
\end{mydef}
If $W$ is $\delta$-easy for a small $\delta$, then we can reliably transmit information near the symmetric capacity of $W$ using a code of blocklength 1 (hence the easiness; there is no need to use codes of large blocklengths): we choose a random code according to $\mathcal{B}$, we reveal this code to the receiver, and then we transmit information using this code. The rate of this code is equal to $\log L$ which is close to the symmetric capacity $I(W)$. On the other hand, the fact that $I(W_{\mathcal{B}})>\log L - \delta$ means that $W_{\mathcal{B}}$ is almost perfect, which ensures that our simple coding scheme has a low probability of error.

Note that we added (2) to our definition in order to induce a uniform distribution on the input. This is important for the polarization process (see the definition of $W^-$ and $W^+$ in Definition \ref{defdef11}: the distribution of $U_1$ and $U_2$ are assumed to be uniform in $\mathcal{X}$).

\begin{mynot}
An $m$-user multiple access channel (MAC) $W$ with input alphabets $\mathcal{X}_1,\ldots,\mathcal{X}_m$ and output alphabet $\mathcal{Y}$ is denoted by $W:\mathcal{X}_1\times\ldots\times\mathcal{X}_m\longrightarrow\mathcal{Y}$. The transition probabilities of $W$ are denoted by $W(y|x_1,\ldots,x_m)$, where $x_1\in\mathcal{X}_1,\;\ldots,\;x_m\in\mathcal{X}_m$ and $y\in\mathcal{Y}$. The probability of error of the ML decoder of $W$ for uniformly distributed input is denoted by $\mathbb{P}_e(W)$. The symmetric sum-capacity of $W$, denoted $I(W)$, is the mutual information $I(X_1,\ldots,X_m;Y)$, where $X_1,\ldots,X_m,Y$ are jointly distributed as $\mathbb{P}_{X_1,\ldots,X_m,Y}(x_1,\ldots,x_m,y)=\frac{1}{|\mathcal{X}_1|\cdots |\mathcal{X}_m|}W(y|x_1,\ldots,x_m)$ (i.e., $X_1,\ldots,X_m$ are independent and uniform in $\mathcal{X}_1,\ldots,\mathcal{X}_m$ respectively and they are used as input to the MAC $W$ while $Y$ is the output).
\end{mynot}

\begin{mydef}
\label{defeasyMAC}
An $m$-user MAC $W:\mathcal{X}_1\times\ldots\times\mathcal{X}_m\longrightarrow\mathcal{Y}$ is said to be $\delta$-\emph{easy} if there exist $m$ integers $L_1\leq |\mathcal{X}_1|,\ldots,L_m\leq |\mathcal{X}_m|$, and $m$ independent random codes $\mathcal{B}_1,\ldots,\mathcal{B}_m$ taking values in the sets $\mathcal{S}_1=\{C_1\subset \mathcal{X}_1:\; |C_1|=L_1\}$, \ldots, $\mathcal{S}_m=\{C_m\subset \mathcal{X}_m:\; |C_m|=L_m\}$ respectively, which satisfy the following:
\begin{itemize}
\item $|I(W)-\log L|<\delta$, where $L = L_1\cdots L_m$.
\item For every $1\leq i\leq m$ and every $x_i\in\mathcal{X}_i$, we have $\displaystyle\sum_{C_i\in\mathcal{S}_i}\frac{1}{L_i}\mathbb{P}_{\mathcal{B}_i}(C_i)\mathds{1}_{x_i\in C_i}=\frac{1}{|\mathcal{X}_i|}$. In other words, if $C_i\in\mathcal{S}_i$ is chosen according to the distribution of $\mathcal{B}_i$ and $X_i$ is chosen uniformly in $C_i$, then the marginal distribution of $X_i$ as a random variable in $\mathcal{X}_i$ is uniform.
\item If for each $1\leq i\leq m$ and each $C_i\in\mathcal{S}_i$ we fix a bijection $f_{i,C_i}:\{1,...,L_i\}\rightarrow C_i$, then $I(W_{\mathcal{B}_1,\ldots,\mathcal{B}_m})>\log L - \delta$, where $W_{\mathcal{B}_1,\ldots,\mathcal{B}_m}: \{1,...,L_1\}\times\ldots\times \{1,...,L_m\}\rightarrow \mathcal{Y}\times \mathcal{S}_1\times\ldots\times\mathcal{S}_m$ is the MAC defined by:
$$W_{\mathcal{B}_1,\ldots,\mathcal{B}_m}(y,C_1,\ldots,C_m|a_1,\ldots,a_m)=W(y|f_{1,C_1}(a_1),\ldots,f_{m,C_m}(a_m)).\prod_{i=1}^m\mathbb{P}_{\mathcal{B}_i}(C_i).$$
Note that the value of $I(W_{\mathcal{B}_1,\ldots,\mathcal{B}_m})$ does not depend on the choice of the bijections $(f_{i,C_i})_{1\leq i\leq m,\; C_i\in\mathcal{S}_i}$.
\end{itemize}
If we also have $\mathbb{P}_e(W_{\mathcal{B}_1,\ldots,\mathcal{B}_m})<\epsilon$, we say that $W$ is $(\delta,\epsilon)$-\emph{easy}.
\end{mydef}

If $W$ is a $\delta$-easy MAC for a small $\delta$, then we can reliably transmit information near the symmetric sum-capacity of $W$ using a code of blocklength 1 (hence the easiness; there is no need to use codes of large blocklengths): we choose a random MAC-code according to $\mathcal{B}_1,\ldots,\mathcal{B}_m$, we reveal this code to the receiver, and then we transmit information using this code. The sum-rate of this code is equal to $\log L_1+\ldots+\log L_m=\log L$ which is close to the sum-capacity $I(W)$. On the other hand, the fact that $I(W_{\mathcal{B}_1,\ldots,\mathcal{B}_m})>\log L - \delta$ means that $W_{\mathcal{B}_1,\ldots,\mathcal{B}_m}$ is almost perfect, which ensures that our simple coding scheme has a low probability of error.

\subsection{Polarization process}

In this subsection, we consider an ordinary (single user) channel $W$ and a binary operation $\ast$ on its input alphabet.

\begin{mydef}
Let $\mathcal{X}$ be an arbitrary set and $\ast$ be a binary operation on $\mathcal{X}$. Let $W:\mathcal{X}\longrightarrow\mathcal{Y}$ be a channel. We define the two channels $W^-:\mathcal{X}\longrightarrow\mathcal{Y}\times\mathcal{Y}$ and $W^+:\mathcal{X}\longrightarrow\mathcal{Y}\times\mathcal{Y}\times \mathcal{X}$ as follows:
$$W^-(y_1,y_2|u_1)=\frac{1}{|\mathcal{X}|}\sum_{u_2\in \mathcal{X}}W(y_1|u_1\ast u_2)W(y_2|u_2),$$
$$W^+(y_1,y_2,u_1|u_2)=\frac{1}{|\mathcal{X}|}W(y_1|u_1\ast u_2)W(y_2|u_2).$$
For every $s=(s_1,\ldots,s_n)\in\{-,+\}^n$, we define $W^s$ recursively as: $$W^s:=((W^{s_1})^{s_2}\ldots)^{s_n}.$$
\label{defdef11}
\end{mydef}

\begin{mydef}
\label{def1}
Let $(B_n)_{n\geq1}$ be i.i.d. uniform random variables in $\{-,+\}$. For each channel $W$ with input alphabet $\mathcal{X}$, we define the channel-valued process $(W_n)_{n\geq0}$ recursively as follows:
\begin{align*}
W_0 &:= W,\\
W_{n} &:=W_{n-1}^{B_n}\;\forall n\geq1.
\end{align*}
\end{mydef}

\begin{mydef}
\label{defPola}
A binary operation $\ast$ is said to be \emph{polarizing} if we have the following two properties:
\begin{itemize}
\item Conservation property: for every channel $W$ with input alphabet $\mathcal{X}$, we have $I(W^-)+I(W^+)=2I(W)$.
\item Polarization property: for every channel $W$ with input alphabet $\mathcal{X}$ and every $\delta>0$, $W_n$ almost surely becomes $\delta$-easy, i.e.,
$$\displaystyle\lim_{n\rightarrow\infty}\mathbb{P}\big[W_n\;\text{is}\;\delta\text{-}\text{easy}\big]=1.$$
\end{itemize}
\end{mydef}

\begin{mynot}
\label{notnotnot}
Throughout this paper, we will write $(U_1,U_2)\stackrel{f_{\ast}}{\longrightarrow}(X_1,X_2)\stackrel{W}{\longrightarrow}(Y_1,Y_2)$ to denote the following:
\begin{itemize}
\item $U_1$ and $U_2$ are two independent random variables uniformly distributed in $\mathcal{X}$.
\item $X_1=U_1\ast U_2$ and $X_2=U_2$.
\item The conditional distribution $(Y_1,Y_2)|(X_1,X_2)$ is given by: $$\mathbb{P}_{Y_1,Y_2|X_1,X_2}(y_1,y_2|x_1,x_2)=W(y_1|x_1)W(y_2|x_2).$$ I.e., $Y_1$ and $Y_2$ are the outputs of two independent copies of the channel $W$ with inputs $X_1$ and $X_2$ respectively.
\item $(U_1,U_2)-(X_1,X_2)-(Y_1,Y_2)$ is a Markov chain.
\end{itemize}
Note that since $X_1=U_1\ast U_2$ and $X_2=U_2$, the chain $(X_1,X_2)-(U_1,U_2)-(Y_1,Y_2)$ is also a Markov chain.
\end{mynot}

\begin{myrem}
\label{rem1}
Let $(U_1,U_2)\stackrel{f_{\ast}}{\longrightarrow}(X_1,X_2)\stackrel{W}{\longrightarrow}(Y_1,Y_2)$. From the definition of $W^-$ and $W^+$, it is easy to see that we have $I(W^-)=I(U_1;Y_1,Y_2)$ and $I(W^+)=I(U_2;Y_1,Y_2,U_1)$. Therefore,
\begin{align*}
I(W^-)+I(W^+)&=I(U_1;Y_1,Y_2)+I(U_2;Y_1,Y_2,U_1)\\
&=I(U_1,U_2;Y_1,Y_2) \stackrel{(a)}{=}I(X_1,X_2;Y_1,Y_2),
\end{align*}
where (a) follows from the fact that both $(U_1,U_2)-(X_1,X_2)-(Y_1,Y_2)$ and $(X_1,X_2)-(U_1,U_2)-(Y_1,Y_2)$ are Markov chains. We have the following:
\begin{itemize}
\item If $\ast$ is not uniformity preserving, then $(X_1,X_2)$ is not uniform in $\mathcal{X}^2$. If $W$ is a perfect channel, i.e., $I(W)=\log|\mathcal{X}|$, we have 
\begin{equation}
I(W^-)+I(W^+)=I(X_1,X_2;Y_1,Y_2)\leq H(X_1,X_2)\stackrel{(a)}{<}2\log|\mathcal{X}|=2I(W),
\label{equnifabcjafuas}
\end{equation}
where (a) follows from the fact that $(X_1,X_2)$ is not uniform in $\mathcal{X}^2$. \eqref{equnifabcjafuas} means that $\ast$ does not satisfy the conservation property of Definition \ref{defPola}. Therefore, every polarizing operation must be uniformity preserving.
\item If $\ast$ is uniformity preserving, then $(X_1,X_2)$ is uniform in $\mathcal{X}^2$, i.e., $X_1$ and $X_2$ are independent and uniform in $\mathcal{X}$. Thus,
\begin{align*}
I(W^-)+I(W^+)&=I(X_1,X_2;Y_1,Y_2)=I(X_1;Y_1)+I(X_2;Y_2)=2I(W).
\end{align*}
Therefore, uniformity preserving operations satisfy the conservation property. \end{itemize}
We conclude that a binary operation $\ast$ satisfies the conservation property if and only if it is uniformity preserving.

\end{myrem}

\begin{mydef}
Let $\ast$ be a polarizing operation on a set $\mathcal{X}$. We say that $\beta\geq 0$ is a $\ast$-\emph{achievable exponent} if for every $\delta>0$ and every channel $W$ with input alphabet $\mathcal{X}$, $W_n$ almost surely becomes $(\delta,2^{-2^{\beta n}})$-easy, i.e.,
$$\lim_{n\rightarrow\infty}\mathbb{P}\big[W_n\;\text{is}\;(\delta,2^{-2^{\beta n}})\text{-}\text{easy}\big]=1.$$
We define the \emph{exponent} of $\ast$ as:
$$E_\ast:=\sup\{\beta\geq 0:\;\beta\;\text{is\;a}\;\text{$\ast$-achievable\;exponent}\}.$$
\end{mydef}

Note that $E_{\ast}$ depends only on $\ast$ and it does not depend on any particular channel $W$. The definition of a $\ast$-achievable exponent ensures that it is achievable for every channel $W$ with input alphabet $\mathcal{X}$.

\begin{myrem}
If $\ast$ is a polarizing operation of exponent $E_{\ast}>0$ on the set $\mathcal{X}$, then for every channel $W$ with input alphabet $\mathcal{X}$, every $\beta< E_{\ast}$ and every $\delta>0$, there exists $n_0=n_0(W,\beta,\delta,\ast)>0$ such that for every $n\geq n_0$, there exists a polar code of blocklength $N=2^n$ and of rate at least $I(W)-\delta$ such that the probability of error of the successive cancellation decoder is at most $2^{-N^{\beta}}$. (The polar code construction in section V of \cite{RajTelA} can be applied here to get such a code).
\end{myrem}

\begin{myex}
If $\mathcal{X}=\mathbb{F}_2=\{0,1\}$ and $\ast$ is the addition modulo 2, then $E_{\ast}=\frac{1}{2}$ (see \cite{ArikanTelatar}).
\end{myex}

\subsection{Polarization process for MACs}

\begin{mydef}
Let $\mathcal{X}_1,\ldots,\mathcal{X}_m$ be $m$ arbitrary sets. Let $\ast_1,\ldots,\ast_m$ be $m$ binary operations on $\mathcal{X}_1,\ldots,\mathcal{X}_m$ respectively, and let $W:\mathcal{X}_1\times\ldots\times\mathcal{X}_m\longrightarrow\mathcal{Y}$ be an $m$-user MAC. We define the two MACs $W^-:\mathcal{X}_1\times\ldots\times \mathcal{X}_m\longrightarrow\mathcal{Y}\times\mathcal{Y}$ and $W^+:\mathcal{X}_1\times\ldots\times\mathcal{X}_m\longrightarrow \mathcal{Y}\times\mathcal{Y}\times \mathcal{X}_1\times\ldots\times\mathcal{X}_m$ as follows:
\begin{align*}
W^-(y_1,y_2|u_{1,1},\ldots,u_{1,m})=\frac{1}{|\mathcal{X}_1|\cdots |\mathcal{X}_m|}\sum_{\substack{u_{2,1}\in \mathcal{X}_1\\\vdots\\u_{2,m}\in \mathcal{X}_m}}W(y_1|u_{1,1}\ast_1 &u_{2,1},\ldots,u_{1,m}\ast_m u_{2,m})\\
&\times W(y_2|u_{2,1},\ldots,u_{2,m}),
\end{align*}
\begin{align*}
W^+(y_1,y_2,u_{1,1},\ldots,u_{1,m}|u_{2,1},\ldots,u_{2,m})=\frac{1}{|\mathcal{X}_1|\cdots|\mathcal{X}_m|}W(y_1|u_{1,1}\ast_1 &u_{2,1},\ldots,u_{1,m}\ast_m u_{2,m})\\
&\times W(y_2|u_{2,1},\ldots,u_{2,m}).
\end{align*}
For every $s=(s_1,\ldots,s_n)\in\{-,+\}^n$, we define $W^s$ recursively as: $$W^s:=((W^{s_1})^{s_2}\ldots)^{s_n}.$$
\label{defdef11MAC}
\end{mydef}

\begin{mydef}
\label{def1MAC}
Let $(B_n)_{n\geq1}$ be i.i.d. uniform random variables in $\{-,+\}$. For each MAC $W$  with input alphabets $\mathcal{X}_1,\ldots,\mathcal{X}_m$, we define the MAC-valued process $(W_n)_{n\geq0}$ recursively as follows:
\begin{align*}
W_0 &:= W,\\
W_{n} &:=W_{n-1}^{B_n}\;\forall n\geq1.
\end{align*}
\end{mydef}

\begin{mydef}
\label{defPolaMac}
A sequence of $m$ binary operations $(\ast_1,\ldots,\ast_m)$ on the sets $\mathcal{X}_1,\ldots,\mathcal{X}_m$ is said to be \emph{MAC-polarizing} if we have the following two properties:
\begin{itemize}
\item Conservation property: for every MAC $W$ with input alphabets $\mathcal{X}_1,\ldots,\mathcal{X}_m$ we have $$I(W^-)+I(W^+)=2I(W).$$
\item Polarization property: for every MAC $W$ with input alphabets $\mathcal{X}_1,\ldots,\mathcal{X}_m$ and every $\delta>0$, $W_n$ almost surely becomes $\delta$-easy, i.e.,
$$\lim_{n\rightarrow\infty}\mathbb{P}\big[W_n\;\text{is}\;\delta\text{-}\text{easy}\big]=1.$$
\end{itemize}
\end{mydef}

Notice that in the conservation property we only ask for the sum-capacity to be preserved and we do not ask for the whole capacity region to be preserved. The reason for this is because MAC polarization sometimes induces a loss in the capacity region (see \cite{SasogluTelYeh}, \cite{AbbeTelatar} and \cite{RajTelA}). There are, however, polar coding techniques that achieve the whole capacity region (e.g., \cite{ArikanMAC} and \cite{Onay}) but those techniques are not based on MAC polarization; they are based on monotone chain rules and single user channel polarization. In the above definition, we are only interested in the MAC polarization phenomenon itself. We note, however, that monotone chain rules can be used together with the general single user polarization theory that is developed here in order to construct MAC codes that achieve the whole capacity region.

\begin{myrem}
\label{remPreservMac}
As in Remark \ref{rem1}, a sequence of binary operations satisfies the conservation property if and only if every operation in the sequence is uniformity preserving.
\end{myrem}

\begin{mydef}
Let $(\ast_1,\ldots,\ast_m)$ be a MAC-polarizing sequence on the sets $\mathcal{X}_1,\ldots,\mathcal{X}_m$. We say that $\beta\geq 0$ is a $(\ast_1,\ldots,\ast_m)$-\emph{achievable exponent} if for every $\delta>0$ and every MAC $W$ with input alphabets $\mathcal{X}_1,\ldots,\mathcal{X}_m$, $W_n$ almost surely becomes $(\delta,2^{-2^{\beta n}})$-easy, i.e.,
$$\lim_{n\rightarrow\infty}\mathbb{P}\big[W_n\;\text{is}\;(\delta,2^{-2^{\beta n}})\text{-}\text{easy}\big]=1.$$
We define the \emph{exponent} of $(\ast_1,\ldots,\ast_m)$ as:
$$E_{\ast_1,\ldots,\ast_m}:=\sup\{\beta\geq 0:\;\beta\;\text{is\;a}\;\text{$(\ast_1,\ldots,\ast_m)$-achievable\;exponent}\}.$$
\end{mydef}

\begin{myrem}
If $(\ast_1,\ldots,\ast_m)$ is a MAC-polarizing sequence of exponent $E_{\ast_1,\ldots,\ast_m}>0$ on the sets $\mathcal{X}_1,\ldots,\mathcal{X}_m$, then for every MAC $W$ with input alphabets $\mathcal{X}_1,\ldots,\mathcal{X}_m$, every $\beta< E_{\ast_1,\ldots,\ast_m}$ and every $\delta>0$, there exists $n_0=n_0(W,\beta,\delta,\ast)>0$ such that for every $n\geq n_0$, there exists a polar code of blocklength $N=2^n$ and of sum-rate at least $I(W)-\delta$ such that the probability of error of the successive cancellation decoder is at most $2^{-N^{\beta}}$.
\end{myrem}

\begin{myrem}
\label{remnecmac}
For each $1\leq i\leq m$ and each ordinary single user channel $W_i:\mathcal{X}_i\longrightarrow\mathcal{Y}$ with input alphabet $\mathcal{X}_i$, consider the MAC $W:\mathcal{X}_1\times\ldots\times\mathcal{X}_m\longrightarrow\mathcal{Y}$ defined as $W(y|x_1,\ldots,x_m)=W_i(y|x_i)$. Let $(W_{i,n})_{n\geq 0}$ be the single user channel valued process obtained from $W_i$ as in Definition \ref{def1}, and let $(W_n)_{n\geq 0}$ be the MAC-valued process obtained from $W$ as in Definition \ref{def1MAC}. It is easy to see that $W_{i,n}$ is $\delta$-easy if and only if $W_n$ is $\delta$-easy. This shows that if the sequence $(\ast_1,\ldots,\ast_m)$ is MAC-polarizing then $\ast_i$ is polarizing for each $1\leq i\leq m$. Moreover, $W_{i,n}$ is $(\delta,\epsilon)$-easy if and only if $W_n$ is $(\delta,\epsilon)$-easy. This implies that $E_{\ast_1,\ldots,\ast_m}\leq E_{\ast_i}$ for each $1\leq i\leq m$. Therefore, $E_{\ast_1,\ldots,\ast_m}\leq \min\{E_{\ast_1},\ldots,E_{\ast_m}\}$.
\end{myrem}

\section{Polarizing operations}

\subsection{Necessary condition}
In this subsection, we show that if $\ast$ is polarizing, then $\ast$ is uniformity preserving and $/^{\ast}$ (the right-inverse of $\ast$) is strongly ergodic. In order to prove this, we need the following two lemmas:

\begin{mylem}
\label{lemBijectNec}
Let $\ast$ be an ergodic operation on a set $\mathcal{X}$. Let $\mathcal{H}$ be a stable partition of $\mathcal{X}$ such that $\mathcal{K}_{\mathcal{H}}\neq\mathcal{H}$, where $\mathcal{K}_{\mathcal{H}}$ is the first residue of $\mathcal{H}$ with respect to $\ast$. Define $\mathcal{A}=\mathcal{H}\cup\mathcal{K}_{\mathcal{H}}$. We have:
\begin{enumerate}
\item For every $A_1,A_2\in\mathcal{A}$, we have:
\begin{itemize}
\item $(A_1\in\mathcal{K}_{\mathcal{H}}\text{ and }A_2\in \mathcal{K}_{\mathcal{H}})\text{ if and only if }(A_1\ast A_2\in {\mathcal{K}_{\mathcal{H}}}^\ast\text{ and }A_2\in \mathcal{K}_{\mathcal{H}})$.
\item $(A_1\in\mathcal{K}_{\mathcal{H}}\text{ and }A_2\in \mathcal{H})\text{ if and only if }(A_1\ast A_2\in {\mathcal{K}_{\mathcal{H}}}^\ast\text{ and }A_2\in \mathcal{H})$.
\item $(A_1\in\mathcal{H}\text{ and }A_2\in \mathcal{K}_{\mathcal{H}})\text{ if and only if }(A_1\ast A_2\in\mathcal{H}^\ast\text{ and }A_2\in \mathcal{K}_{\mathcal{H}})$.
\item $(A_1\in\mathcal{H}\text{ and }A_2\in \mathcal{H})\text{ if and only if }(A_1\ast A_2\in \mathcal{H}^\ast\text{ and }A_2\in \mathcal{H})$.
\end{itemize}
\item For every $u_1,u_2\in \mathcal{X}$ and every $A_1,A_2\in\mathcal{A}$, we have $$(u_1\in A_1\ast A_2\text{ and }u_2\in A_2)\text{ if and only if }(u_1/^{\ast}u_2\in A_1\text{ and }u_2\in A_2).$$
\end{enumerate}
\end{mylem}
\begin{proof}
1) We have $\mathcal{A}=\mathcal{H}\cup\mathcal{K}_\mathcal{H}$. Therefore, for every $A_1,A_2\in\mathcal{A}$, one of the following four conditions holds true:
\begin{itemize}
\item[(i)] $A_1\in\mathcal{K}_{\mathcal{H}}\text{ and }A_2\in \mathcal{K}_{\mathcal{H}}$.
\item[(ii)] $A_1\in\mathcal{K}_{\mathcal{H}}\text{ and }A_2\in \mathcal{H}$.
\item[(iii)] $A_1\in\mathcal{H}\text{ and }A_2\in \mathcal{K}_{\mathcal{H}}$.
\item[(iv)] $A_1\in\mathcal{H}\text{ and }A_2\in \mathcal{H}$.
\end{itemize}

Now since $\mathcal{K}_{\mathcal{H}}\neq\mathcal{H}$ and $\mathcal{K}_{\mathcal{H}}\preceq\mathcal{H}$, we have $||\mathcal{K}_{\mathcal{H}}||<||\mathcal{H}||$. Therefore, for every $K\in\mathcal{K}_{\mathcal{H}}$ and every $H\in\mathcal{H}$, we have $|K|=||\mathcal{K}_{\mathcal{H}}||<||\mathcal{H}||=|H|$. This implies that $K\neq H$ for every $K\in\mathcal{K}_{\mathcal{H}}$ and every $H\in\mathcal{H}$, hence $\mathcal{K}_{\mathcal{H}}\cap \mathcal{H}=\o$. Similarly, ${\mathcal{K}_{\mathcal{H}}}^\ast\cap {\mathcal{H}}^\ast=\o$. We conclude that for every $A_1,A_2\in\mathcal{A}$, the following four conditions are mutually exclusive:
\begin{itemize}
\item[(a)] $A_1\ast A_2\in {\mathcal{K}_{\mathcal{H}}}^\ast\text{ and }A_2\in \mathcal{K}_{\mathcal{H}}$.
\item[(b)] $A_1\ast A_2\in {\mathcal{K}_{\mathcal{H}}}^\ast\text{ and }A_2\in \mathcal{H}$.
\item[(c)] $A_1\ast A_2\in\mathcal{H}^\ast\text{ and }A_2\in \mathcal{K}_{\mathcal{H}}$.
\item[(d)] $A_1\ast A_2\in \mathcal{H}^\ast\text{ and }A_2\in \mathcal{H}$.
\end{itemize}

We have:
\begin{itemize}
\item If $A_1\in\mathcal{K}_{\mathcal{H}}$ and $A_2\in\mathcal{K}_{\mathcal{H}}$, then $A_1\ast A_2\in{\mathcal{K}_{\mathcal{H}}}^\ast$. Therefore, (i) implies (a).
\item If $A_1\in\mathcal{K}_{\mathcal{H}}$ and $A_2\in\mathcal{H}$, then $A_1\ast A_2\in{\mathcal{K}_{\mathcal{H}}}^\ast$ (see Theorem \ref{theres} of Part I \cite{RajErgI}). Therefore, (ii) implies (b).
\item If $A_1\in\mathcal{H}$ and $A_2\in\mathcal{K}_{\mathcal{H}}$, let $H\in\mathcal{H}$ be such that $A_2\subset H$. (Note that there is no contradiction here between $A_2\subset H\in\mathcal{H}$, $A_2\in\mathcal{K}_{\mathcal{H}}$ and $\mathcal{H}\cap\mathcal{K}_{\mathcal{H}}=\o$.) We have $A_1\ast A_2\subset A_1\ast H$ and $|A_1\ast A_2|\geq |A_1|=||\mathcal{H}||=||\mathcal{H}^{\ast}||=|A_1\ast H|$. Therefore, $A_1\ast A_2=A_1\ast H\in\mathcal{H}^{\ast}$. Hence (iii) implies (c).
\item If $A_1\in\mathcal{H}$ and $A_2\in\mathcal{H}$, then $A_1\ast A_2\in\mathcal{H}^{\ast}$. Therefore, (iv) implies (d).
\end{itemize}

Now let $A_1,A_2\in\mathcal{A}$ and suppose that (a) holds true (i.e., $A_1\ast A_2\in {\mathcal{K}_{\mathcal{H}}}^\ast\text{ and }A_2\in \mathcal{K}_{\mathcal{H}}$). Since $A_1\in\mathcal{A}$ then either $A_1\in\mathcal{K}_\mathcal{H}$ or $A_1\in\mathcal{H}$. But $A_2\in \mathcal{K}_{\mathcal{H}}$, so either (i) or (iii) holds true. On the other hand, we have shown that (iii) implies (c), and (c) contradicts (a), so (iii) cannot be true. Therefore, (i) must be true. We conclude that (a) implies (i). Similarly, we can show that (b) implies (ii), (c) implies (iii), and (d) implies (iv).

\vspace*{2mm}
2) Fix $A_1,A_2\in\mathcal{A}$. We have:
\begin{itemize}
\item If $A_1\in\mathcal{K}_{\mathcal{H}}\text{ and }A_2\in \mathcal{K}_{\mathcal{H}}$, then $|A_1\ast A_2|=||{\mathcal{K}_{\mathcal{H}}}^\ast||=||\mathcal{K}_{\mathcal{H}}||=|A_1|$.
\item If $A_1\in\mathcal{K}_{\mathcal{H}}\text{ and }A_2\in \mathcal{H}$, then from 1) we have $A_1\ast A_2\in {\mathcal{K}_{\mathcal{H}}}^\ast$. Therefore, $|A_1\ast A_2|=||{\mathcal{K}_{\mathcal{H}}}^\ast||=||\mathcal{K}_{\mathcal{H}}||=|A_1|$.
\item If $A_1\in\mathcal{H}\text{ and }A_2\in \mathcal{K}_{\mathcal{H}}$ then from 1) we have $A_1\ast A_2\in\mathcal{H}^\ast$. Therefore, $|A_1\ast A_2|=||\mathcal{H}^\ast||=||\mathcal{H}||=|A_1|$.
\item If $A_1\in\mathcal{H}\text{ and }A_2\in \mathcal{H}$, then $|A_1\ast A_2|=||\mathcal{H}^\ast||=||\mathcal{H}||=|A_1|$.
\end{itemize}
We conclude that in all cases, we have $|A_1\ast A_2|=|A_1|$.

For every $u_1,u_2\in \mathcal{X}$, we have:
\begin{itemize}
\item If $u_1/^{\ast}u_2\in A_1$ and $u_2\in A_2$, then $u_1=(u_1/^{\ast}u_2)\ast u_2\in A_1\ast A_2$.
\item If $u_1\in A_1\ast A_2$ and $u_2\in A_2$, we have $A_1\ast u_2\subset A_1\ast A_2$. On the other hand, we have $|A_1\ast A_2|=|A_1|=|A_1\ast u_2|$ (where the last equality holds true because $\ast$ is uniformity preserving). We conclude that $A_1\ast A_2=A_1\ast u_2$. Therefore, $(A_1\ast A_2)/^{\ast} u_2=A_1$ which implies that $u_1/^{\ast}u_2\in A_1$.
\end{itemize}
\end{proof}

\begin{mydef}
A channel $W:\mathcal{X}\longrightarrow\mathcal{Y}$ is said to be degraded from a channel $V:\mathcal{X}\longrightarrow\mathcal{Z}$ if there exists a channel $P:\mathcal{Z}\longrightarrow\mathcal{Y}$ such that for every $x\in\mathcal{X}$ and every $y\in\mathcal{Y}$ we have:
$$W(y|x)=\sum_{z\in\mathcal{Z}}V(z|x)P(y|z).$$
If $W$ is degraded from $V$ and $V$ is degraded from $W$, we say that $W$ is equivalent to $V$.
\end{mydef}

\begin{mylem}
Let $\ast$ be a uniformity preserving operation on a set $\mathcal{X}$, and let $W:\mathcal{X}\longrightarrow\mathcal{Y}$. If $I(W^-)=I(W)$ then $W^+$ is equivalent to $W$.
\label{PlusEquiv}
\end{mylem}
\begin{proof}
Since $I(W^+)+I(W^-)=2I(W)$ and since $I(W^-)=I(W)$, we have $I(W^+)=I(W)$. Let $(U_1,U_2)\stackrel{f_{\ast}}{\longrightarrow}(X_1,X_2)\stackrel{W}{\longrightarrow}(Y_1,Y_2)$ (See Notation \ref{notnotnot}). We have:
$$I(W)=I(W^+)=I(U_2;Y_1,Y_2,U_1)=I(U_2;Y_2)+I(U_2;Y_1,U_1|Y_2)=I(W)+I(U_2;Y_1,U_1|Y_2).$$

This shows that $I(U_2;Y_1,U_1|Y_2)=0$. This means that $Y_2$ is a sufficient statistic for the channel $U_2\longrightarrow(Y_1,Y_2,U_1)$ (which is equivalent to $W^+$). We conclude that $W^+$ is equivalent to the channel $U_2\longrightarrow Y_2$, which is equivalent to $W$.
\end{proof}

\begin{myprop}
\label{propnec}
Let $\ast$ be a binary operation on a set $\mathcal{X}$. If $\ast$ is polarizing then $\ast$ is uniformity preserving and $/^{\ast}$ is strongly ergodic.
\end{myprop}
\begin{proof}
If $\ast$ is polarizing then $\ast$ must be uniformity preserving (see Remark \ref{rem1}).

We first prove that $\ast$ is irreducible. Suppose to the contrary that $\ast$ is not irreducible. Proposition \ref{lemerg} of Part I \cite{RajErgI} shows that there exist two disjoint non-empty subsets $A_1$ and $A_2$ of $\mathcal{X}$ such that $A_1\cup A_2=\mathcal{X}$, $A_1\ast\mathcal{X}=A_1$ and $A_2\ast\mathcal{X}=A_2$. This means that for every $u_1,u_2\in\mathcal{X}$ and every $y\in\{1,2\}$, we have $u_1\in A_y$ if and only if $u_1\ast u_2\in A_y$.

For each $\epsilon>0$ define the channel $W_{\epsilon}:\mathcal{X}\longrightarrow \{1,2,e\}$ as follows:
\begin{equation*}
W_{\epsilon}(y|x)=
\begin{cases}
1-\epsilon\;&\text{if}\;y\in\{1,2\}\;\text{and}\;x\in A_y,\\
0\;&\text{if}\;y\in\{1,2\}\;\text{and}\;x\notin A_y,\\
\epsilon\;&\text{if}\; y=e.
\end{cases}
\end{equation*}
$I(W_{\epsilon})=(1-\epsilon)h_2\big(\frac{|A_1|}{|\mathcal{X}|}\big)$, so there exists $\epsilon'>0$ such that $I(W_{\epsilon'})$ is not the logarithm of any integer. For such $\epsilon'$, there exists $\delta>0$ such that $W_{\epsilon'}$ is not $\delta$-easy.

Let $(U_1,U_2)\stackrel{f_{\ast}}{\longrightarrow}(X_1,X_2)\stackrel{W_{\epsilon'}}{\longrightarrow}(Y_1,Y_2)$ (See Notation \ref{notnotnot}). Consider the channel $U_1\longrightarrow (Y_1,Y_2)$ which is equivalent to $W_{\epsilon'}^-$. We have:
\begin{equation}
\begin{aligned}
\mathbb{P}_{Y_1,Y_2|U_1}(y_1,y_2|u_1)&=\frac{1}{|\mathcal{X}|}\sum_{u_2\in\mathcal{X}} W_{\epsilon'}(y_1|u_1\ast u_2)W_{\epsilon'}(y_2|u_2)\stackrel{(a)}{=}\frac{1}{|\mathcal{X}|}\sum_{u_2\in\mathcal{X}} W_{\epsilon'}(y_1|u_1)W_{\epsilon'}(y_2|u_2)\\
&\stackrel{(b)}{=}\sum_{u_2\in\mathcal{X}} W_{\epsilon'}(y_1|u_1)\mathbb{P}_{Y_2|U_2}(y_2|u_2)\mathbb{P}_{U_2}(u_2)= W_{\epsilon'}(y_1|u_1)\mathbb{P}_{Y_2}(y_2),
\end{aligned}
\label{sndchbwuy8723cty}
\end{equation}
where (a) follows from the fact that if $y_1=e$ then $W_{\epsilon'}(y_1|u_1\ast u_2)=W_{\epsilon'}(y_1|u_1)=\epsilon'$ and if $y_1\in\{1,2\}$ then $u_1\in A_{y_1}$ if and only if $u_1\ast u_2\in A_{y_1}$, which implies that $W_{\epsilon'}(y_1|u_1\ast u_2)=W_{\epsilon'}(y_1|u_1)$. (b) follows from the fact that the channel $U_2\longrightarrow Y_2$ is equivalent to $W_{\epsilon'}$ and the fact that $U_2$ is uniform in $\mathcal{X}$. 

\eqref{sndchbwuy8723cty} implies that $Y_1$ is a sufficient statistic for the channel $U_1\longrightarrow (Y_1,Y_2)$ (which is equivalent to $W_{\epsilon'}^-$). Moreover, since $\mathbb{P}_{Y_1,Y_2|U_1}(y_1,y_2|u_1)=W_{\epsilon'}(y_1|u_1)\mathbb{P}_{Y_2}(y_2)$,  we conclude that the channel $W_{\epsilon'}^-$ is equivalent to $W_{\epsilon'}$. This implies that $I(W_{\epsilon'}^-)=I(W_{\epsilon'})$. Now Lemma \ref{PlusEquiv} implies that $W_{\epsilon'}^+$ is equivalent to $W_{\epsilon'}$. Therefore, for any $l>0$ and any $s\in\{-,+\}^l$, $W_{\epsilon'}^s$ is equivalent to $W_{\epsilon'}$ which is not $\delta$-easy. This contradicts the fact that $\ast$ is polarizing. We conclude that $\ast$ must be irreducible.

Suppose that $\ast$ is not ergodic. Proposition \ref{lemerg} of Part I \cite{RajErgI} shows that there exists a partition $\{H_0,\ldots,H_{n-1}\}$ of $\mathcal{X}$ such that $H_i\ast \mathcal{X}=H_{i+1\bmod n}$ for all $0\leq i<n$ and $|H_0|=\ldots=|H_{n-1}|$. This means that for every $u_1,u_2\in\mathcal{X}$ and every $y\in\{0,\ldots,n-1\}$, we have $u_1\ast u_2\in H_y$ if and only if $u_1\in H_{y-1\bmod n}$ .

For each $0\leq i<n$ and each $0<\epsilon<1$, define the channel $W_{i,\epsilon}:\mathcal{X}\longrightarrow \{0,\ldots,n-1,e\}$ as follows:
\begin{equation*}
W_{i,\epsilon}(y|x)=
\begin{cases}
1-\epsilon\;&\text{if}\;y\in\{0,\ldots,n-1\}\;\text{and}\;x\in H_{y+i\bmod n},\\
0\;&\text{if}\;y\in\{0,\ldots,n-1\}\;\text{and}\;x\notin H_{y+i\bmod n},\\
\epsilon\;&\text{if}\;y=e.
\end{cases}
\end{equation*}
$I(W_{i,\epsilon})=(1-\epsilon)\log n$ so there exists $\epsilon'>0$ such that $I(W_{i,\epsilon'})$ is not the logarithm of any integer. For such $\epsilon'$, there exists $\delta>0$ such that $W_{i,\epsilon'}$ is not $\delta$-easy for any $0\leq i< n$.

Let $(U_1,U_2)\stackrel{f_{\ast}}{\longrightarrow}(X_1,X_2)\stackrel{W_{i,\epsilon'}}{\longrightarrow}(Y_1,Y_2)$. Consider the channel $U_1\longrightarrow (Y_1,Y_2)$ which is equivalent to $W_{i,\epsilon'}^-$. We have:
\begin{equation}
\begin{aligned}
\mathbb{P}_{Y_1,Y_2|U_1}(y_1,y_2|u_1)&=\frac{1}{|\mathcal{X}|}\sum_{u_2\in\mathcal{X}} W_{i,\epsilon'}(y_1|u_1\ast u_2)W_{i,\epsilon'}(y_2|u_2)\stackrel{(a)}{=}\frac{1}{|\mathcal{X}|}\sum_{u_2\in\mathcal{X}} W_{i-1\bmod n,\epsilon'}(y_1|u_1)W_{i,\epsilon'}(y_2|u_2)\\
&\stackrel{(b)}{=}\sum_{u_2\in\mathcal{X}} W_{i-1\bmod n,\epsilon'}(y_1|u_1)\mathbb{P}_{Y_2|U_2}(y_2|u_2)\mathbb{P}_{U_2}(u_2)= W_{i-1\bmod n,\epsilon'}(y_1|u_1)\mathbb{P}_{Y_2}(y_2),
\end{aligned}
\label{sjsebfjhzxc83y2rwy}
\end{equation}
where (a) follows from the fact that if $y_1=e$ then $W_{i,\epsilon'}(y_1|u_1\ast u_2)=W_{i-1\bmod n,\epsilon'}(y_1|u_1)=\epsilon'$ and if $y_1\in\{0,\ldots,n-1\}$ then $u_1\ast u_2\in H_{y_1+i\bmod n}$ if and only if $u_1\in H_{y_1+i-1\bmod n}$ (which implies that $W_{i,\epsilon'}(y_1|u_1\ast u_2)=W_{i-1\bmod n,\epsilon'}(y_1|u_1)$). (b) follows from the fact that the channel $U_2\longrightarrow Y_2$ is equivalent to $W_{i,\epsilon'}$ and the fact that $U_2$ is uniform in $\mathcal{X}$.

\eqref{sjsebfjhzxc83y2rwy} implies that $Y_1$ is a sufficient statistic for the channel $U_1\longrightarrow (Y_1,Y_2)$ (which is equivalent to $W_{i,\epsilon'}^-$). Moreover, since $\mathbb{P}_{Y_1,Y_2|U_1}(y_1,y_2|u_1)=W_{i-1\bmod n,\epsilon'}(y_1|u_1)\mathbb{P}_{Y_2}(y_2)$,  we conclude that the channel $W_{i,\epsilon'}^-$ is equivalent to $W_{i-1\bmod n,\epsilon'}$. This implies that $I(W_{i,\epsilon'}^-)=I(W_{i-1\bmod n,\epsilon'})=(1-\epsilon')\log n=I(W_{i,\epsilon'})$. Now Lemma \ref{PlusEquiv} implies that $W_{i,\epsilon'}^+$ is equivalent to $W_{i,\epsilon'}$. Therefore, for any $l>0$ and any $s\in\{-,+\}^l$, $W_{i,\epsilon'}^s$ is equivalent to $W_{i-|s|^-\bmod n,\epsilon'}$ (where $|s|^-$ is the number of appearances of the $-$ sign in the sequence $s$) which is not $\delta$-easy. This contradicts the fact that $\ast$ is polarizing. We conclude that $\ast$ must be ergodic.

Since $\ast$ is ergodic, $/^{\ast}$ is ergodic as well. Suppose that $/^{\ast}$ is not strongly ergodic. Theorem \ref{thestrong} of Part I \cite{RajErgI} implies the existence of a stable partition $\mathcal{H}$ of $(\mathcal{X},/^{\ast})$ such that $\mathcal{K}_{\mathcal{H}}\neq \mathcal{H}$ (where $\mathcal{K}_{\mathcal{H}}$ here denotes the first residue of $\mathcal{H}$ with respect to the right-inverse operation $/^{\ast}$). For each $i\geq 0$ and each $\epsilon>0$ define the channel $W_{i,\epsilon}:\mathcal{X}\longrightarrow{\mathcal{K}_{\mathcal{H}}}^{i/^{\ast}}\cup \mathcal{H}^{i/^{\ast}}$ as follows:
\begin{equation*}
W_{i,\epsilon}(y|x)=
\begin{cases}
1-\epsilon\;&\text{if}\;x\in y\;\text{and}\;y\in{\mathcal{K}_{\mathcal{H}}}^{i/^{\ast}},\\
\epsilon\;&\text{if}\;x\in y\;\text{and}\;y\in\mathcal{H}^{i/^{\ast}},\\
0\;&\text{if}\;x\notin y.
\end{cases}
\end{equation*}
We emphasize that $y$ here is a subset of $\mathcal{X}$ and it is not an element of it. We have $$I(W_{i,\epsilon})=(1-\epsilon)\log |{\mathcal{K}_{\mathcal{H}}}^{i/^{\ast}}|+\epsilon\log|\mathcal{H}^{i/^{\ast}}|=(1-\epsilon)\log |\mathcal{K}_{\mathcal{H}}|+\epsilon\log|\mathcal{H}|.$$ Now since $\mathcal{K}_{\mathcal{H}}\neq \mathcal{H}$ and $\mathcal{K}_{\mathcal{H}}\preceq\mathcal{H}$, we have $|\mathcal{H}|\neq |\mathcal{K}_{\mathcal{H}}|$. Therefore, there exists $\epsilon'>0$ such that $I(W_{i,\epsilon'})$ is not the logarithm of any integer. For such $\epsilon'>0$, there exists $\delta>0$ such that $I(W_{i,\epsilon'})$ is not $\delta$-easy for any $i\geq 0$.

Let $(U_1,U_2)\stackrel{f_{\ast}}{\longrightarrow}(X_1,X_2)\stackrel{W_{i,\epsilon'}}{\longrightarrow}(Y_1,Y_2)$. Consider the channel $U_1\longrightarrow (Y_1,Y_2)$, which is equivalent to $W_{i,\epsilon'}^-$. We have:
\begin{equation}
\begin{aligned}
&\mathbb{P}_{Y_1,Y_2|U_1}(y_1,y_2|u_1)\\
&=\frac{1}{|\mathcal{X}|}\sum_{u_2\in\mathcal{X}} W_{i,\epsilon'}(y_1|u_1\ast u_2)W_{i,\epsilon'}(y_2|u_2)\\
&=\frac{1}{|\mathcal{X}|}\sum_{u_2\in\mathcal{X}} \left[\mathds{1}_{u_1\ast u_2\in y_1}\cdot\left((1-\epsilon')\mathds{1}_{y_1\in{\mathcal{K}_{\mathcal{H}}}^{i/^{\ast}}}+\epsilon' \mathds{1}_{y_1\in\mathcal{H}^{i/^{\ast}}} \right)\right]\left[\mathds{1}_{u_2\in y_2}\cdot\left((1-\epsilon')\mathds{1}_{y_2\in{\mathcal{K}_{\mathcal{H}}}^{i/^{\ast}}}+\epsilon' \mathds{1}_{y_2\in\mathcal{H}^{i/^{\ast}}} \right)\right]\\
&=\frac{1}{|\mathcal{X}|}\sum_{u_2\in\mathcal{X}} \mathds{1}_{u_1\ast u_2\in y_1,\;u_2\in y_2}\cdot\left((1-\epsilon')\mathds{1}_{y_1\in{\mathcal{K}_{\mathcal{H}}}^{i/^{\ast}}}+\epsilon' \mathds{1}_{y_1\in\mathcal{H}^{i/^{\ast}}} \right)\left((1-\epsilon')\mathds{1}_{y_2\in{\mathcal{K}_{\mathcal{H}}}^{i/^{\ast}}}+\epsilon' \mathds{1}_{y_2\in\mathcal{H}^{i/^{\ast}}} \right)\\
&\stackrel{(a)}{=}\frac{1}{|\mathcal{X}|}\sum_{u_2\in\mathcal{X}} \mathds{1}_{u_1\in y_1/^{\ast}y_2,\;u_2\in y_2}\cdot\left((1-\epsilon')\mathds{1}_{y_1\in{\mathcal{K}_{\mathcal{H}}}^{i/^{\ast}}}+\epsilon' \mathds{1}_{y_1\in\mathcal{H}^{i/^{\ast}}} \right)\left((1-\epsilon')\mathds{1}_{y_2\in{\mathcal{K}_{\mathcal{H}}}^{i/^{\ast}}}+\epsilon' \mathds{1}_{y_2\in\mathcal{H}^{i/^{\ast}}} \right)\\
&=\frac{1}{|\mathcal{X}|}\sum_{u_2\in\mathcal{X}} \mathds{1}_{u_1\in y_1/^{\ast}y_2,\;u_2\in y_2}
\cdot\Big((1-\epsilon')^2\mathds{1}_{y_1\in{\mathcal{K}_{\mathcal{H}}}^{i/^{\ast}},\;y_2\in{\mathcal{K}_{\mathcal{H}}}^{i/^{\ast}}}+(1-\epsilon')\epsilon' \mathds{1}_{y_1\in{\mathcal{K}_{\mathcal{H}}}^{i/^{\ast}},\;y_2\in\mathcal{H}^{i/^{\ast}}}\\
&\;\;\;\;\;\;\;\;\;\;\;\;\;\;\;\;\;\;\;\;\;\;\;\;\;\;\;\;\;\;\;\;\;\;\;\;\;\;\;\;\;\;\;\;\;\;\;\;\;\;\;\;\;\;\;\;\;\;\;\;\;\;\;\;\;\;\;\;\;\;\;\;\;\;+\epsilon'(1-\epsilon') \mathds{1}_{y_1\in\mathcal{H}^{i/^{\ast}},\;y_2\in{\mathcal{K}_{\mathcal{H}}}^{i/^{\ast}}}+\epsilon'^2 \mathds{1}_{y_1\in\mathcal{H}^{i/^{\ast}},\;y_2\in\mathcal{H}^{i/^{\ast}}} \Big)\\
&\stackrel{(b)}{=}\frac{1}{|\mathcal{X}|}\sum_{u_2\in\mathcal{X}} \mathds{1}_{u_1\in y_1/^{\ast}y_2,\;u_2\in y_2}\cdot
\Big((1-\epsilon')^2\mathds{1}_{y_1/^{\ast}y_2\in{\mathcal{K}_{\mathcal{H}}}^{(i+1)/^{\ast}},\;y_2\in{\mathcal{K}_{\mathcal{H}}}^{i/^{\ast}}}+(1-\epsilon')\epsilon' \mathds{1}_{y_1/^{\ast}y_2\in{\mathcal{K}_{\mathcal{H}}}^{(i+1)/^{\ast}},\;y_2\in\mathcal{H}^{i/^{\ast}}}\\
&\;\;\;\;\;\;\;\;\;\;\;\;\;\;\;\;\;\;\;\;\;\;\;\;\;\;\;\;\;\;\;\;\;\;\;\;\;\;\;\;\;\;\;\;\;\;\;\;\;\;\;\;\;\;\;+\epsilon'(1-\epsilon') \mathds{1}_{y_1/^{\ast}y_2\in\mathcal{H}^{(i+1)/^{\ast}},\;y_2\in{\mathcal{K}_{\mathcal{H}}}^{i/^{\ast}}}+\epsilon'^2 \mathds{1}_{y_1/^{\ast}y_2\in\mathcal{H}^{(i+1)/^{\ast}},\;y_2\in\mathcal{H}^{i/^{\ast}}} \Big)\\
&=\frac{1}{|\mathcal{X}|}\sum_{u_2\in\mathcal{X}} \left[\mathds{1}_{u_1\in y_1/^{\ast}y_2}\cdot\left((1-\epsilon')\mathds{1}_{y_1/^{\ast}y_2\in{\mathcal{K}_{\mathcal{H}}}^{(i+1)/^{\ast}}}+\epsilon' \mathds{1}_{y_1/^{\ast}y_2\in\mathcal{H}^{(i+1)/^{\ast}}} \right)\right]\\
&\;\;\;\;\;\;\;\;\;\;\;\;\;\;\;\;\;\;\;\;\;\;\;\;\;\;\;\;\;\;\;\;\;\;\;\;\;\;\;\;\;\;\;\;\;\;\;\;\;\;\;\;\;\;\;\;\;\;\;\;\;\;\;\;\;\;\;\;\;\;\;\;\;\;\times\left[\mathds{1}_{u_2\in y_2}\cdot\left((1-\epsilon')\mathds{1}_{y_2\in{\mathcal{K}_{\mathcal{H}}}^{i/^{\ast}}}+\epsilon' \mathds{1}_{y_2\in\mathcal{H}^{i/^{\ast}}} \right)\right]\\
&=\frac{1}{|\mathcal{X}|}\sum_{u_2\in\mathcal{X}} W_{i+1,\epsilon'}(y_1/^{\ast}y_2|u_1)W_{i,\epsilon'}(y_2|u_2)\stackrel{(c)}{=}\sum_{u_2\in\mathcal{X}} W_{i+1,\epsilon'}(y_1/^{\ast}y_2|u_1) \mathbb{P}_{Y_2|U_2}(y_2|u_2)\mathbb{P}_{U_2}(u_2)\\
&=W_{i+1,\epsilon'}(y_1/^{\ast}y_2|u_1) \mathbb{P}_{Y_2}(y_2),
\end{aligned}
\label{jdsfsjdbcw74ycr83yc}
\end{equation}
where (a) follows from applying the second point of Lemma \ref{lemBijectNec} on the ergodic operation $/^{\ast}$ and the stable partition $\mathcal{H}^{i/^{\ast}}$. (b) follows from applying the first point of Lemma \ref{lemBijectNec} on the ergodic operation $/^{\ast}$ and the stable partition $\mathcal{H}^{i/^{\ast}}$. (c) follows from the fact that $W_{i,\epsilon'}$ is equivalent to the channel $U_2\longrightarrow Y_2$ and from the fact that $U_2$ is uniform in $\mathcal{X}$.

\eqref{jdsfsjdbcw74ycr83yc} implies that $Y_1/^{\ast}Y_2$ is a sufficient statistic for the channel $U_1\longrightarrow (Y_1,Y_2)$ (which is equivalent to $W_{i,\epsilon'}^-$). Moreover, since $\mathbb{P}_{Y_1,Y_2|U_1}(y_1,y_2|u_1)=W_{i+1,\epsilon'}(y_1/^{\ast}y_2|u_1) \mathbb{P}_{Y_2}(y_2)$, we conclude that the channel $W_{i,\epsilon'}^-$ is equivalent to $W_{i+1,\epsilon'}$. This implies that $I(W_{i,\epsilon'}^-)=I(W_{i+1,\epsilon'})=(1-\epsilon')\log |\mathcal{K}_{\mathcal{H}}|+\epsilon'\log|\mathcal{H}|=I(W_{i,\epsilon'})$. Now Lemma \ref{PlusEquiv} implies that $W_{i,\epsilon'}^+$ is equivalent to $W_{i,\epsilon'}$. Therefore, for any $l>0$ and any $s\in\{-,+\}^l$, $W_{i,\epsilon'}^s$ is equivalent to $W_{i+|s|^-,\epsilon'}$ (where $|s|^-$ is the number of appearances of the $-$ sign in the sequence $s$) which is not $\delta$-easy. This again contradicts the fact that $\ast$ is polarizing. We conclude that $/^{\ast}$ must be strongly ergodic.
\end{proof}

\subsection{Sufficient condition}

In this subsection, we prove a converse for Proposition \ref{propnec}. We will show that for any uniformity preserving operation $\ast$, the strong ergodicity of $/^{\ast}$ implies that $\ast$ is polarizing. We will prove this in three steps.

\subsubsection*{Step 1: Polarized channels are projection channels onto stable partitions}

\begin{mynot}
For every sequence $\mathbf{x}=(x_i)_{0\leq i<N}$ of $N$ elements of $\mathcal{X}$, and for every $0\leq j\leq k<N$, we define the subsequence $\mathbf{x}_j^k$ as the sequence $(x_i')_{0\leq i\leq k-j}$, where $x_i'=x_{i+j}$ for every $0\leq i\leq k-j$.
\end{mynot}

\begin{mynot}
\label{defnotgstar}
For every $k\geq 0$ and every sequence $\mathbf{x}=(x_i)_{0\leq i<2^k}$ of $|\mathbf{x}|=2^k$ elements of $\mathcal{X}$, we define $g_{\ast}(\mathbf{x})\in\mathcal{X}$ recursively on $k$ as follows:
\begin{itemize}
\item If $k=0$ (i.e., $\mathbf{x}=(x_0)$), $g_{\ast}(\mathbf{x})=x_0$.
\item If $k>0$, $g_{\ast}(\mathbf{x})=g_{\ast}(\mathbf{x}_{0}^{|\mathbf{x}|/2-1})\ast g_{\ast}(\mathbf{x}_{|\mathbf{x}|/2}^{|\mathbf{x}|-1})=g_{\ast}(\mathbf{x}_{0}^{2^{k-1}-1})\ast g_{\ast}(\mathbf{x}_{2^{k-1}}^{2^{k}-1})$.
\end{itemize}
For example, we have:
\begin{itemize}
\item $g_{\ast}(\mathbf{x}_0^1)=x_0\ast x_1$.
\item $g_{\ast}(\mathbf{x}_0^3)=(x_0\ast x_1)\ast(x_2\ast x_3)$.
\item $g_{\ast}(\mathbf{x}_0^7)=\big((x_0\ast x_1)\ast(x_2\ast x_3)\big)\ast\big((x_4\ast x_5)\ast(x_6\ast x_7)\big)$.
\end{itemize}
\end{mynot}

\begin{mydef}
Let $A$ be a subset of $\mathcal{X}$. We define the probability distribution $\mathbb{I}_A$ on $\mathcal{X}$ as $\mathbb{I}_A(x)=\frac{1}{|A|}$ if $x\in A$ and $\mathbb{I}_A(x)=0$ otherwise.
\end{mydef}

\begin{mydef}
Let $\mathcal{Y}$ be an arbitrary set, $\mathcal{H}$ be a balanced partition of $\mathcal{X}$ and $(X,Y)$ be a random pair in $\mathcal{X}\times\mathcal{Y}$. For every $\gamma>0$, we define:
$$\mathcal{Y}_{\mathcal{H},\gamma}(X,Y)=\Big\{y\in\mathcal{Y}:\;\exists H_y\in\mathcal{H},\; \|\mathbb{P}_{X|Y=y}-\mathbb{I}_{H_y}\|_{\infty}<\gamma\Big\},$$
$$\mathcal{P}_{\mathcal{H},\gamma}(X,Y)=\mathbb{P}_Y\big(\mathcal{Y}_{\mathcal{H},\gamma}(X,Y)\big).$$
\end{mydef}

Note that if $\mathcal{P}_{\mathcal{H},\gamma}(X,Y)\approx 1$ for a small $\gamma$ then $Y$ is ``almost equivalent" to the projection of $X$ onto $\mathcal{H}$. This will be proved rigorously in step 2. The next proposition will be used later to show that a relation $\mathcal{P}_{\mathcal{H},\gamma}(X,Y)\approx 1$ is satisfied between the input and output of a polarized channel, where $\mathcal{H}$ is a stable partition. This is why we say that polarized channels are projection channels onto stable partitions.

\begin{myprop}
Let $\ast$ be a strongly ergodic operation on a set $\mathcal{X}$. Define $k=2^{2^{|\mathcal{X}|}}+ \scon(\ast)$ and let $\mathcal{Y}$ be an arbitrary set. For any $\gamma>0$, there exists $\epsilon(\gamma)>0$ depending only on $\mathcal{X}$ such that if $(X_i,Y_i)_{0\leq i< 2^k}$ is a sequence of $2^k$ random pairs satisfying:
\begin{enumerate}
\item $(X_i,Y_i)_{0\leq i< 2^k}$ are independent and identically distributed in $\mathcal{X}\times \mathcal{Y}$,
\item $X_i$ is uniform in $\mathcal{X}$ for all $0\leq i<2^k$,
\item $H\big(g_{\ast}(X_0^{2^{k}-1})|Y_0^{2^{k}-1}\big)<H(X_0|Y_0)+\epsilon(\gamma)$,
\end{enumerate}
then there exists a stable partition $\mathcal{H}$ of $(\mathcal{X},\ast)$ such that $\mathcal{P}_{\mathcal{H},\gamma}(X_0,Y_0)>1-\gamma$.
\label{MainProp}
\end{myprop}
\begin{proof}
See Appendix \ref{apppA}.
\end{proof}

\vspace*{4mm}
\subsubsection*{Step 2: Structure of projection channels}

\begin{mylem}
Let $\mathcal{X}$ be an arbitrary set and let $\ast$ be an ergodic operation on $\mathcal{X}$. For every $\delta>0$, there exists $\gamma:=\gamma(\delta)>0$ such that for any stable partition $\mathcal{H}$ of $(\mathcal{X},\ast)$, if $(X,Y)$ is a pair of random variables in $\mathcal{X}\times\mathcal{Y}$ satisfying
\begin{enumerate}
\item $X$ is uniform in $\mathcal{X}$,
\item $\mathcal{P}_{\mathcal{H},\gamma}(X;Y)>1-\gamma$,
\end{enumerate}
then $\Big|I\big(\proj_{\mathcal{H}'}(X);Y\big)-\log\frac{|\mathcal{H}|\cdot\|\mathcal{H}\wedge\mathcal{H}'\|}{\|\mathcal{H}'\|}\Big|<\delta$ for every stable partition $\mathcal{H}'$ of $(\mathcal{X},\ast)$.
\label{lemSufWedge}
\end{mylem}
\begin{proof}
Let $\mathcal{H}'$ be a stable partition of $\mathcal{X}$. Note that the entropy function is continuous and the space of probability distributions on $\mathcal{H}'$ is compact. Therefore, the entropy function is uniformly continuous, which means that for every $\delta>0$ there exists $\gamma_{\mathcal{H'}}'(\delta)>0$ such that if $p_1$ and $p_2$ are two probability distributions on $\mathcal{H}'$ satisfying $\|p_1-p_2\|_{\infty}<\gamma_{\mathcal{H'}}'(\delta)$ then $|H(p_1)-H(p_2)|<\frac{\delta}{2}$. Let $\delta>0$ and define $\gamma_{\mathcal{H}'}(\delta)=\min\Big\{\frac{\delta}{2\log(|\mathcal{H'}|+1)},\frac{1}{\|\mathcal{H}'\|}\gamma'_{\mathcal{H'}}(\delta)\Big\}$. Now define $\gamma(\delta)=\min\{\gamma_{\mathcal{H}'}(\delta):\;\mathcal{H}'\;\text{is\;a\;stable\;partition}\}$ which depends only on $(\mathcal{X},\ast)$ and $\delta$. Clearly, $\|\mathcal{H}'\|\gamma(\delta)\leq\gamma'_{\mathcal{H}'}(\delta)$ for every stable partition $\mathcal{H}'$ of $\mathcal{X}$.

Let $\mathcal{H}$ be a stable partition of $\mathcal{X}$ and suppose that $\mathcal{P}_{\mathcal{H},\gamma(\delta)}(X;Y)>1-\gamma(\delta)$, where $X$ is uniform in $\mathcal{X}$. Fix $y\in\mathcal{Y}_{\mathcal{H},\gamma(\delta)}(X;Y)$. By the definition of $\mathcal{Y}_{\mathcal{H},\gamma(\delta)}(X;Y)$, there exists $H_y\in\mathcal{H}$ such that $|\mathbb{P}_{X|Y}(x|y)-\mathbb{I}_{H_y}(x)|<\gamma(\delta)$ for every $x\in\mathcal{X}$.

Let $\mathcal{H}'$ be a stable partition of $\mathcal{X}$. Corollary \ref{corWedge} of Part I \cite{RajErgI} shows that $\mathcal{H}\wedge\mathcal{H}'$ is also a stable partition of $\mathcal{X}$. From the definition of $\mathcal{H}\wedge\mathcal{H}'$, for every $H'\in\mathcal{H}'$ we have either $H_y\cap H'=\o$ or $H_y\cap H'\in \mathcal{H}\wedge\mathcal{H}'$. Therefore, we have either $|H_y\cap H'|=0$ or $|H_y\cap H'|=\|\mathcal{H}\wedge\mathcal{H}'\|$. Let $\mathcal{H}'_y=\{H'\in\mathcal{H}':\; H_y\cap H'\neq \o\}$, so $|H_y\cap H'|=\|\mathcal{H}\wedge\mathcal{H}'\|$ for all $H'\in\mathcal{H}_y'$. Now since $\displaystyle H_y=\bigcup_{H'\in\mathcal{H}'}(H_y\cap H')$, we have $\displaystyle \|\mathcal{H}\|=|H_y|=\sum_{H'\in\mathcal{H}'}|H_y\cap H'|=|\mathcal{H}'_y|\cdot\|\mathcal{H}\wedge\mathcal{H}'\|$. Therefore,
\begin{equation}
\label{leqHprimey}
\frac{\|\mathcal{H}\|}{\|\mathcal{H}\wedge\mathcal{H}'\|}=|\mathcal{H}'_{y}|\leq |\mathcal{H}'|.
\end{equation}

We will now show that for every $y\in\mathcal{Y}_{\mathcal{H},\gamma(\delta)}$, we have $\|\mathbb{P}_{\proj_{\mathcal{H'}}(X)|Y=y}-\mathbb{I}_{\mathcal{H}'_y}\|_{\infty}<\gamma'_{\mathcal{H}'}(\delta)$, where $\mathbb{I}_{\mathcal{H}'_y}$ is the probability distribution on $\mathcal{H}'$ defined as $\mathbb{I}_{\mathcal{H}'_y}(H')=\frac{1}{|\mathcal{H}'_y|}$ if $H'\in\mathcal{H}'_y$ and $\mathbb{I}_{\mathcal{H}'_y}(H')=0$ otherwise. This will be useful to show that $\Big|H(\proj_{\mathcal{H'}}(X)|Y=y)-\log \frac{\|\mathcal{H}\|}{\|\mathcal{H}\wedge\mathcal{H}'\|}\Big|<\frac{\delta}{2}$ for all $y\in\mathcal{Y}_{\mathcal{H},\gamma(\delta)}$.

Let $y\in\mathcal{Y}_{\mathcal{H},\gamma(\delta)}$ and $H'\in \mathcal{H}'$. We have $\displaystyle\mathbb{P}_{\proj_{\mathcal{H'}}(X)|Y}(H'|y)=\sum_{x\in H'}\mathbb{P}_{X|Y}(x|y)$. But since $|\mathbb{P}_{X|Y}(x|y)-\frac{1}{|H_y|}|<\gamma(\delta)$ for every $x\in H_y$, and since $\mathbb{P}_{X|Y}(x|y)<\gamma(\delta)$ if $x\in\mathcal{X}\setminus H_y$, we conclude that $\big|\mathbb{P}_{\proj_{\mathcal{H'}}(X)|Y}(H'|y)-\frac{|H'\cap H_y|}{|H_y|}\big|< |H'|\gamma(\delta)= \|\mathcal{H}'\|\gamma(\delta)\leq \gamma'_{\mathcal{H}'}(\delta)$. We conclude:
\begin{itemize}
\item If $H'\in\mathcal{H}_y'$, we have $|H'\cap H_y|=\|\mathcal{H}\wedge\mathcal{H}'\|$ which means that $\frac{|H'\cap H_y|}{|H_y|}=\frac{\|\mathcal{H}\wedge\mathcal{H}'\|}{\|\mathcal{H}\|}\stackrel{(a)}{=}\frac{1}{|\mathcal{H}'_y|}$, where (a) follows from \eqref{leqHprimey}. Thus $|\mathbb{P}_{\proj_{\mathcal{H'}}(X)|Y}(H'|y)-\frac{1}{|\mathcal{H}'_y|}|< \gamma'_{\mathcal{H}'}(\delta)$.
\item If $H'\in\mathcal{H}'\setminus\mathcal{H}_y'$, $\frac{|H'\cap H_y|}{|H_y|}=0$ and so $\mathbb{P}_{\proj_{\mathcal{H'}}(X)|Y}(H'|y)< \gamma'_{\mathcal{H}'}(\delta)$.
\end{itemize}
Therefore, $\|\mathbb{P}_{\proj_{\mathcal{H'}}(X)|Y=y}-\mathbb{I}_{\mathcal{H}'_y}\|_{\infty}<\gamma'_{\mathcal{H}'}(\delta)$. This means that $\big|H(\proj_{\mathcal{H'}}(X)|Y=y)-H(\mathbb{I}_{\mathcal{H}'_y})\big|<\frac{\delta}{2}$. But $H(\mathbb{I}_{\mathcal{H}'_y})=\log|\mathcal{H}'_y|\stackrel{(a)}{=}\log \frac{\|\mathcal{H}\|}{\|\mathcal{H}\wedge\mathcal{H}'\|}$, where (a) follows from \eqref{leqHprimey}. Therefore,
\begin{equation}
\label{eqleqwedge1}
\forall y\in\mathcal{Y}_{\mathcal{H},\gamma(\delta)},\;\Big|H(\proj_{\mathcal{H'}}(X)|Y=y)-\log \frac{\|\mathcal{H}\|}{\|\mathcal{H}\wedge\mathcal{H}'\|}\Big|<\frac{\delta}{2}.
\end{equation}

On the other hand, for every $y\in\mathcal{Y}_{\mathcal{H},\gamma(\delta)}^c$, $\mathbb{P}_{\proj_{\mathcal{H'}}(X)|Y=y}$ is a probability distribution on $\mathcal{H}'$ which implies that $0\leq H(\proj_{\mathcal{H'}}(X)|Y=y)\leq \log|\mathcal{H}'|$. Moreover, we have $0\leq \log \frac{\|\mathcal{H}\|}{\|\mathcal{H}\wedge\mathcal{H}'\|}\leq \log|\mathcal{H}'|$ from \eqref{leqHprimey}. Therefore,

\begin{equation}
\label{eqleqwedge2}
\forall y\in\mathcal{Y}_{\mathcal{H},\gamma(\delta)}^c,\;\Big|H(\proj_{\mathcal{H'}}(X)|Y=y)-\log \frac{\|\mathcal{H}\|}{\|\mathcal{H}\wedge\mathcal{H}'\|}\Big|\leq \log|\mathcal{H}'|.
\end{equation}

We conclude that:
\begin{align*}
\Big|H(\proj_{\mathcal{H'}}(X)|Y)-\log \frac{\|\mathcal{H}\|}{\|\mathcal{H}\wedge\mathcal{H}'\|}\Big|&\leq \sum_{y\in\mathcal{Y}} \Big|H(\proj_{\mathcal{H'}}(X)|Y=y)-\log \frac{\|\mathcal{H}\|}{\|\mathcal{H}\wedge\mathcal{H}'\|}\Big|\cdot\mathbb{P}_Y(y)\\
&\stackrel{(a)}{\leq}\sum_{y\in\mathcal{Y}_{\mathcal{H},\gamma(\delta)}}\frac{\delta}{2}\cdot\mathbb{P}_Y(y) +\sum_{y\in\mathcal{Y}_{\mathcal{H},\gamma(\delta)}^c} (\log|\mathcal{H'}|)\cdot\mathbb{P}_Y(y)\\
&=\frac{\delta}{2}\cdot\mathbb{P}_Y(\mathcal{Y}_{\mathcal{H},\gamma(\delta)}) + (\log|\mathcal{H}'|)\mathbb{P}_Y(\mathcal{Y}_{\mathcal{H},\gamma(\delta)}^c)\stackrel{(b)}{<} \frac{\delta}{2} + (\log|\mathcal{H}'|)\gamma(\delta)\\
&\leq \frac{\delta}{2} + (\log|\mathcal{H}'|)
\cdot\frac{\delta}{2\log(|\mathcal{H}'|+1)}< \delta,
\end{align*}
where (a) follows from \eqref{eqleqwedge1} and \eqref{eqleqwedge2}. (b) follows from the second condition of the lemma.

Now since $\proj_{\mathcal{H}'}(X)$ is uniform in $\mathcal{H}'$, we have $H(\proj_{\mathcal{H}'}(X))=\log|\mathcal{H}'|$. We conclude that if $\mathcal{P}_{\mathcal{H},\gamma(\delta)}(X,Y)>1-\gamma(\delta)$ then for every stable partition $\mathcal{H}'$ of $(\mathcal{X},\ast)$, we have $$\Big|I\big(\proj_{\mathcal{H}'}(X);Y\big)-\log\frac{|\mathcal{H}'|\cdot\|\mathcal{H}\wedge\mathcal{H}'\|}{\|\mathcal{H}\|}\Big|<\delta,$$ which implies that $\Big|I\big(\proj_{\mathcal{H}'}(X);Y\big)-\log\frac{|\mathcal{H}|\cdot\|\mathcal{H}\wedge\mathcal{H}'\|}{\|\mathcal{H}'\|}\Big|<\delta$ since $|\mathcal{H}|\cdot\|\mathcal{H}\|=|\mathcal{H}'|\cdot\|\mathcal{H}'\|=|\mathcal{X}|$.
\end{proof}
\vspace*{4mm}
\subsubsection*{Step 3: Projection channels are easy}

\begin{mydef}
\label{defprojchannel}
Let $\mathcal{H}$ be a balanced partition of $\mathcal{X}$ and let $W:\mathcal{X}\longrightarrow\mathcal{Y}$. We define the channel $W[\mathcal{H}]:\mathcal{H}\longrightarrow\mathcal{Y}$ by:
$$W[\mathcal{H}](y|H)=\frac{1}{\|\mathcal{H}\|}\sum_{\substack{x\in \mathcal{X}:\\\proj_{\mathcal{H}}(x) = H}}W(y|x)=\frac{1}{|H|}\sum_{x\in H}W(y|x).$$
\end{mydef}

\begin{myrem}
If $X$ is a random variable uniformly distributed in $\mathcal{X}$ and $Y$ is the output of the channel $W$ when $X$ is the input, then it is easy to see that $I(W[\mathcal{H}])=I(\proj_{\mathcal{H}}(X);Y)$.
\end{myrem}

\begin{mythe}
Let $\mathcal{X}$ be an arbitrary set and let $\ast$ be a uniformity preserving operation on $\mathcal{X}$ such that $/^{\ast}$ is strongly ergodic. Let $W:\mathcal{X}\longrightarrow\mathcal{Y}$ be an arbitrary channel. Then for any $\delta>0$, we have:
\begin{align*}
\lim_{n\to\infty} \frac{1}{2^n} \Bigg|\bigg\{ s &\in\{-,+\}^n:\;\exists \mathcal{H}_s\; \text{a stable partition of $(\mathcal{X},/^{\ast})$},\\
&\Big| I(W^s[\mathcal{H}'])-\log\frac{|\mathcal{H}_s|\cdot\|\mathcal{H}_s\wedge\mathcal{H}'\|}{\|\mathcal{H}'\|}\Big|<\delta\;\text{for all stable partitions $\mathcal{H}'$ of $(\mathcal{X},/^{\ast})$} \bigg\}\Bigg| = 1.
\end{align*}
\label{mainthe11}
\end{mythe}
\begin{proof}
Let $(W_n)_n$ be as in Definition \ref{def1}. Since $\ast$ is uniformity preserving, it satisfies the conservation property of Definition \ref{defPola} (see Remark \ref{rem1}). Therefore, we have:
$$\mathbb{E}\big[I(W_{n+1})|W_n\big]=\frac{1}{2}I(W_n^-)+\frac{1}{2}I(W_n^+)=I(W_n).$$
This implies that the process $(I(W_n))_n$ is a martingale, and so it converges almost surely. Therefore, the process $\big(I(W_{n+k})-I(W_n)\big)_n$ converges almost surely to zero, where $k=2^{2^{|\mathcal{X}|}}+\scon(/^{\ast})$. In particular, $\big(I(W_{n+k})-I(W_n)\big)_n$ converges in probability to zero, hence for every $\delta>0$ we have $$\displaystyle \lim_{n\rightarrow\infty}\mathbb{P}\Big[|I(W_{n+k})-I(W_n)|\geq\epsilon\big(\gamma(\delta)\big)\Big]=0,$$ where $\epsilon(.)$ is given by Proposition \ref{MainProp} and $\gamma(.)$ is given by Lemma \ref{lemSufWedge}. We have:

\begin{align*}
\mathbb{P}\Big[|I(W_{n+k})-I(W_n)|\geq \epsilon\big(\gamma(\delta)\big)\Big]=\frac{1}{2^{n+k}}|A_{n,k}|,
\end{align*}
where
$\displaystyle A_{n,k}=\bigg\{ (s,s') \in\{-,+\}^n\times\{-,+\}^k:\; |I(W^{(s,s')})-I(W^s)|\geq \epsilon\big(\gamma(\delta)\big) \bigg\}$. Define:
$$B_{n,k}=\bigg\{ s \in\{-,+\}^n:\; |I(W^{(s,[k]^-)})-I(W^s)|\geq \epsilon\big(\gamma(\delta)\big) \bigg\},$$
where $[k]^-\in\{-,+\}^k$ is the sequence consisting of $k$ minus signs. Clearly, $B_{n,k}\times \{[k]^-\}\subset A_{n,k}$ and so $|B_{n,k}|\leq |A_{n,k}|$. Now since $\displaystyle \lim_{n\rightarrow\infty} \frac{1}{2^{n+k}}|A_{n,k}|=\lim_{n\rightarrow\infty}\mathbb{P}\Big[|I(W_{n+k})-I(W_n)|\geq\epsilon\big(\gamma(\delta)\big)\Big]=0$, we must have $\displaystyle \lim_{n\rightarrow\infty} \frac{1}{2^{n+k}}|B_{n,k}|=0$. Therefore, $\displaystyle \lim_{n\rightarrow\infty} \frac{1}{2^{n}}|B_{n,k}|=2^k\times 0=0$ and so $\displaystyle \lim_{n\rightarrow\infty} \frac{1}{2^{n}}|B_{n,k}^c|=1$.

Now suppose that $s\in B_{n,k}^c$, i.e., $|I(W^{(s,[k]^-)})-I(W^s)|< \epsilon\big(\gamma(\delta)\big)$. Let $U_0,\ldots,U_{2^k-1}$ be $2^k$ independent random variables uniformly distributed in $\mathcal{X}$. For every $0\leq j\leq k$, define the sequence $U_{j,0},\ldots,U_{j,2^k-1}$ recursively as follows:
\begin{itemize}
\item $U_{0,i}=U_i$ for every $0\leq i<2^k$.
\item For every $0\leq j< k$ and every $0\leq i<2^k$, define $U_{j+1,i}$ as follows:
$$U_{j+1,i}=\begin{cases}
U_{j,i}\ast U_{j,i+2^{k-j-1}}\;&\text{if }0\leq i\bmod 2^{k-j}<2^{k-j-1},\\
U_{j,i}\;&\text{if }2^{k-j-1}\leq i\bmod 2^{k-j}<2^{k-j}.
\end{cases}$$
\end{itemize}
Since $\ast$ is uniformity preserving, it is easy to see that for every $0\leq i\leq k$, the $2^k$ random variables $U_{j,0},\ldots,U_{j,2^k-1}$ are independent and uniform in $\mathcal{X}$. In particular, if we define $X_i=U_{k,i}$ for $0\leq i<2^k$, then $X_0,\ldots,X_{2^k-1}$ are $2^k$ independent random variables uniformly distributed in $\mathcal{X}$. Suppose that $X_0,\ldots,X_{2^k-1}$ are sent through $2^k$ independent copies of the channel $W^s$ and let $Y_0,\ldots,Y_{2^k-1}$ be the output of each copy of the channel respectively. Clearly, $(X_i,Y_i)_{0\leq i<2^k}$ are independent and uniformly distributed in $\mathcal{X}\times\mathcal{Y}$. Moreover, $I(W^s)=I(X_i;Y_i)$ for every $0\leq i<2^k$. In particular, $I(W^s)=I(X_0;Y_0)=H(X_0)-H(X_0|Y_0)=\log|\mathcal{X}|-H(X_0|Y_0)$. We will show by backward induction on $0\leq j\leq k$ that for every $0\leq q< 2^j$ we have:
\begin{itemize}
\item $W^{(s,[k-j]^-)}$ is equivalent to the channel $U_{j,q\cdot 2^{k-j}}\longrightarrow Y_{q\cdot 2^{k-j}}^{(q+1)\cdot 2^{k-j}-1}$.
\item $U_{j,q\cdot 2^{k-j}}=g_{/^{\ast}}\big(X_{q\cdot 2^{k-j}}^{(q+1)\cdot 2^{k-j}-1}\big)$.
\end{itemize}
The claim is trivial for $j=k$. Now let $0\leq j<k$ and suppose that the claim is true for $j+1$. Let $0\leq q<2^j$. From the induction hypothesis we have:
\begin{itemize}
\item $W^{(s,[k-j-1]^-)}$ is equivalent to the channel $U_{j+1,q\cdot 2^{k-j}}\longrightarrow Y_{q\cdot 2^{k-j}}^{(2q+1)\cdot 2^{k-j-1}-1}$.
\item $U_{j+1,q\cdot 2^{k-j}}=g_{/^{\ast}}\big(X_{q\cdot 2^{k-j}}^{(2q+1)\cdot 2^{k-j-1}-1}\big)$.
\item $W^{(s,[k-j-1]^-)}$ is equivalent to the channel $U_{j+1,(2q+1)\cdot 2^{k-j-1}}\longrightarrow Y_{(2q+1)\cdot 2^{k-j-1}}^{(q+1)\cdot 2^{k-j}-1}$.
\item $U_{j+1,(2q+1)\cdot 2^{k-j-1}}=g_{/^{\ast}}\big(X_{(2q+1)\cdot 2^{k-j-1}}^{(q+1)\cdot 2^{k-j}-1}\big)$.
\end{itemize}
Now since $U_{j+1,q\cdot 2^{k-j}}=U_{j,q\cdot 2^{k-j}}\ast U_{j,(2q+1)\cdot 2^{k-j-1}}$ and $U_{j+1,(2q+1)\cdot 2^{k-j-1}}=U_{j,(2q+1)\cdot 2^{k-j-1}}$, it follows that $W^{(s,[k-j]^-)}=(W^{(s,[k-j-1]^-)})^-$ is equivalent to the channel $U_{j,q\cdot 2^{k-j}}\longrightarrow Y_{q\cdot 2^{k-j}}^{(q+1)\cdot 2^{k-j}-1}$ (see Remark \ref{rem1}). Moreover, we have
\begin{align*}
U_{j,q\cdot 2^{k-j}}&=U_{j+1,q\cdot 2^{k-j}}/^{\ast} U_{j,(2q+1)\cdot 2^{k-j-1}}=U_{j+1,q\cdot 2^{k-j}}/^{\ast} U_{j+1,(2q+1)\cdot 2^{k-j-1}}\\
&=g_{/^{\ast}}\big(X_{q\cdot 2^{k-j}}^{(2q+1)\cdot 2^{k-j-1}-1}\big)/^{\ast}g_{/^{\ast}}\big(X_{(2q+1)\cdot 2^{k-j-1}}^{(q+1)\cdot 2^{k-j}-1}\big)=g_{/^{\ast}}\big(X_{q\cdot 2^{k-j}}^{(q+1)\cdot 2^{k-j}-1}\big).
\end{align*}
This terminates the induction argument and so the claim is true for all $0\leq j\leq k$. In particular, for $j=0$ and $q=0$, we have $U_0=U_{0,0}=g_{/^{\ast}}\big(X_{0}^{2^k-1}\big)$ and $W^{(s,[k]^-)}$ is equivalent to the channel $U_0\longrightarrow Y_{0}^{2^k-1}$. Thus, $$I(W^{(s,[k]^-)})=I(U_0;Y_0^{2^k-1})=H(U_0)-H(U_0|Y_0^{2^k-1})=\log|\mathcal{X}|-H(U_0|Y_0^{2^k-1}).$$ Hence
\begin{align*}
I(W^{(s,[k]^-)})-I(W^s)&=\log|\mathcal{X}|-H(U_0|Y_0^{2^k-1})-\log|\mathcal{X}|+H(X_0|Y_0)\\
&\stackrel{(a)}{=}H(X_0|Y_0)-H\big(g_{/^\ast}(X_0^{2^k-1})|Y_0^{2^k-1}\big),
\end{align*}
where (a) follows from the fact that $U_0=g_{/^{\ast}}\big(X_{0}^{2^k-1}\big)$. We conclude that $$\big|H\big(g_{/^\ast}(X_0^{2^k-1})|Y_0^{2^k-1}\big)-H(X_0|Y_0)\big|=|I(W^{(s,[k]^-)})-I(W^s)|<\epsilon\big(\gamma(\delta)\big).$$ Proposition \ref{MainProp}, applied to $/^{\ast}$, implies the existence of a stable partition $\mathcal{H}_s$ of $(\mathcal{X},/^{\ast})$ such that $\mathcal{P}_{\mathcal{H}_s,\gamma(\delta)}(X_0,Y_0)>1-\gamma(\delta)$. Now Lemma \ref{lemSufWedge}, applied to $/^{\ast}$, implies that for every stable partition $\mathcal{H}'$ of $(\mathcal{X},/^{\ast})$, we have $\Big|I(W^s[\mathcal{H}'])-\log\frac{|\mathcal{H}_s|\cdot\|\mathcal{H}_s\wedge\mathcal{H}'\|}{\|\mathcal{H}'\|}\Big|=\Big|I\big(\proj_{\mathcal{H}'}(X_0);Y_0\big)-\log\frac{|\mathcal{H}_s|\cdot\|\mathcal{H}_s\wedge\mathcal{H}'\|}{\|\mathcal{H}'\|}\Big|<\delta$. But this is true for every $s\in B_{n,k}^c$. Therefore, $B_{n,k}^c\subset D_n$, where $D_n$ is defined as:
\begin{align*}
D_n=\bigg\{ s \in\{-,+\}^n:\;&\exists \mathcal{H}_s\; \text{a stable partition of $(\mathcal{X},/^{\ast})$},\\
&\Big| I(W^s[\mathcal{H}'])-\log\frac{|\mathcal{H}_s|\cdot\|\mathcal{H}_s\wedge\mathcal{H}'\|}{\|\mathcal{H}'\|}\Big|<\delta\;\text{for all stable partitions $\mathcal{H}'$ of $(\mathcal{X},/^{\ast})$} \bigg\}.
\end{align*}
Now since $\displaystyle \lim_{n\rightarrow\infty} \frac{1}{2^{n}}|B_{n,k}^c|=1$ and $B_{n,k}^c\subset D_n$, we must have $\displaystyle \lim_{n\rightarrow\infty} \frac{1}{2^{n}}|D_n|=1$.
\end{proof}

\begin{mycor}
Let $\mathcal{X}$ be an arbitrary set and let $\ast$ be a uniformity preserving operation on $\mathcal{X}$ such that $/^{\ast}$ is strongly ergodic, and let $W:\mathcal{X}\longrightarrow\mathcal{Y}$ be an arbitrary channel. Then for any $\delta>0$, we have:
\begin{align*}
\lim_{n\to\infty} \frac{1}{2^n} \bigg|\Big\{ s\in\{-,+\}^n:\;&\exists \mathcal{H}_s\; \text{a stable partition of $(\mathcal{X},/^{\ast})$},\\
&\big| I(W^s)-\log|\mathcal{H}_s|\big|<\delta, \big| I(W^s[\mathcal{H}_s])-\log|\mathcal{H}_s|\big|<\delta \Big\}\bigg| = 1.
\end{align*}
\label{cor1}
\end{mycor}
\begin{proof}
We apply Theorem \ref{mainthe11} and we consider the two particular cases where $\mathcal{H}'=\big\{\{x\}:\; x\in\mathcal{X}\big\}$ and $\mathcal{H}'=\mathcal{H}_s$.
\end{proof}

\begin{myrem}
Corollary \ref{cor1} can be interpreted as follows: In a polarized channel $W^s$, we have $I(W^s)\approx I(W^s[\mathcal{H}_s])\approx \log|\mathcal{H}_s|$ for some stable partition $\mathcal{H}_s$ of $(\mathcal{X},/^{\ast})$. Let $X_s$ and $Y_s$ be the input and output of the channel $W^s$ respectively. $I(W^s[\mathcal{H}_s])\approx \log|\mathcal{H}_s|$ means that $Y_s$ ``almost" determines $\proj_{\mathcal{H}_s}(X_s)$. On the other hand, $I(W^s)\approx I(W^s[\mathcal{H}_s])$ means that there is ``almost" no other information about $X_s$ which can be determined from $Y_s$. Therefore, $W^s$ is ``almost" equivalent to the channel $X_s\longrightarrow\proj_{\mathcal{H}_s}(X_s)$.
\end{myrem}

\begin{mylem}
\label{lemeasy}
Let $W:\mathcal{X}\longrightarrow\mathcal{Y}$ be an arbitrary channel. If there exists a balanced partition $\mathcal{H}$ of $\mathcal{X}$ such that $\big| I(W)-\log|\mathcal{H}|\big|<\delta$ and $\big| I(W[\mathcal{H}])-\log|\mathcal{H}|\big|<\delta$, then $W$ is $\delta$-easy.
\end{mylem}
\begin{proof}
Let $L=|\mathcal{H}|$ and let $H_1,\ldots,H_L$ be the $L$ members of $\mathcal{H}$. Let $\mathcal{S}=\big\{C\subset \mathcal{X}:\; |C|=L\big\}$ and $\mathcal{S}_{\mathcal{H}}=\big\{\{x_1,\ldots,x_L\}:\; x_1\in H_1,\ldots, x_L\in H_L\big\}\subset\mathcal{S}$. For each $1\leq i\leq L$, let $X_i$ be a random variable uniformly distributed in $H_i$. Define $\mathcal{B}=\{X_1,\ldots,X_L\}$, which is a random set taking values in $\mathcal{S}_{\mathcal{H}}$. Note that we can see $\mathcal{B}$ as a random variable in $\mathcal{S}$ since $\mathcal{S}_{\mathcal{H}}\subset \mathcal{S}$. For every $x\in\mathcal{X}$, let $H_i$ be the unique element of $\mathcal{H}$ such that $x\in H_i$. We have:
\begin{equation}
\label{eqUnif123}
\frac{1}{L}\sum_{C\in\mathcal{H}}\mathbb{P}_{\mathcal{B}}(C)\mathds{1}_{x\in C}=\frac{1}{|\mathcal{H}|}\mathbb{P}[x\in\mathcal{B}]\stackrel{(a)}{=}\frac{1}{|\mathcal{H}|}\mathbb{P}[X_i=x]=\frac{1}{|\mathcal{H}|}\cdot\frac{1}{|H_i|}=\frac{1}{|\mathcal{H}|}\cdot\frac{1}{\|\mathcal{H}\|}=\frac{1}{|\mathcal{X}|},
\end{equation}
where (a) follows from the fact that $x\in\mathcal{B}$ if and only if $X_i=x$. Now for each $C\in\mathcal{S}_{\mathcal{H}}$, define the bijection $f_C:\{1,\ldots,L\}\rightarrow C$ as follows: for each $1\leq i\leq L$, $f_C(i)$ is the unique element in $C\cap H_i$ (so $\proj_{\mathcal{H}}(f_C(i))=H_i$). Let $U$ be a random variable chosen uniformly in $\{1,\ldots,L\}$ and independently from $\mathcal{B}$, and let $X=f_\mathcal{B}(U)$ (so $\proj_{\mathcal{H}}(X)=H_U$). From \eqref{eqUnif123} we get that $X$ is uniform in $\mathcal{X}$.

Let $Y$ be the output of the channel $W$ when $X$ is the input. From Definition \ref{defeasy}, we have $I(W_\mathcal{B})=I(U;Y,\mathcal{B})$. On the other hand, $I(W[\mathcal{H}])=I(\proj_{\mathcal{H}}(X);Y)=I(H_U;Y)$. Therefore, $I(W_\mathcal{B})=I(U;Y,\mathcal{B})\geq I(U;Y)\stackrel{(a)}{=}I(H_U;Y)=I(W[\mathcal{H}])\stackrel{(b)}{>}\log L -\delta$, where (a) follows from the fact that the mapping $u\rightarrow H_u$ is a bijection from $\{1,\ldots,L\}$ to $\mathcal{H}$ and (b) follows from the fact that $\big| I(W[\mathcal{H}])-\log|\mathcal{H}|\big|<\delta$. We conclude that $W$ is $\delta$-easy since $I(W_\mathcal{B})>\log L-\delta$ and $|I(W)-\log L|<\delta$.
\end{proof}

\begin{myprop}
\label{propsuff}
If $\ast$ is a uniformity preserving operation on a set $\mathcal{X}$ and $/^{\ast}$ is strongly ergodic, then $\ast$ is polarizing.
\end{myprop}
\begin{proof}
We have the following:
\begin{itemize}
\item We know from Remark \ref{rem1} that since $\ast$ is uniformity preserving, it satisfies the conservation property of Definition \ref{defPola}.
\item The polarization property of Definition \ref{defPola} follows immediately from Corollary \ref{cor1} and Lemma \ref{lemeasy}.
\end{itemize}
Therefore, $\ast$ is polarizing.
\end{proof}

\begin{mythe}
If $\ast$ is a binary operation on a set $\mathcal{X}$, then $\ast$ is polarizing if and only if $\ast$ is uniformity preserving and $/^{\ast}$ is strongly ergodic.
\label{thecaraccharac}
\end{mythe}
\begin{proof}
The theorem follows from Propositions \ref{propnec} and \ref{propsuff}.
\end{proof}

\section{Exponent of a polarizing operation}

In this section, we study the exponent of polarizing operations.

\begin{mydef}
Let $W$ be a channel with input alphabet $\mathcal{X}$ and output alphabet $\mathcal{Y}$. For every $x,x'\in\mathcal{X}$, we define the channel $W_{x,x'}:\{0,1\}\rightarrow\mathcal{Y}$ as follows:
\begin{equation*}
W_{x,x'}(y|b)=
\begin{cases}
W(y|x)\;\text{if}\;b=0,\\
W(y|x')\;\text{if}\;b=1.
\end{cases}
\end{equation*}
The \emph{Battacharyya parameter between $x$ and $x'$ of the channel $W$} is the Bhattacharyya parameter of the channel $W_{x,x'}$:
$$Z(W_{x,x'}):=\sum_{y\in\mathcal{Y}}\sqrt{W_{x,x'}(y|0)W_{x,x'}(y|1)}=\sum_{y\in\mathcal{Y}}\sqrt{W(y|x)W(y|x')}.$$
It is easy to see that $0\leq Z(W_{x,x'})\leq 1$ for every $x,x'\in \mathcal{X}$. Moreover, if $x=x'$ we have $Z(W_{x,x'})=Z(W_{x,x})=1$.

If $|\mathcal{X}|\geq 2$, the \emph{Battacharyya parameter of the channel $W$} is defined as:
$$Z(W):=\frac{1}{|\mathcal{X}|(|\mathcal{X}|-1)}\sum_{\substack{(x,x')\in\mathcal{X}\times\mathcal{X}\\x\neq x'}}Z(W_{x,x'}).$$
We can easily see that $0\leq Z(W)\leq 1$.
\end{mydef}

\begin{myprop}
\label{propBhat}
The Bhattacharyya parameter of a channel $W:\mathcal{X}\rightarrow\mathcal{Y}$ has the following properties:
\begin{enumerate}
\item $\displaystyle Z(W)^2\leq 1-\frac{I(W)}{\log|\mathcal{X}|}$.
\item $\displaystyle I(W)\geq \log\frac{|\mathcal{X}|}{1+(|\mathcal{X}|-1)Z(W)}$.
\item $\displaystyle \frac{1}{4} Z(W)^2\leq \mathbb{P}_e(W)\leq (|\mathcal{X}|-1)Z(W)$, where $\mathbb{P}_e(W)$ is the probability of error of the maximum likelihood decoder of $W$ for uniformly distributed input.
\end{enumerate}
\end{myprop}
\begin{proof}
See Appendix \ref{appB}.
\end{proof}

\begin{myrem}
Proposition \ref{propBhat} shows that $Z(W)$ measures the ability of the receiver to reliably decode the output and correctly estimate the input:
\begin{itemize}
\item If $Z(W)$ is low, the inequality $\mathbb{P}_e(W)\leq (|\mathcal{X}|-1)Z(W)$ implies that $\mathbb{P}_e(W)$ is also low and the receiver can determine the input from the output with high probability. This is also expressed by inequality 2) of Proposition \ref{propBhat}: if $Z(W)$ is close to 0, $I(W)$ is close to $\log|\mathcal{X}|$.
\item If $Z(W)$ is close to 1, inequality 1) of Proposition \ref{propBhat} implies that $I(W)$ is close to 0, which means that the input and the output are ``almost" independent and so it is not possible to recover the input reliably. This is also expressed by the inequality $\displaystyle \mathbb{P}_e(W)\geq \frac{1}{4} Z(W)^2$: if $Z(W)$ is high, $\mathbb{P}_e(W)$ cannot be too low.
\end{itemize}
Since $W_{x,x'}$ is the binary input channel obtained by sending either $x$ or $x'$ through $W$, $Z(W_{x,x'})$ can be interpreted as a measure of the ability of the receiver to distinguish between $x$ and $x'$: if $Z(W_{x,x'})\approx 0$, the receiver can reliably distinguish between $x$ and $x'$ and if $Z(W_{x,x'})\approx 1$, the receiver cannot distinguish between $x$ and $x'$.
\end{myrem}

\begin{mynot}
Let $x,x'\in\mathcal{X}$ and let $s\in\{-,+\}^n$. Throughout this section, $W^s_{x,x'}$ denotes $(W^s)_{x,x'}$. The channel $W^s_{x,x'}$ should not be confused with $(W_{x,x'})^s$ which is not defined unless a binary operation on $\{0,1\}$ is specified.
\end{mynot}

\begin{mylem}
\label{BhatMinus}
For every $u_1,u_1',v\in\mathcal{X}$, we have $\displaystyle Z(W^{-}_{u_1,u_1'})\geq \frac{1}{|\mathcal{X}|} Z(W_{u_1\ast v,u_1'\ast v})$.
\end{mylem}
\begin{proof}
\begin{align*}
Z(W^-_{u_1,u_1'})&=\sum_{y_1,y_2\in\mathcal{Y}}\sqrt{W^-(y_1,y_2|u_1)W^-(y_1,y_2|u_1')}\\
&=\sum_{y_1,y_2\in\mathcal{Y}}\sqrt{\sum_{u_2,u_2'\in\mathcal{X}}\frac{1}{|\mathcal{X}|^2} W(y_1|u_1\ast u_2)W(y_2|u_2)W(y_1|u_1'\ast u_2')W(y_2|u_2')}\\
&\geq \frac{1}{|\mathcal{X}|}\sum_{y_1,y_2\in\mathcal{Y}}\sqrt{W(y_1|u_1\ast v)W(y_2|v)W(y_1|u_1'\ast v)W(y_2|v)}\\
&=\frac{1}{|\mathcal{X}|}\sum_{y_1,y_2\in\mathcal{Y}}W(y_2|v)\sqrt{W(y_1|u_1\ast v)W(y_1|u_1'\ast v)}\\
&= \frac{1}{|\mathcal{X}|}\sum_{y_1\in\mathcal{Y}}\sqrt{W(y_1|u_1\ast v)W(y_1|u_1'\ast v)}= \frac{1}{|\mathcal{X}|} Z(W_{u_1\ast v,u_1'\ast v}).
\end{align*}
\end{proof}

\begin{mylem}
\label{BhatPlus}
For every $u_2,u_2'\in\mathcal{X}$, we have $\displaystyle Z(W^{+}_{u_2,u_2'})=\frac{1}{|\mathcal{X}|}\sum_{u_1\in\mathcal{X}} Z(W_{u_1\ast u_2,u_1\ast u_2'})Z(W_{u_2,u_2'})$.
\end{mylem}
\begin{proof}
\begin{align*}
Z(W^+_{u_2,u_2'})&=\sum_{y_1,y_2\in\mathcal{Y}}\sum_{u_1\in\mathcal{X}} \sqrt{W^+(y_1,y_2,u_1|u_2)W^+(y_1,y_2,u_1|u_2')}\\
&=\sum_{y_1,y_2\in\mathcal{Y}}\sum_{u_1\in\mathcal{X}} \sqrt{\frac{1}{|\mathcal{X}|^2} W(y_1|u_1\ast u_2)W(y_2|u_2)W(y_1|u_1\ast u_2')W(y_2|u_2')}\\
&= \frac{1}{|\mathcal{X}|}\sum_{u_1\in\mathcal{X}} \sum_{y_1,y_2\in\mathcal{Y}}\sqrt{W(y_1|u_1\ast u_2)W(y_1|u_1\ast u_2')}\sqrt{W(y_2|u_2)W(y_2|u_2')}\\
&=\frac{1}{|\mathcal{X}|}\sum_{u_1\in\mathcal{X}} Z(W_{u_1\ast u_2,u_1\ast u_2'})Z(W_{u_2,u_2'}).
\end{align*}
\end{proof}

\begin{mynot}
If $W$ is a channel with input alphabet $\mathcal{X}$. We denote $\displaystyle \max_{\substack{x,x'\in\mathcal{X}\\x\neq x'}} Z(W_{x,x'})$ and $\displaystyle \min_{\substack{x,x'\in\mathcal{X}\\x\neq x'}} Z(W_{x,x'})$ by $Z_{\max}(W)$ and $Z_{\min}(W)$ respectively. Note that we can also express $Z_{\min}(W)$ as $\displaystyle \min_{\substack{x,x'\in\mathcal{X}}} Z(W_{x,x'})$ since $Z_{\min}(W)\leq 1$ and $Z_{x,x}(W)=1$ for every $x\in\mathcal{X}$.
\end{mynot}

\begin{myprop}
\label{propbadpol}
Let $\ast$ be a polarizing operation on $\mathcal{X}$, where $|\mathcal{X}|\geq 2$. If for every $u_2,u_2'\in\mathcal{X}$ there exists $u_1\in\mathcal{X}$ such that $u_1\ast u_2=u_1\ast u_2'$, then $E_\ast=0$.
\end{myprop}
\begin{proof}
Let $\beta>0$ and $0<\beta'<\beta$. Clearly, $\displaystyle\frac{1}{4}\left(2^{-2^{\beta'n}}\right)^2>2^{-2^{\beta n}}$ for $n$ large enough. We have:
\begin{itemize}
\item For every $u_2,u_2'\in\mathcal{X}$ satisfying $u_2\neq u_2'$, let $u_1\in\mathcal{X}$ be such that $u_1\ast u_2=u_1\ast u_2'$. Lemma \ref{BhatPlus} implies that $\displaystyle Z(W^{+}_{u_2,u_2'})\geq \frac{1}{|\mathcal{X}|}Z(W_{u_1\ast u_2,u_1\ast u_2'})Z(W_{u_2,u_2'})=\frac{1}{|\mathcal{X}|}Z(W_{u_2,u_2'})$ since $Z(W_{u_1\ast u_2,u_1\ast u_2'})=1$. Therefore, $\displaystyle Z_{\max}(W^+)=\max_{\substack{x,x'\in\mathcal{X}\\x\neq x'}} Z(W^{+}_{x,x'})\geq \frac{1}{|\mathcal{X}|} \max_{\substack{x,x'\in\mathcal{X}\\x\neq x'}} Z(W_{x,x'})=\frac{1}{|\mathcal{X}|} Z_{\max}(W)$.
\item By fixing $v\in\mathcal{X}$, Lemma \ref{BhatMinus} implies that $$Z_{\max}(W^-)= \max_{\substack{x,x'\in\mathcal{X}\\x\neq x'}} Z(W^{-}_{x,x'})\geq \frac{1}{|\mathcal{X}|} \max_{\substack{x,x'\in\mathcal{X}\\x\neq x'}} Z(W_{x\ast v,x'\ast v})\stackrel{(a)}{=} \frac{1}{|\mathcal{X}|} \max_{\substack{x,x'\in\mathcal{X}\\x\neq x'}} Z(W_{x,x'})=\frac{1}{|\mathcal{X}|} Z_{\max}(W),$$ where (a) follows from the fact that $\ast$ is uniformity preserving, which implies that $$\{(x\ast v,x'\ast v):\;x,x'\in\mathcal{X},\;x\neq x'\}=\{(x,x'):\;x,x'\in\mathcal{X},\;x\neq x'\}.$$
\end{itemize}
By induction on $n>0$, we conclude that for every $s\in\{-,+\}^n$ we have:
$$\displaystyle Z_{\max}(W^s)\geq \frac{1}{|\mathcal{X}|^n}Z_{\max}(W)=\frac{1}{2^{n\log_2|\mathcal{X}|}}Z_{\max}(W).$$ 

If $Z(W)>0$ we have $Z_{\max}(W)>0$, and $$\displaystyle Z(W^s)\geq \frac{1}{|\mathcal{X}|(|\mathcal{X}|-1)}Z_{\max}(W^s)\geq \frac{Z_{\max}(W)}{|\mathcal{X}|(|\mathcal{X}|-1)\cdot (2^{n})^{\log_2|\mathcal{X}|}},$$ which means that the decay of $Z(W^s)$ in terms of the blocklength $2^n$ can be at best polynomial. Therefore, for $n$ large enough we have $Z(W^s)>2^{-2^{\beta' n}}$ for every $s\in\{-,+\}^n$.

Now let $\delta=\frac{1}{3}\log|\mathcal{X}|-\frac{1}{3}\log(|\mathcal{X}|-1)>0$ and let $W$ be any channel satisfying $\log|\mathcal{X}|-\delta<I(W)<\log|\mathcal{X}|$ (we can easily construct such a channel). Since $I(W)<\log|\mathcal{X}|$, Proposition \ref{propBhat} implies that we have $Z(W)>0$. Let $W_n$ be the process introduced in Definition \ref{def1}. Since $\ast$ is polarizing, we have $\mathbb{P}[W_n\;\text{is}\;\delta\text{-easy}]>\frac{3}{4}$ (i.e., $\frac{1}{2^n}|\{s\in\{-,+\}^n:\; W^s\;\text{is}\;\delta\text{-easy}\}|>\frac{3}{4}$) for $n$ large enough. On the other hand, since $\ast$ satisfies the conservation property, we have $\displaystyle \mathbb{E}[I(W_n)]=\frac{1}{2^n}\sum_{s\in\{-,+\}^n}I(W^s)=\displaystyle I(W)>\log|\mathcal{X}|-\delta$. Therefore, we must have $\mathbb{P}\big[I(W_n)>\log|\mathcal{X}|-2\delta\big]>\frac{1}{2}$ and so for $n$ large enough, we have $$\mathbb{P}\big[I(W_n)>\log|\mathcal{X}|-2\delta\;\text{and}\;W_n\;\text{is}\;\delta\text{-easy}\big]>\frac{1}{4}.$$
Now suppose $s\in\{-,+\}^n$ is such that $W^s$ is $\delta$-easy and $I(W^s)>\log|\mathcal{X}|-2\delta$, and let $L$ and $\mathcal{B}$ be as in Definition \ref{defeasy}. We have $I(W^s)-\log(|\mathcal{X}|-1)>3\delta-2\delta=\delta$ and so the only possible value for $L$ is $|\mathcal{X}|$. But since the only subset of $\mathcal{X}$ of size $|\mathcal{X}|$ is $\mathcal{X}$, we have $\mathcal{B}=\mathcal{X}$ with probability 1. Therefore, $W^s_{\mathcal{B}}$ is equivalent to $W^s$ which means that $Z(W^s_{\mathcal{B}})=Z(W^s)>2^{-2^{\beta' n}}$. Now Proposition \ref{propBhat} implies that $\mathbb{P}_e(W^s_{\mathcal{B}})>\frac{1}{4}\left(2^{-2^{\beta'n}}\right)^2>2^{-2^{\beta n}}$ and so $W^s$ is not $(\delta,2^{-\beta n})$-easy. Thus, $\mathbb{P}\big[W_n\;\text{is}\;(\delta,2^{-2^{\beta n}})\text{-easy}\big]<\frac{3}{4}$ for $n$ large enough.

We conclude that no exponent $\beta>0$ is $\ast$-achievable. Therefore, $E_\ast=0$.
\end{proof}

\begin{myrem}
Consider the following uniformity preserving operation:
\begin{center}
  \begin{tabular}{ | c || c | c | c | c | }
    \hline
    $\ast$ & 0 & 1 & 2 & 3 \\ \hline \hline
    0 & 3 & 3 & 3 & 3 \\ \hline
    1 & 0 & 1 & 0 & 0 \\ \hline
    2 & 1 & 0 & 1 & 1 \\ \hline
    3 & 2 & 2 & 2 & 2 \\ \hline
  \end{tabular}
\end{center}
It is easy to see that $/^{\ast}$ is strongly ergodic on so $\ast$ is polarizing. Moreover, $\ast$ satisfies the property of Proposition \ref{propbadpol}, hence it has a zero exponent. This shows that the exponent of a polarizing operation can be as low as 0.
\end{myrem}

The following lemma will be used to show that $E_{\ast}\leq \frac{1}{2}$ for every polarizing operation $\ast$.

\begin{mylem}
\label{lemsss}
Let $\ast$ be a uniformity preserving operation on $\mathcal{X}$ and let $W$ be a channel with input alphabet $\mathcal{X}$. For every $n>0$ and every $s\in\{-,+\}^n$, we have $\displaystyle Z_{\min}(W^s)\geq \left(\frac{Z_{\min}(W)}{|\mathcal{X}|}\right)^{(|s|^-+1)2^{|s|^+}}$, where $|s|^-$ (resp. $|s|^+$) is the number of $-$ signs (resp. $+$ signs) in the sequence $s$.
\end{mylem}
\begin{proof}
We will prove the lemma by induction on $n>0$. If $n=1$, then either $s=-$ or $s=+$. If $s=-$, let $v\in\mathcal{X}$. We have:
\begin{equation}
Z_{\min}(W^s)=Z_{\min}(W^-)=\min_{u_1,u_1'\in\mathcal{X}}Z(W^-_{u_1,u_1'})\stackrel{(a)}{\geq} \min_{u_1,u_1'\in\mathcal{X}} \frac{1}{|\mathcal{X}|}Z(W_{u_1\ast v,u_1'\ast v})\stackrel{(b)}{\geq}\left(\frac{Z_{\min}(W)}{|\mathcal{X}|}\right)^{(|s|^-+1)2^{|s|^+}},\label{pouier}
\end{equation}
where (a) follows from Lemma \ref{BhatMinus} and (b) follows from the fact that $(|s|^-+1)2^{|s|^+}=2$ since $|s|^-=1$ and $|s|^+=0$ when $s=-$.

If $s=+$, we have:
\begin{align}
Z_{\min}(W^s)=Z_{\min}(W^+)=\min_{u_2,u_2'\in\mathcal{X}}Z(W^+_{u_2,u_2'})&\stackrel{(a)}{\geq} \min_{u_2,u_2'\in\mathcal{X}} \frac{1}{|\mathcal{X}|}\sum_{u_1\in\mathcal{X}} Z(W_{u_1\ast u_2,u_1\ast u_2'})Z(W_{u_2,u_2'})\nonumber\\
&\geq Z_{\min}(W)^2 \stackrel{(b)}{\geq} \left(\frac{Z_{\min}(W)}{|\mathcal{X}|}\right)^{(|s|^-+1)2^{|s|^+}}, \label{pouier1}
\end{align}
where (a) follows from Lemma \ref{BhatPlus} and (b) follows from the fact that $(|s|^-+1)2^{|s|^+}=2$ since $|s|^-=0$ and $|s|^+=1$ when $s=+$.
Therefore, the lemma is true for $n=1$. Now let $n>1$ and suppose that it is true for $n-1$. Let $s=(s',s_n)\in\{-,+\}^n$, where $s'\in\{-,+\}^{n-1}$ and $s_n\in\{-,+\}$. From the induction hypothesis, we have $\displaystyle Z_{\min}(W^{s'})\geq \left(\frac{Z_{\min}(W)}{|\mathcal{X}|}\right)^{(|s'|^-+1)2^{|s'|^+}}$. 

If $s_n=-$, we can apply \eqref{pouier} on $W^{s'}$ to get:
\begin{align*}
Z_{\min}(W^s)&\geq\frac{1}{|\mathcal{X}|}Z_{\min}(W^{s'})\geq \frac{1}{|\mathcal{X}|} \left(\frac{Z_{\min}(W)}{|\mathcal{X}|}\right)^{(|s'|^-+1)2^{|s'|^+}}\geq \left(\frac{Z_{\min}(W)}{|\mathcal{X}|}\right)^{1+(|s'|^-+1)2^{|s'|^+}}\\
&\geq \left(\frac{Z_{\min}(W)}{|\mathcal{X}|}\right)^{(|s'|^-+2)2^{|s'|^+}}=\left(\frac{Z_{\min}(W)}{|\mathcal{X}|}\right)^{(|s|^-+1)2^{|s|^+}}.
\end{align*}

If $s_n=+$, we can apply \eqref{pouier1} on $W^{s'}$ to get:
\begin{align*}
Z_{\min}(W^s)&\geq Z_{\min}(W^{s'})^2\geq \left(\left(\frac{Z_{\min}(W)}{|\mathcal{X}|}\right)^{(|s'|^-+1)2^{|s'|^+}}\right)^2= \left(\frac{Z_{\min}(W)}{|\mathcal{X}|}\right)^{2(|s'|^-+1)2^{|s'|^+}}\\
&= \left(\frac{Z_{\min}(W)}{|\mathcal{X}|}\right)^{(|s'|^-+1)2^{|s'|^++1}}=\left(\frac{Z_{\min}(W)}{|\mathcal{X}|}\right)^{(|s|^-+1)2^{|s|^+}}.
\end{align*}
We conclude that the lemma is true for every $n>0$.
\end{proof}

\begin{myprop}
\label{propexponentpol}
If $\ast$ is polarizing, then $E_\ast\leq \frac{1}{2}$.
\end{myprop}
\begin{proof}
Let $\beta>\frac{1}{2}$, and let $\frac{1}{2}< \beta'< \beta$.
Let $\epsilon>0$ be such that $(1-\epsilon)\log|\mathcal{X}|>\log|\mathcal{X}|-\delta$, where $\delta=\frac{1}{3}|\mathcal{X}|-\frac{1}{3}(|\mathcal{X}|-1)$. Let $e\notin \mathcal{X}$ and consider the channel $W:\mathcal{X}\longrightarrow\mathcal{X}\cup\{e\}$ defined as follows:
$$W(y|x)=\begin{cases}1-\epsilon\;&\text{if\;}y=x,\\\epsilon\;
&\text{if\;}y=e,\\0\;&\text{otherwise}.\end{cases}$$
We have $I(W)=(1-\epsilon)\log|\mathcal{X}|>\log|\mathcal{X}|-\delta$ and $Z(W_{x,x'})=\epsilon$ for every $x,x'\in\mathcal{X}$ such that $x\neq x'$, and thus $Z_{\min}(W)=\epsilon$. We have the following:
\begin{itemize}
\item Since $\beta'>\frac{1}{2}$, the law of large numbers implies that $\displaystyle\frac{1}{2^n}\left|\left\{s\in\{-,+\}^n:|s|^+\leq \beta'n\right\}\right|$ converges to 1 as $n$ goes to infinity. Therefore, for $n$ large enough, we have $\displaystyle\frac{1}{2^n}|B_n|>\frac{7}{8}$ where $$B_n=\left\{s\in\{-,+\}^n:|s|^+\leq \beta'n\right\}.$$
\item Since $\displaystyle\sum_{s\in\{-,+\}^n}I(W^s)=2^n I(W)>2^n(\log|\mathcal{X}|-\delta)$, we must have $\displaystyle\frac{1}{2^n}|C_n|>\frac{1}{2}$ where $$C_n=\big\{s\in\{-,+\}^n:I(W^s)>\log|\mathcal{X}|-2\delta\big\}.$$
\item Since $\ast$ is polarizing, we have $\displaystyle\frac{1}{2^n}|D_n|>\frac{7}{8}$ for $n$ large enough, where $$D_n=\big\{s\in\{-,+\}^n:\;W^s\;\text{is}\;\delta\text{-easy}\big\}.$$
\end{itemize}
We conclude that for $n$ large enough, we have $\displaystyle\frac{1}{2^n}|A_n|>\frac{1}{4}$, where
$$A_n=B_n\cap C_n\cap D_n=\big\{s\in\{-,+\}^n:\; |s|^+\leq \beta'n,\;W^s\;\text{is}\;\delta\text{-easy\;and\;}I(W^s)>\log|\mathcal{X}|-2\delta\big\}.$$
Now let $s\in A_n$. Let $L$ and $\mathcal{B}$ be as in Definition \ref{defeasy}. We have $I(W^s)-\log(|\mathcal{X}|-1)>3\delta-2\delta=\delta$ and so the only possible value for $L$ is $|\mathcal{X}|$, and since the only subset of $\mathcal{X}$ of size $|\mathcal{X}|$ is $\mathcal{X}$, we have $\mathcal{B}=\mathcal{X}$ with probability 1. Therefore, $W^s_{\mathcal{B}}$ is equivalent to $W^s$. Thus,
\begin{align*}
Z(W^s_{\mathcal{B}})=Z(W^s)\geq Z_{\min}(W^s)\stackrel{(a)}{\geq} \left(\frac{Z_{\min}(W)}{|\mathcal{X}|}\right)^{(|s|^-+1)2^{|s|^+}}\stackrel{(b)}{\geq} \left(\frac{\epsilon}{|\mathcal{X}|}\right)^{(n+1)2^{\beta'n}},
\end{align*}
where (a) follows from Lemma \ref{lemsss} and (b) follows from the fact that $|s|^-\leq n$ and $|s|^+\leq\beta'n$ for $s\in A_n$, and from the fact that $Z_{\min}(W)=\epsilon$ which was proved earlier. Now Proposition \ref{propBhat} implies that $\displaystyle \mathbb{P}_e(W^s_{\mathcal{B}})\geq\frac{1}{4}\left(\frac{\epsilon}{|\mathcal{X}|}\right)^{2(n+1)2^{\beta'n}}$. On the other hand, since $\beta'<\beta$, we have $\displaystyle \frac{1}{4}\left(\frac{\epsilon}{|\mathcal{X}|}\right)^{2(n+1)2^{\beta'n}}>2^{-2^{\beta n}}$ for $n$ large enough. Therefore, $W^s$ is not $(\delta,2^{-\beta n})$-easy if $s\in A_n$ and $n$ is large enough. Let $W_n$ be the process introduced in Definition \ref{def1}. For $n$ large enough, we have $$\displaystyle \mathbb{P}\big[W_n\;\text{is}\;(\delta,2^{-2^{\beta n}})\text{-easy}\big]\leq 1-\frac{1}{2^n}|A_n|<1-\frac{1}{4}=\frac{3}{4}.$$ 

We conclude that every exponent $\beta>\frac{1}{2}$ is not $\ast$-achievable. Therefore, $E_\ast\leq \frac{1}{2}$.
\end{proof}

\begin{mycor}
\label{corquasi}
If $\ast$ is a quasigroup operation, then $E_\ast = \frac{1}{2}$. 
\end{mycor}
\begin{proof}
The quasigroup-based polar code construction in \cite{RajTel} shows that every $\beta<\frac{1}{2}$ is a $\ast$-achievable exponent. Therefore, $E_\ast\geq \frac{1}{2}$. On the other hand, since $\ast$ is polarizing, Proposition \ref{propexponentpol} implies that $E_\ast\leq\frac{1}{2}$. Therefore, $E_\ast=\frac{1}{2}$.
\end{proof}

\begin{myconj}
\label{ConjConjConj}
If $\ast$ is a polarizing operation which is not a quasigroup operation, then $E_\ast<\frac{1}{2}$.
\end{myconj}

Conjecture \ref{ConjConjConj} implies that quasigroup operations are the best polarizing operations. Therefore, if the conjecture is true and we are looking for good polar codes with large blocklength, it is sufficient to consider quasigroup operations.

\section{Polarization theory for MACs}
\begin{mydef}
Let $W:\mathcal{X}_1\times\ldots\times\mathcal{X}_m\longrightarrow\mathcal{Y}$ be an $m$-user MAC. Let $\mathcal{X}=\mathcal{X}_1\times\ldots\times\mathcal{X}_m$. The single-user channel obtained from $W$ is the channel $W':\mathcal{X}\longrightarrow\mathcal{Y}$ defined by $W'\big(y\big|(x_1,\ldots,x_m)\big)=W(y|x_1,\ldots,x_m)$ for every $(x_1,\ldots,x_m)\in\mathcal{X}$.
\label{deffeddef}
\end{mydef}

\begin{mynot}
Let $W:\mathcal{X}_1\times\ldots\times\mathcal{X}_m\longrightarrow\mathcal{Y}$ be an $m$-user MAC. Let $\ast_1,\ldots,\ast_m$ be $m$ ergodic operations on $\mathcal{X}_1,\ldots,\mathcal{X}_m$ respectively, and let $\ast=\ast_1\otimes\ldots\otimes\ast_m$, which is an ergodic operation on $\mathcal{X}=\mathcal{X}_1\times\ldots\times\mathcal{X}_m$. Let $\mathcal{H}$ be a stable partition of $(\mathcal{X},\ast)$. $W[\mathcal{H}]$ denotes the single user channel $W'[\mathcal{H}]:\mathcal{H}\longrightarrow\mathcal{Y}$ (see Definition \ref{defprojchannel}), where $W'$ is the single user channel obtained from $W$.
\end{mynot}

\begin{mylem}
Let $W:\mathcal{X}_1\times\ldots\times\mathcal{X}_m\longrightarrow\mathcal{Y}$ be an $m$-user MAC. Let $\ast_1,\ldots,\ast_m$ be $m$ ergodic operations on $\mathcal{X}_1,\ldots,\mathcal{X}_m$ respectively, and let $\ast=\ast_1\otimes\ldots\otimes\ast_m$. If there exists $\delta>0$ and a stable partition $\mathcal{H}$ of $(\mathcal{X},\ast)$ such that $\big|I(W)-\log|\mathcal{H}|\big|<\delta$ and $\big|I(W[\mathcal{H}])-\log|\mathcal{H}|\big|<\delta$, then $W$ is a $\delta$-easy MAC. Moreover, if we also have $\mathbb{P}_e(W[\mathcal{H}])<\epsilon$, then $W$ is a $(\delta,\epsilon)$-easy MAC.
\label{lemeasyMAC}
\end{mylem}
\begin{proof}
Let $(\mathcal{H}_i)_{1\leq i\leq m}$ be the canonical factorization of $\mathcal{H}$ (see Definition \ref{defCanFac} of Part I \cite{RajErgI}). Let $L=|\mathcal{H}|$. For each $1\leq i\leq m$ let $L_i=|\mathcal{H}_i|$ and define $\mathcal{S}_i:=\{C_i\subset \mathcal{X}_i:\; |C_i|=L_i\}$. We have $L=L_1\cdots L_m$ (see Proposition \ref{PropProdProdProd} of Part I \cite{RajErgI}). Moreover, we have
\begin{equation}
\label{eqeqtataeaea1}
|I(W)-\log L|=\big|I(W)-\log |\mathcal{H}|\big|\leq \delta.
\end{equation}

Now for each $1\leq i\leq m$ let $H_{i,1},\ldots,H_{i,L_i}$ be the elements of $\mathcal{H}_i$, and for each $1\leq j\leq L_i$ let $X_{i,j}$ be a uniform random variable in $H_{i,j}$. We suppose that $X_{i,j}$ is independent from $X_{i',j'}$ for all $(i',j')\neq (i,j)$. Define $\mathcal{B}_i=\{X_{i,1},\ldots,X_{i,L_i}\}$ which is a random subset of $\mathcal{X}_i$. Clearly, $|\mathcal{B}_i|=L_i$ since each $X_{i,j}$ is drawn from a different element of $\mathcal{H}_i$. Therefore, $\mathcal{B}_i$ takes values in $\mathcal{S}_i$ and $\mathcal{B}_1,\ldots,\mathcal{B}_m$ are independent.

For each $1\leq i\leq m$ and each $x_i\in\mathcal{X}_i$, let $j$ be the unique index $1\leq j\leq L_i$ such that $x_i\in H_{i,j}$. Since we are sure that $x_i\notin H_{i,j'}$ for $j'\neq j$, then $x_i\in\mathcal{B}_i$ if and only if $X_{i,j}=x_i$. We have:
\begin{equation}
\label{eqeqtataeaea2}
\sum_{C_i\in\mathcal{S}_i} \frac{1}{L_i}\mathbb{P}_{\mathcal{B}_i}(C_i)\mathds{1}_{x_i\in C_i}=\frac{1}{L_i}\mathbb{P}[x_i\in\mathcal{B}_i]\stackrel{(a)}{=}\frac{1}{L_i}\mathbb{P}[X_{i,j}=x_i]=\frac{1}{L_i}\cdot\frac{1}{|H_{i,j}|}=\frac{1}{|\mathcal{H}_i|}\cdot\frac{1}{\|\mathcal{H}_i\|}=\frac{1}{|\mathcal{X}_i|},
\end{equation}
where (a) follows from the fact that $x_i\in\mathcal{B}_i$ if and only if $X_{i,j}=x_i$.

Now for each $1\leq i\leq m$ and each $C_i\subset \mathcal{S}_i$, let $f_{i,C_i}:\{1,\ldots,L_i\}\rightarrow C_i$ be a fixed bijection. Let $T_1,\ldots,T_m$ be $m$ independent random variables that are uniform in $\{1,\ldots,L_1\}$, \ldots, $\{1,\ldots,L_m\}$ respectively, and which are independent of $\mathcal{B}_1,\ldots,\mathcal{B}_m$. For each $1\leq i\leq m$, let $X_i=f_{i,\mathcal{B}_i}(T_i)$. Send $X_1,\ldots,X_m$ through the MAC $W$ and let $Y$ be the output. The MAC $T_1,\ldots,T_m\longrightarrow (Y,\mathcal{B}_1,\ldots,\mathcal{B}_m)$ is equivalent to the MAC $W_{\mathcal{B}_1,\ldots,\mathcal{B}_m}$ (see Definition \ref{defeasyMAC}). Our aim now is to show that $I(W_{\mathcal{B}_1,\ldots,\mathcal{B}_m})=I(T_1,\ldots,T_m;Y,\mathcal{B}_1,\ldots,\mathcal{B}_m)>\log L-\delta$, which will imply that $W$ is $\delta$-easy (see Definition \ref{defeasyMAC}).

We have $I(T_1,\ldots,T_m;Y,\mathcal{B}_1,\ldots,\mathcal{B}_m)=H(T_1,\ldots,T_m)-H(T_1,\ldots,T_m|Y,\mathcal{B}_1,\ldots,\mathcal{B}_m)$. Now since $H(T_1,\ldots,T_m)=H(T_1)+\ldots+H(T_m)=\log L_1 +\ldots + \log L_m=\log L$, it is sufficient to show that $H(T|Y,\mathcal{B})<\delta$, where $T=(T_1,\ldots,T_m)$ and $\mathcal{B}=(\mathcal{B}_1,\ldots,\mathcal{B}_m)\in \mathcal{S}_1\times\ldots\times\mathcal{S}_m$.

Now for each $1\leq i\leq m$ and each $x_i\in\mathcal{X}_i$, we have:
\begin{align*}
\mathbb{P}_{X_i}(x_i)&=\mathbb{P}[f_{i,\mathcal{B}_i}(T_i)=x_i]\stackrel{(a)}{=}\sum_{C_i\in\mathcal{S}_i:\; x_i\in C_i}\mathbb{P}[f_{i,C_i}(T_i)=x_i]\mathbb{P}_{\mathcal{B}_i}(C_i) \\
&\stackrel{(b)}{=} \sum_{C_i\in\mathcal{S}_i:\; x_i\in C_i}\frac{1}{L_i}\mathbb{P}_{\mathcal{B}_i}(C_i)=\sum_{C_i\in\mathcal{S}_i}\frac{1}{L_i}\mathbb{P}_{\mathcal{B}_i}(C_i)\mathds{1}_{x_i\in C_i}\stackrel{(c)}{=}\frac{1}{|\mathcal{X}_i|},
\end{align*}
where (a) follows from the fact that $f_{i,C_i}(T_i)\in C_i$ and so if $x_i\notin C_i$ then there is a probability of zero to have $f_{i,C_i}(T_i) = x_i$. (b) follows from the fact that $T_i$ is uniform in $\{1,\ldots,L_i\}$ and $f_{i,C_i}$ is a bijection from $\{1,\ldots,L_i\}$ to $C_i$ which imply that $f_{i,C_i}(T_i)$ is uniform in $C_i$ and so $\mathbb{P}[f_{i,C_i}(T_i)=x_i]=\frac{1}{|C_i|}=\frac{1}{L_i}$. (c) follows from Equation \eqref{eqeqtataeaea2}. Therefore, $X:=(X_1,\ldots,X_m)$ is uniform in $\mathcal{X}$ since $X_1,\ldots,X_m$ are independent and uniform in $\mathcal{X}_1,\ldots,\mathcal{X}_m$ respectively. This means that $$I(W[\mathcal{H}])=I(\proj_{\mathcal{H}}(X);Y)=H(\proj_{\mathcal{H}}(X))-H(\proj_{\mathcal{H}}(X)|Y)=\log|\mathcal{H}|-H(\proj_{\mathcal{H}}(X)|Y).$$
Moreover, we have $\big|I(W[\mathcal{H}])-\log|\mathcal{H}|\big|<\delta$ by hypothesis. We conclude that 
\begin{equation}
\label{eqydsnbhawq}
H(\proj_{\mathcal{H}}(X)|Y)<\delta.
\end{equation}

For each $1\leq i\leq m$, let $\mathcal{S}_{\mathcal{H}_i}=\big\{\{x_1,\ldots,x_{L_i}\}:\; x_j\in H_{i,j},\; \forall1\leq j\leq L_i\big\}$ be the set of sections of $\mathcal{H}_i$ (see Definition \ref{defSect} of Part I \cite{RajErgI}). By construction, $\mathcal{B}_i$ takes values in $\mathcal{S}_{\mathcal{H}_i}$. Now define $$\mathcal{S}_{\mathcal{H}}=\{C_1\times\ldots\times C_m:\; C_1\in\mathcal{S}_{\mathcal{H}_1},\ldots,C_m\in\mathcal{S}_{\mathcal{H}_m}\}.$$

For each $C= C_1\times\ldots\times C_m\in\mathcal{S}_{\mathcal{H}}$, define $f_C:\{1,\ldots,L_1\}\times\ldots\times\{1,\ldots,L_m\}\rightarrow \mathcal{H}$
as $$f_C(t_1,\ldots,t_m)=\proj_{\mathcal{H}}\big(f_{1,C_1}(t_1),\ldots,f_{m,C_m}(t_m)\big),$$
Since $C_1,\ldots,C_m$ are sections of $\mathcal{H}_1,\ldots,\mathcal{H}_m$ respectively, $C=C_1\times\ldots\times C_m$ is a section of $\mathcal{H}$ (see Proposition \ref{PropProdProdProd} of Part I \cite{RajErgI}). Therefore, for every $H\in\mathcal{H}$, there exists a unique $x=(x_1,\ldots,x_m)\in C$ such that $H=\proj_{\mathcal{H}}(x)$. This implies that there exist unique $t_1\in\{1,\ldots,L_1\}$, \ldots,$t_m\in\{1,\ldots,L_m\}$ such that $f_C(t_1,\ldots,t_m)=H$. Therefore, $f_C$ is a bijection from $\{1,\ldots,L_1\}\times\ldots\times\{1,\ldots,L_m\}$ to $\mathcal{H}$.

Now since $f_C$ is a bijection for every $C\in \mathcal{S}_{\mathcal{H}}$ and since $\mathcal{B}_1\times\ldots\times\mathcal{B}_m$ takes values in $\mathcal{S}_\mathcal{H}$, we have 
\begin{align*}
H(T|Y,\mathcal{B})&=H\big(f_{\mathcal{B}_1\times\ldots\times\mathcal{B}_m}(T)\big|Y,\mathcal{B}\big)=H\big(\proj_{\mathcal{H}}\big(f_{1,\mathcal{B}_1}(T_1),\ldots,f_{m,\mathcal{B}_m}(T_m)\big)\big|Y,\mathcal{B}\big)\\
&=H\big(\proj_{\mathcal{H}}(X_1,\ldots,X_m)\big|Y,\mathcal{B}\big)=H(\proj_{\mathcal{H}}(X)|Y,\mathcal{B})\leq H(\proj_{\mathcal{H}}(X)|Y)\stackrel{(a)}{<}\delta
\end{align*}
as required, where (a) follows from \eqref{eqydsnbhawq}. We conclude that $W$ is $\delta$-easy.

Now suppose that we also have $\mathbb{P}_e(W[\mathcal{H}])<\epsilon$. Consider the following decoder for the MAC $W_{\mathcal{B}}=W_{\mathcal{B}_1,\ldots,\mathcal{B}_m}$:
\begin{itemize}
\item Compute an estimate $\hat{H}$ of $\proj_{\mathcal{H}}(X)$ using the ML decoder of the channel $W[\mathcal{H}]$.
\item Compute $\hat{T}=f_{\mathcal{B}_1\times\ldots\times\mathcal{B}_m}^{-1}(\hat{H})$.
\end{itemize}
The probability of error of this decoder is:
\begin{align*}
\mathbb{P}[\hat{T}\neq T]&=\mathbb{P}[\hat{H}\neq f_{\mathcal{B}_1\times\ldots\times\mathcal{B}_m}(T)]=\mathbb{P}\big[\hat{H}\neq \proj_{\mathcal{H}}\big(f_{1,\mathcal{B}_1}(T_1),\ldots,f_{m,\mathcal{B}_m}(T_m)\big)\big]\\
&=\mathbb{P}[\hat{H}\neq \proj_{\mathcal{H}}(X_1,\ldots,X_m)]=\mathbb{P}[\hat{H}\neq \proj_{\mathcal{H}}(X)]=\mathbb{P}_e(W[\mathcal{H}])<\epsilon.
\end{align*}
Now since the ML decoder of $W_{\mathcal{B}}$ minimizes the probability of error, we conclude that $\mathbb{P}_e(W_{\mathcal{B}})<\epsilon$. Therefore, $W$ is a $(\delta,\epsilon)$-easy MAC.
\end{proof}

\begin{mythe}
Let $\ast_1,\ldots,\ast_m$ be $m$ binary operations on $\mathcal{X}_1,\ldots,\mathcal{X}_m$ respectively. The sequence $(\ast_1,\ldots,\ast_m)$ is MAC-polarizing if and only if $\ast_1,\ldots,\ast_m$ are polarizing.
\end{mythe}
\begin{proof}
Suppose that $(\ast_1,\ldots,\ast_m)$ is MAC-polarizing. By Remark \ref{remnecmac}, $\ast_1,\ldots,\ast_m$ are polarizing.

Conversely, suppose that $\ast_1,\ldots,\ast_m$ are polarizing. Theorem \ref{thecaraccharac} implies that $\ast_1,\ldots,\ast_m$ are uniformity preserving and $/^{\ast_1},\ldots,/^{\ast_m}$ are strongly ergodic. Now Theorem \ref{theprodstrong} of Part I \cite{RajErgI} implies that the binary operation $/^{\ast_1}\otimes\ldots\otimes/^{\ast_m}$ is strongly ergodic. By noticing that $/^{\ast_1\otimes\ldots\otimes\ast_m}=/^{\ast_1}\otimes\ldots \otimes/^{\ast_m}$, we conclude that $/^{\ast}$ is strongly ergodic, where $\ast=\ast_1\otimes\ldots\otimes\ast_m$.

Now let $W:\mathcal{X}_1\times\ldots\times\mathcal{X}_m\longrightarrow\mathcal{Y}$ be an $m$-user MAC. Let $\mathcal{X}=\mathcal{X}_1\times\ldots\times\mathcal{X}_m$ and let $W':\mathcal{X}\longrightarrow\mathcal{Y}$ be the single user channel obtained from $W$ (see Definition \ref{deffeddef}).

For each $n>0$ and each $s\in\{-,+\}^n$, let $W'^s$ be obtained from $W'$ using the operation $\ast$ (see Definition \ref{defdef11}), and let $W^s$ be obtained from $W$ using the operations $\ast_1,\ldots,\ast_m$ (see Definition \ref{defdef11MAC}). Now since $/^{\ast}$ is strongly ergodic, then by Corollary \ref{cor1}, for any $\delta>0$ we have:
\begin{align*}
\lim_{n\to\infty} \frac{1}{2^n} \bigg|\Big\{ s\in\{-,+\}^n:\;&\exists \mathcal{H}_s\; \emph{a stable partition of $(\mathcal{X},/^{\ast})$},\\
&\big| I(W'^s)-\log|\mathcal{H}_s|\big|<\delta, \big| I(W'^s[\mathcal{H}_s])-\log|\mathcal{H}_s|\big|<\delta \Big\}\bigg| = 1.
\end{align*}
It is easy to see that $W'^s$ is the single user channel obtained from $W^s$. Therefore, $I(W^s)=I(W'^s)$ and $I(W^s[\mathcal{H}])=I(W'^s[\mathcal{H}])$ (by definition). Therefore,
\begin{align*}
\lim_{n\to\infty} \frac{1}{2^n} \bigg|\Big\{ s\in\{-,+\}^n:\;&\exists \mathcal{H}_s\; \emph{a stable partition of $(\mathcal{X},/^{\ast})$},\\
&\big| I(W^s)-\log|\mathcal{H}_s|\big|<\delta, \big| I(W^s[\mathcal{H}_s])-\log|\mathcal{H}_s|\big|<\delta \Big\}\bigg| = 1.
\end{align*}
Now Lemma \ref{lemeasyMAC}, applied to $/^{\ast_1},\ldots,/^{\ast_m}$, implies that:
\begin{align*}
\lim_{n\to\infty} \frac{1}{2^n} \big|\big\{ s\in\{-,+\}^n:\; W^s\;\text{is\;}\delta\text{-easy}\big\}\big| = 1.
\end{align*}
Therefore, $(\ast_1,\ldots,\ast_m)$ satisfies the polarization property of Definition \ref{defPolaMac}. On the other hand, since $\ast_1,\ldots,\ast_m$ are uniformity preserving, Remark \ref{remPreservMac} implies that $(\ast_1,\ldots,\ast_m)$ satisfies the conservation property of Definition \ref{defPolaMac}. We conclude that $(\ast_1,\ldots,\ast_m)$ is MAC-polarizing.
\end{proof}

\begin{myprop}
Let $\ast_1,\ldots,\ast_m$ be $m$ binary operations on $\mathcal{X}_1,\ldots,\mathcal{X}_m$ respectively. If $(\ast_1,\ldots,\ast_m)$ is MAC-polarizing, then $E_{\ast_1,\ldots,\ast_m}\leq E_{\ast_1\otimes\ldots\otimes\ast_m}\leq\min\{E_{\ast_1},\ldots,E_{\ast_m}\}\leq\frac{1}{2}$.
\label{propExponentMAC}
\end{myprop}
\begin{proof} 
Define $\ast=\ast_1\otimes\ldots\otimes\ast_m$. Let $W:\mathcal{X}_1\times\ldots\times\mathcal{X}_m\longrightarrow\mathcal{Y}$ be an $m$-user MAC and let $W':\mathcal{X}\longrightarrow\mathcal{Y}$ be the single user channel obtained from $W$. Note that every MAC polar code for the MAC $W$ constructed using $(\ast_1,\ldots,\ast_m)$ can be seen as a polar code for the channel $W'$ constructed using the operation $\ast$. Moreover, the probability of error of the ML decoder is the same. Therefore, every $(\ast_1,\ldots,\ast_m)$-achievable exponent is $\ast$-achievable. Hence, $E_{\ast_1,\ldots,\ast_m}\leq E_{\ast}$.

Now let $\mathcal{X}=\mathcal{X}_1\times\ldots\times\mathcal{X}_m$. For each $1\leq i\leq m$ and each single user channel $W_i:\mathcal{X}_i\longrightarrow\mathcal{Y}$ with input alphabet $\mathcal{X}_i$, consider the single user channel $W:\mathcal{X}\longrightarrow\mathcal{Y}$ with input alphabet $\mathcal{X}$ defined as $W\big(y\big|(x_1,\ldots,x_m)\big)=W_i(y|x_i)$. Let $(W_{i,n})_{n\geq 0}$ be the single user channel valued process obtained from $W_i$ using the operation $\ast_i$ as in Definition \ref{def1}, and let $(W_n)_{n\geq 0}$ be the single user channel valued process obtained from $W$ using the operation $\ast$ as in Definition \ref{def1}. It is easy to see that for every $\delta>0$ and every $\epsilon>0$, $W_{i,n}$ is $(\delta,\epsilon)$-easy if and only if $W_n$ is $(\delta,\epsilon)$-easy. This implies that each $\ast$-achievable exponent is $\ast_i$-achievable. Therefore, $E_{\ast}\leq E_{\ast_i}$ for every $1\leq i\leq m$, hence $E_{\ast}\leq \min\{E_{\ast_1},\ldots,E_{\ast_m}\}$. Now from Proposition \ref{propexponentpol}, we have $\min\{E_{\ast_1},\ldots,E_{\ast_m}\}\leq\frac{1}{2}$.
\end{proof}

\begin{myprop}
If $\ast_1,\ldots,\ast_m$ are quasigroup operations, then $E_{\ast_1,\ldots,\ast_m}=\frac{1}{2}$.
\end{myprop}
\begin{proof}
Let $\ast=\ast_1\otimes\ldots\otimes\ast_m$, then $\ast$ is a quasigroup operation. Let $\beta<\beta'<\frac{1}{2}$. Let $W:\mathcal{X}_1\times\ldots\times\mathcal{X}_m\longrightarrow\mathcal{Y}$ be an $m$-user MAC. Define $\mathcal{X}=\mathcal{X}_1\times\ldots\times\mathcal{X}_m$ and let $W':\mathcal{X}\longrightarrow\mathcal{Y}$ be the single user channel obtained from $W$. For each $n>0$ and each $s\in\{-,+\}^n$, let $W'^s$ be obtained from $W'$ using the operation $\ast$ (see Definition \ref{defdef11}), and let $W^s$ be obtained from $W$ using the operations $\ast_1,\ldots,\ast_m$ (see Definition \ref{defdef11MAC}). From Theorem 4 of \cite{RajTelA}, we have:
\begin{align*}
\lim_{n\to\infty} \frac{1}{2^n} \bigg|\Big\{ s\in\{-&,+\}^n:\;\exists \mathcal{H}_s\; \text{a stable partition of $(\mathcal{X},/^{\ast})$},\\
&\big| I(W'^s)-\log|\mathcal{H}_s|\big|<\delta,\; \big| I(W'^s[\mathcal{H}_s])-\log|\mathcal{H}_s|\big|<\delta,\; Z(W'^s[\mathcal{H}_s])<2^{-2^{\beta'n}}\Big\}\bigg| = 1.
\end{align*}
On the other hand, we have $\mathbb{P}_e(W'^s[\mathcal{H}_s])\leq(|\mathcal{H}_s|-1)Z(W'^s[\mathcal{H}_s])\leq(|\mathcal{X}|-1)Z(W'^s[\mathcal{H}_s])$ from Proposition \ref{propBhat}. Therefore,
\begin{align*}
\lim_{n\to\infty} \frac{1}{2^n} \bigg|\Big\{ &s\in\{-,+\}^n:\;\exists \mathcal{H}_s\; \text{a stable partition of $(\mathcal{X},/^{\ast})$},\\
&\big| I(W'^s)-\log|\mathcal{H}_s|\big|<\delta,\; \big| I(W'^s[\mathcal{H}_s])-\log|\mathcal{H}_s|\big|<\delta,\; \mathbb{P}_e(W'^s[\mathcal{H}_s])<(|\mathcal{X}|-1)2^{-2^{\beta'n}}\Big\}\bigg| = 1.
\end{align*}
It is easy to see that $W'^s$ is the single user channel obtained from $W^s$. Therefore, $I(W^s)=I(W'^s)$, $I(W^s[\mathcal{H}_s])=I(W'^s[\mathcal{H}_s])$ (by definition) and $\mathbb{P}_e(W^s[\mathcal{H}_s])=\mathbb{P}_e(W'^s[\mathcal{H}_s])$. On the other hand, we have $(|\mathcal{X}|-1)2^{-2^{\beta'n}}<2^{-2^{\beta n}}$ for $n$ large enough. We conclude that:
\begin{align*}
\lim_{n\to\infty} \frac{1}{2^n} \bigg|\Big\{ s\in\{-&,+\}^n:\;\exists \mathcal{H}_s\; \text{a stable partition of $(\mathcal{X},/^{\ast})$},\\
&\big| I(W^s)-\log|\mathcal{H}_s|\big|<\delta, \big| I(W^s[\mathcal{H}_s])-\log|\mathcal{H}_s|\big|<\delta,\; \mathbb{P}_e(W^s[\mathcal{H}_s])<2^{-2^{\beta n}}\Big\}\bigg| = 1.
\end{align*}
Now since $/^{\ast}=/^{\ast_1}\otimes\ldots\otimes/^{\ast_m}$ and since $/^{\ast_i}$ is ergodic (as it is a quasigroup operation) for every $1\leq i\leq m$, Lemma \ref{lemeasyMAC} implies that:
$$\lim_{n\to\infty} \frac{1}{2^n}\big|\big\{s\in\{-,+\}^n:\; W^s\;\text{is}\;(\delta,2^{-2^{\beta n}})\text{-easy}\big\}\big|=1.$$
We conclude that every $0\leq \beta<\frac{1}{2}$ is a $(\ast_1,\ldots,\ast_m)$-achievable exponent. Therefore, $E_{\ast_1,\ldots,\ast_m}\geq\frac{1}{2}$. On the other hand, we have $E_{\ast_1,\ldots,\ast_m}\leq\frac{1}{2}$ from Proposition \ref{propExponentMAC}. Hence $E_{\ast_1,\ldots,\ast_m}=\frac{1}{2}$.
\end{proof}

\begin{mycor}
For every $\delta>0$, every $\beta<\frac{1}{2}$, every MAC $W:\mathcal{X}_1\times\ldots\times\mathcal{X}_m\longrightarrow\mathcal{Y}$, and every quasigroup operations $\ast_1,\ldots,\ast_m$ on $\mathcal{X}_1,\ldots,\mathcal{X}_m$ respectively, there exists a polar code for the MAC $W$ constructed using $\ast_1,\ldots,\ast_m$ such that its sum-rate is at least $I(W)-\delta$ and its probability of error under successive cancellation decoder is less than $2^{-N^\beta}$, where $N=2^n$ is the blocklength.
\end{mycor}

\section{Conclusion}

A complete characterization of polarizing operations is provided and it is shown that the exponent of polarizing operations cannot exceed $\frac{1}{2}$. Therefore, if we wish to construct polar codes that have a better exponent, we have to use other Ar{\i}kan style constructions that are not based on binary operations. Korada et. al. showed that it is possible to achieve exponents that exceed $\frac{1}{2}$ by combining more than two channels at each polarization step \cite{KoradaSasUrb}.

The transformation used in \cite{KoradaSasUrb} as kernel is linear. Presman et. al. showed that nonlinear kernels can achieve strictly better exponents than linear kernels \cite{Presman}. An important problem, which remains open, is to find a characterization of all polarizing transformations in the general non-linear case. A generalization of the ergodic theory of binary operations that we developed in Part I \cite{RajErgI} is likely to provide such a characterization.

\appendices

\section{Proof of Proposition \ref{MainProp}}

\label{apppA}

Let $(X_i,Y_i)_{0\leq i< 2^k}$ be a sequence of $2^k$ random pairs that satisfy conditions 1) and 2) of Proposition \ref{MainProp}.

\begin{mynot}
For every sequence $\mathbf{x}=(x_i)_{1\leq i<2^k}$ of $2^{k}-1$ elements of $\mathcal{X}$, define the mapping $\pi_{\mathbf{x}}:\mathcal{X}\rightarrow\mathcal{X}$ as $\pi_{\mathbf{x}}(x_0)=g_{\ast}\big((x_0,\mathbf{x})\big)$ for all $x_0\in\mathcal{X}$, where $(x_0,\mathbf{x})$ is the sequence of $2^k$ elements obtained by concatenating $x_0$ and $\mathbf{x}$. Note that $\pi_{\mathbf{x}}$ is a bijection since $\ast$ is uniformity preserving. Define:
\begin{itemize}
\item $p_{y}(x):=\mathbb{P}_{X_0|Y_0}(x|y)$ for every $x\in\mathcal{X}$ and every $y\in\mathcal{Y}$. Note that $p_{y}(x)=\mathbb{P}_{X_i|Y_i}(x|y)$ for every $0\leq i< 2^k$ since $(X_i,Y_i)$ and $(X_0,Y_0)$ are identically distributed.
\item $p_{y_0,\mathbf{x}}(x):=p_{y_0}\big(\pi_{\mathbf{x}}^{-1}(x)\big)$ for every $x\in\mathcal{X}$, every $y_0\in\mathcal{Y}$ and every sequence $\mathbf{x}=(x_i)_{1\leq i<2^k}\in\mathcal{X}^{2^k-1}$.
\item For every $\mathbf{x}=(x_i)_{1\leq i<2^k}\in\mathcal{X}^{2^k-1}$, and every $y_1^{2^k-1}=(y_i)_{1\leq i<2^k}\in\mathcal{Y}^{2^k-1}$, define $$p_{y_{1}^{2^k-1}}(\mathbf{x}):=\prod_{i=1}^{2^k-1}p_{y_i}(x_i)=\mathbb{P}_{X_1^{2^k-1}|Y_1^{2^k-1}}(\mathbf{x}|y_1^{2^k-1}).$$
\end{itemize}
\end{mynot}

Fix $\gamma>0$ and let $\displaystyle\gamma'=\min\left\{\frac{\gamma}{2^{|\mathcal{X}|}+1},\frac{1}{(2^{|\mathcal{X}|}+2)|\mathcal{X}|}\right\}$.

\begin{mynot}
Define:
\begin{align*}
\mathcal{C}=\Big\{y_0^{2^k-1} \in \mathcal{Y}^{2^k}:\; &\forall \mathbf{x}\in \mathcal{X}^{2^k-1},\forall\mathbf{x}'\in \mathcal{X}^{2^k-1},\\
&\big(p_{y_{1}^{2^k-1}}(\mathbf{x})\geq \gamma'^{2^k-1}\text{\; and\; }p_{y_{1}^{2^k-1}}(\mathbf{x}')\geq \gamma'^{2^k-1}\big)\Rightarrow \|p_{y_0,\mathbf{x}}-p_{y_0,\mathbf{x}'}\|_{\infty}< \gamma' \Big\}.
\end{align*}
\end{mynot}

\begin{mylem}
\label{Lemma3}
There exists $\epsilon(\gamma)>0$ such that if $H\big(g_\ast(X_0^{2^k-1})\big|Y_0^{2^k-1}\big)<H(X_0|Y_0)+\epsilon(\gamma)$, then $$\mathbb{P}_{Y_0^{2^k-1}}(\mathcal{C})>1-\gamma'^{2^k}.$$
\end{mylem}
\begin{proof}
For every $x\in\mathcal{X}$ and every $y_{0}^{2^k-1}\in\mathcal{Y}^{2^k}$, we have:
\begin{equation*}
\begin{aligned}
\mathbb{P}_{g_{\ast}(X_0^{2^k-1})|Y_0^{2^k-1}}&(x|y_0^{2^k-1})\\
&=\sum_{\substack{x_0,\ldots,x_{2^k-1}\in\mathcal{X}:\\ g_\ast(x_0^{2^{k}-1})=x}}\left(\prod_{i=0}^{2^k-1}p_{y_i}(x_i)\right)= \sum_{\substack{\mathbf{x}\in\mathcal{X}^{2^k-1}, \\\mathbf{x}=(x_i)_{1\leq i<2^k}}}\sum_{\substack{x_0\in\mathcal{X}:\\ g_{\ast}((x_0,\mathbf{x}))=x}}\left(\prod_{i=1}^{2^k-1}p_{y_i}(x_i)\right)p_{y_0}(x_0)\\
&= \sum_{\mathbf{x}\in\mathcal{X}^{2^k-1}}p_{y_{1}^{2^k-1}}(\mathbf{x})\sum_{\substack{x_0\in\mathcal{X}:\\ \pi_{\mathbf{x}}(x_0)=x}}p_{y_0}(x_0)= \sum_{\mathbf{x}\in\mathcal{X}^{2^k-1}}p_{y_{1}^{2^k-1}}(\mathbf{x})p_{y_0}\big(\pi_{\mathbf{x}}^{-1}(x)\big)\\
&=\sum_{\mathbf{x}\in\mathcal{X}^{2^k-1}}p_{y_{1}^{2^k-1}}(\mathbf{x})p_{y_0,\mathbf{x}}(x).
\end{aligned}
\end{equation*}
Therefore, for every $y_{0}^{2^k-1}\in\mathcal{Y}^{2^k}$ we have:
\begin{equation}
\label{Eqeq1}
\mathbb{P}_{g_{\ast}(X_0^{2^k-1})|Y_0^{2^k-1}}(x|y_0^{2^k-1})=\sum_{\mathbf{x}\in\mathcal{X}^{2^k-1}} p_{y_{1}^{2^k-1}}(\mathbf{x})p_{y_0,\mathbf{x}}(x).
\end{equation}

Due to the concavity of the entropy function, it follows from \eqref{Eqeq1} that for every sequence $y_0^{2^k-1}\in\mathcal{Y}^{2^k}$ we have:

\begin{align}
\nonumber H\big(g_\ast(X_0^{2^k-1}) &\big|Y_0^{2^k-1}=y_0^{2^k-1}\big)\\
&\geq 
\sum_{\mathbf{x}\in\mathcal{X}^{2^k-1}}p_{y_1^{2^k-1}}(\mathbf{x}) H(p_{y_0,\mathbf{x}})\stackrel{(a)}{=} \sum_{\mathbf{x}\in\mathcal{X}^{2^k-1}}p_{y_1^{2^k-1}}(\mathbf{x}) H(p_{y_0})=H(p_{y_0})=H(X_0|Y_0=y_0),
\label{Eqeqp} 
\end{align}
where (a) follows from the fact that the distribution $p_{y_0,\mathbf{x}}$ is a permuted version of the distribution $p_{y_0}$, which implies that $p_{y_0,\mathbf{x}}$ and $p_{y_0}$ have the same entropy. Now if $y_0^{2^k-1}\in\mathcal{C}^c$, there exist $\mathbf{x}\in\mathcal{X}^{2^k-1}$ and $\mathbf{x}'\in\mathcal{X}^{2^k-1}$ such that $p_{y_1^{2^k-1}}(\mathbf{x})\geq \gamma'^{2^k-1}$, $p_{y_1^{2^k-1}}(\mathbf{x}')\geq \gamma'^{2^k-1}$ and $\|p_{y_0,\mathbf{x}}-p_{y_0,\mathbf{x}'}\|_{\infty}\geq \gamma'$. Therefore, due to the strict concavity of the entropy function, it follows from \eqref{Eqeq1} that there exists $\epsilon'(\gamma')>0$ such that:
\begin{equation}
\label{Eqeq4} H\big(g_\ast(X_0^{2^k-1}) \big|Y_0^{2^k-1}=y_0^{2^k-1}\big)\geq \Big(\sum_{\mathbf{x}\in\mathcal{X}^{2^k-1}}p_{y_1^{2^k-1}}(\mathbf{x}) H(p_{y_0,\mathbf{x}})\Big)+\epsilon'(\gamma')= H(X_0|Y_0=y_0)+\epsilon'(\gamma').
\end{equation}
Moreover, since the space of probability distributions on $\mathcal{X}$ is compact, $\epsilon'(\gamma')>0$ can be chosen so that it depends only on $\gamma'$ and $|\mathcal{X}|$. We have:
\begin{align*}
H\big(&g_\ast(X_0^{2^k-1})\big|Y_0^{2^k-1}\big)\\
&=\sum_{y_0^{2^k-1}\in\mathcal{C}}H\big(g_\ast(X_0^{2^k-1})\big|Y_0^{2^k-1}=y_0^{2^k-1}\big)\mathbb{P}_{Y_0^{2^k-1}}(y_0^{2^k-1}) \\
&\;\;\;\;\;\;\;\;\;\;\;\;\;\;\;\;\;\;\;\;\;\;\;\;\;\;\;\;\;\;\;\;\;\;\;\;\; \;\;\;\;\;\;\;\;\;\;\;\;\;\;\;\;+ \sum_{y_0^{2^k-1}\in\mathcal{C}^c}H\big(g_\ast(X_0^{2^k-1})\big|Y_0^{2^k-1}=y_0^{2^k-1}\big)\mathbb{P}_{Y_0^{2^k-1}}(y_0^{2^k-1})\\
&\stackrel{(a)}{\geq} \sum_{y_0^{2^k-1}\in\mathcal{C}}H(X_0|Y_0=y_0)\mathbb{P}_{Y_0^{2^k-1}}(y_0^{2^k-1}) + \sum_{y_0^{2^k-1}\in\mathcal{C}^c}\big(H(X_0|Y_0=y_0)+\epsilon'(\gamma')\big)\mathbb{P}_{Y_0^{2^k-1}}(y_0^{2^k-1}) \\
&=\Big(\sum_{y_0^{2^k-1}\in\mathcal{Y}^{2^k-1}}H(X_0|Y_0=y_0)\mathbb{P}_{Y_0^{2^k-1}}(y_0^{2^k-1})\Big) +\epsilon'(\gamma')\mathbb{P}_{Y_0^{2^k-1}}(\mathcal{C}^c)= H(X_0|Y_0)+\epsilon'(\gamma')\mathbb{P}_{Y_0^{2^k-1}}(\mathcal{C}^c),
\end{align*}
where (a) follows from \eqref{Eqeqp} and \eqref{Eqeq4}. Let $\epsilon(\gamma)=\epsilon'(\gamma')\gamma'^{2^k}$.

Clearly, if $H\big(g_\ast(X_0^{2^k-1})\big|Y_0^{2^k-1}\big)<H(X_0|Y_0)+\epsilon(\gamma)$, then we must have $\mathbb{P}_{Y_0^{2^k-1}}(\mathcal{C}^c)<\gamma'^{2^k}$.
\end{proof}

\vspace*{3mm}

In the next few definitions and lemmas, $(X_i,Y_i)_{0\leq i< 2^k}$ is a sequence of $2^k$ random pairs that satisfy conditions 1), 2) and 3) of Proposition \ref{MainProp} where $\epsilon(\gamma)$ is as in Lemma \ref{Lemma3}. In particular, we have $H\big(g_\ast(X_0^{2^k-1})\big|Y_0^{2^k-1}\big)<H(X_0|Y_0)+\epsilon(\gamma)$ and so by Lemma \ref{Lemma3} we have $\mathbb{P}_{Y_0^{2^k-1}}(\mathcal{C})>1-\gamma'^{2^k}$, where $\displaystyle\gamma'=\min\left\{\frac{\gamma}{2^{|\mathcal{X}|}+1},\frac{1}{(2^{|\mathcal{X}|}+2)|\mathcal{X}|}\right\}$.

\begin{mynot}
Define the following:
\begin{itemize}
\item For each $y_0\in\mathcal{Y}$, let $\mathcal{C}_{y_0}:=\big\{y_1^{2^k-1}\in\mathcal{Y}^{2^k-1}:\; y_0^{2^k-1}\in\mathcal{C}\big\}$.
\item $\mathcal{C}_0:=\big\{y_0\in\mathcal{Y}:\; \mathbb{P}_{Y_1^{2^k-1}}(\mathcal{C}_{y_0})> 1-\gamma'^{2^k-1}\big\}$.
\item For each $y\in \mathcal{Y}$, let $A_y=\{x\in \mathcal{X}:\; p_{y}(x)\geq \gamma'\}$.
\item For each $D\subset \mathcal{X}$, let $\mathcal{Y}_D=\{y\in\mathcal{Y}:\;A_y=D\}$.
\item $\mathcal{A}=\{D_0\subset \mathcal{X}:\; \mathbb{P}_{Y_0}(\mathcal{Y}_{D_0})\geq \gamma'\}$.
\end{itemize}
\end{mynot}

We will show later that $\mathcal{A}$ is actually the stable partition $\mathcal{H}$ of $(\mathcal{X},\ast)$ that is claimed in Proposition \ref{MainProp}.

\begin{mylem}
\label{lemtemlemtem1}
We have:
\begin{itemize}
\item $\mathbb{P}_{Y_0}(\mathcal{C}_0)>1-\gamma'$.
\item For every $D_0\in\mathcal{A}$ there exists $y_0\in \mathcal{C}_0$ such that $A_{y_0}=D_0$.
\end{itemize}
\end{mylem}
\begin{proof}
We have 
\begin{align*}
1-\gamma'^{2^k}&<\mathbb{P}_{Y_0^{2^k-1}}(\mathcal{C})=\sum_{y_0\in\mathcal{C}_0}\mathbb{P}_{Y_0}(y_0)\mathbb{P}_{Y_1^{2^k-1}}(\mathcal{C}_{y_0})+ \sum_{y_0\in\mathcal{C}_0^c}\mathbb{P}_{Y_0}(y_0)\mathbb{P}_{Y_1^{2^k-1}}(\mathcal{C}_{y_0}) \\
&\stackrel{(a)}{\leq} \mathbb{P}_{Y_0}(\mathcal{C}_0)+\mathbb{P}_{Y_0}(\mathcal{C}_0^c)(1-\gamma'^{2^k-1})=1-\gamma'^{2^k-1}\mathbb{P}_{Y_0}(\mathcal{C}_0^c),
\end{align*}
where (a) follows from the fact that $\mathbb{P}_{Y_1^{2^k-1}}(\mathcal{C}_{y_0})\leq 1-\gamma'^{2^k-1}$ for every $y_0\in\mathcal{C}_0^c$.
We conclude that $\mathbb{P}_{Y_0}(\mathcal{C}_0^c)<\gamma'$, hence $\mathbb{P}_{Y_0}(\mathcal{C}_0)>1-\gamma'$.

Now let $D_0\in\mathcal{A}$. We have $\mathbb{P}_{Y_0}(\mathcal{Y}_{D_0})\geq \gamma'$ by definition. But we have just shown that $\mathbb{P}_{Y_0}(\mathcal{C}_0)>1-\gamma'$, hence $1\geq \mathbb{P}_{Y_0}(\mathcal{Y}_{D_0}\cup \mathcal{C}_0)=\mathbb{P}_{Y_0}(\mathcal{Y}_{D_0})+\mathbb{P}_{Y_0}(\mathcal{C}_0)-\mathbb{P}_{Y_0}(\mathcal{Y}_{D_0}\cap \mathcal{C}_0)>\gamma'+1-\gamma'-\mathbb{P}_{Y_0}(\mathcal{Y}_{D_0}\cap \mathcal{C}_0)$, thus $\mathbb{P}_{Y_0}(\mathcal{Y}_{D_0}\cap \mathcal{C}_0)>0$. This implies that $\mathcal{Y}_{D_0}\cap \mathcal{C}_0\neq \o$. Therefore, there exists $y_0\in\mathcal{C}_0$ such that $A_{y_0}=D_0$.
\end{proof}

\begin{mylem}
\label{lemtemlemtem2}
$\mathcal{A}$ is an $\mathcal{X}$-cover.
\end{mylem}
\begin{proof}
For every $y_0\in\mathcal{Y}$, let $\displaystyle a_{y_0}=\operatorname*{arg\,max}_{x} p_{y_0}(x)$. Clearly, $\mathbb{P}_{Y_0}(a_{y_0})\geq\frac{1}{|\mathcal{X}|}>\gamma'$. Therefore, $a_{y_0}\in A_{y_0}$ and so $A_{y_0}\neq\o$ for every $y_0\in\mathcal{Y}$. This means that $\mathcal{Y}_{\o}=\o$, hence $\mathbb{P}_{Y_0}(\mathcal{Y}_{\o})=0<\gamma'$. We conclude that $\o\notin \mathcal{A}$.

Suppose that $\mathcal{A}$ is not an $\mathcal{X}$-cover. This means that $\displaystyle\bigcup_{D_0\in\mathcal{A}}D_0\neq \mathcal{X}$. Therefore, there exists $x_0\in\mathcal{X}$ such that $\displaystyle x_0\notin\bigcup_{D_0\in\mathcal{A}} D_0$ and so $x_0\notin D_0$ for every $D_0\in\mathcal{A}$. We have:
\begin{align*}
\frac{1}{|\mathcal{X}|}&=\mathbb{P}_{X_0}(x_0)=\sum_{y_0\in\mathcal{Y}}\mathbb{P}_{Y_0}(y_0)p_{y_0}(x_0)\stackrel{(a)}{=}
\sum_{D_0\subset \mathcal{X}}\sum_{\substack{y_0\in\mathcal{Y}_{D_0}}}\mathbb{P}_{Y_0}(y_0)p_{y_0}(x_0)\\
&=\sum_{D_0\in\mathcal{A}}\sum_{\substack{y_0\in\mathcal{Y}_{D_0}}} \mathbb{P}_{Y_0}(y_0)p_{y_0}(x_0) + \sum_{\substack{D_0\subset \mathcal{X}\\D_0\notin\mathcal{A}}}\sum_{\substack{y_0\in \mathcal{Y}_{D_0}}} \mathbb{P}_{Y_0}(y_0)p_{y_0}(x_0)\\
&\stackrel{(b)}{\leq} \sum_{D_0\in\mathcal{A}}\sum_{\substack{y_0\in\mathcal{Y}_{D_0}}} \mathbb{P}_{Y_0}(y_0)\gamma' + \sum_{\substack{D_0\subset \mathcal{X}\\D_0\notin\mathcal{A}}}\sum_{\substack{y_0\in \mathcal{Y}_{D_0}}} \mathbb{P}_{Y_0}(y_0)= \sum_{D_0\in\mathcal{A}} \mathbb{P}_{Y_0}(\mathcal{Y}_{D_0})\gamma' + \sum_{\substack{D_0\subset \mathcal{X}\\D_0\notin\mathcal{A}}} \mathbb{P}_{Y_0}(\mathcal{Y}_{D_0})\\
&\stackrel{(c)}{\leq} \mathbb{P}_{Y_0}\Big(\bigcup_{D_0\in\mathcal{A}}\mathcal{Y}_{D_0}\Big) \gamma' + \sum_{\substack{D_0\subset \mathcal{X}\\D_0\notin\mathcal{A}}} \gamma'\stackrel{(d)}{\leq} \gamma' + 2^{|\mathcal{X}|}\gamma'\leq (2^{|\mathcal{X}|}+1)\frac{1}{(2^{|\mathcal{X}|}+2)|\mathcal{X}|}<\frac{1}{|\mathcal{X}|},
\end{align*}
where (a) follows from the fact that $\{\mathcal{Y}_{D_0}:\;D_0\subset\mathcal{X}\}$ is a partition of $\mathcal{Y}$. (b) follows from the fact that if $D_0\in \mathcal{A}$ and $y_0\in\mathcal{Y}_{D_0}$, then $A_{y_0}=D_0\in\mathcal{A}$ and so $x_0\notin A_{y_0}$ (since $x_0\notin D_0$ for every $D_0\in\mathcal{A}$) which implies that $p_{y_0}(x_0)<\gamma'$. (c) follows from the fact that $\{\mathcal{Y}_{D_0}:\;D_0\subset\mathcal{X}\}$ is a partition of $\mathcal{Y}$ and from the fact that $\mathbb{P}_{Y_0}(\mathcal{Y}_{D_0})<\gamma'$ for every $D_0\notin \mathcal{A}$. (d) follows from the fact that there are at most $2^{|\mathcal{X}|}$ subsets of $\mathcal{X}$. We conclude that if $\mathcal{A}$ is not an $\mathcal{X}$-cover, then $\frac{1}{|\mathcal{X}|}<\frac{1}{|\mathcal{X}|}$ which is a contradiction. Therefore, $\mathcal{A}$ is an $\mathcal{X}$-cover.
\end{proof}

\vspace*{3mm}
The next three lemmas will be used to show that $\mathcal{A}$ is a stable partition.

\begin{mylem}
\label{propinduced1}
Let $k = 2^{2^{|\mathcal{X}|}}+\scon(\ast)$. For every $x\in\mathcal{X}$ there exists a sequence $\mathfrak{X}=(X_i)_{0\leq i<k}$ of length $k$ such that $X_i\in \mathcal{A}^{i\ast}$ for every $0\leq i<k$, and $x\ast\mathfrak{X}\in\mathcal{A}^{k\ast}$.
\end{mylem}
\begin{proof}
Since $\mathcal{A}$ is an $\mathcal{X}$-cover, we can apply Theorem \ref{theinduced} of Part I \cite{RajErgI}. Therefore, there exists $0\leq n<2^{2^{|\mathcal{X}|}}$ such that $\mathcal{A}^{n\ast}=\langle\mathcal{A}\rangle$. Fix $x\in \mathcal{X}$ and $X\in\mathcal{A}^{k\ast}=\langle\mathcal{A}\rangle^{(k-n)\ast}$, and let $A\in\mathcal{A}$ be such that $x\in A$. Choose an arbitrary sequence $\mathfrak{X}_1=(X_i)_{0\leq i<n}$ such that $X_i\in\mathcal{A}^{i\ast}$ for $0\leq i<n$. Let $B=A\ast \mathfrak{X}_1$. Clearly, $B\in \mathcal{A}^{n\ast}=\langle \mathcal{A}\rangle$.

Since $k = 2^{2^{|\mathcal{X}|}}+\scon(\ast)$ and $0\leq n<2^{2^{|\mathcal{X}|}}$, we have $k-n>\scon(\ast)$. Let $x'\in x\ast\mathfrak{X}_1$. Since $k-n>\scon(\ast)$, we can apply Theorem \ref{thestrong} of Part I \cite{RajErgI} to get a sequence $\mathfrak{X}_2=(X_i')_{0\leq i<k-n}$ such that $X_i'\in\langle\mathcal{A}\rangle^{i\ast}=\mathcal{A}^{(n+i)\ast}$ for every $0\leq i<k-n$, and $x'\ast\mathfrak{X}_2=X$. Since $x'\in x\ast\mathfrak{X}_1\subset A\ast\mathfrak{X}_1= B$, we have $X=x'\ast\mathfrak{X}_2\subset (x\ast\mathfrak{X}_1)\ast\mathfrak{X}_2 \subset B\ast\mathfrak{X}_2$. But both $X$ and $B\ast\mathfrak{X}_2$ are elements of $\langle \mathcal{A}\rangle^{(k-n)\ast}$ which is a partition, so we must have $B\ast\mathfrak{X}_2=X$. Now define $\mathfrak{X}=(\mathfrak{X}_1,\mathfrak{X}_2)$. We have $X=x'\ast\mathfrak{X}_2\subset x\ast\mathfrak{X}\subset B\ast\mathfrak{X}_2=X$. Therefore, $x\ast\mathfrak{X} = X\in\mathcal{A}^{k\ast}$.
\end{proof}

\begin{mylem}
For every $i\geq 0$ and every $X\in\mathcal{A}^{i\ast}$ there exist $2^i$ sets $B_0,\ldots,B_{2^i-1}\in\mathcal{A}$ such that
$$\displaystyle X=\bigg\{g_\ast(\mathbf{x}):\; \mathbf{x}\in\prod_{j=0}^{2^i-1}B_j\bigg\}:=\left\{g_\ast(x_0,\ldots,x_{2^i-1}):\;x_0\in B_0,\ldots,x_{2^i-1}\in B_{2^i-1}\right\}.$$
\label{lemlemtemyem}
\end{mylem}
\begin{proof}
We will show the lemma by induction on $i\geq0$. The lemma is trivial for $i=0$: Take $B_0=X\in\mathcal{A}$, we get 
$$\displaystyle X=\{x:\;x\in B_0\}=\bigg\{g_\ast(\mathbf{x}):\; \mathbf{x}\in\prod_{j=0}^{2^0-1}B_j\bigg\}.$$
Now let $i>0$ and suppose that the lemma is true for $i-1$. Let $X\in\mathcal{A}^{i\ast}$, and let $X',X''\in\mathcal{A}^{(i-1)\ast}$ be such that $X=X'\ast X''$. It follows from the induction hypothesis that there exist $2^{i-1}$ sets $B_0',\ldots,B_{2^{i-1}-1}'\in\mathcal{A}$ and $2^{i-1}$ sets $B_0'',\ldots,B_{2^{i-1}-1}''\in\mathcal{A}$ such that
$$\displaystyle X'=\bigg\{g_\ast(\mathbf{x}'):\; \mathbf{x}'\in\prod_{j=0}^{2^{i-1}-1}B_j'\bigg\}\;\;\;\text{and}\;\;\;X''=\bigg\{g_\ast(\mathbf{x}''):\; \mathbf{x}''\in\prod_{j=0}^{2^{i-1}-1}B_j''\bigg\}.$$
We have:
\begin{align*}
X&=X'\ast X''=\bigg\{g_\ast(\mathbf{x}'):\; \mathbf{x}'\in\prod_{j=0}^{2^{i-1}-1}B_j'\bigg\}\ast \bigg\{g_\ast(\mathbf{x}''):\; \mathbf{x}''\in\prod_{j=0}^{2^{i-1}-1}B_j''\bigg\}\\
&=\bigg\{g_\ast(\mathbf{x}')\ast g_\ast(\mathbf{x}''):\; \mathbf{x}'\in\prod_{j=0}^{2^{i-1}-1}B_j',\; \mathbf{x}''\in\prod_{j=0}^{2^{i-1}-1}B_j''\bigg\}\\
&=\bigg\{ g_\ast(\mathbf{x}):\; \mathbf{x}\in\bigg(\prod_{j=0}^{2^{i-1}-1}B_j'\bigg)\times\bigg( \prod_{j=0}^{2^{i-1}-1}B_j''\bigg)\bigg\}=\bigg\{ g_\ast(\mathbf{x}):\; \mathbf{x}\in \prod_{j=0}^{2^{i}-1}B_j\bigg\},
\end{align*}
where
$$B_j=\begin{cases}B_j'\;&\text{if}\;0\leq j<2^{i-1},\\ B_{j-2^{i-1}}''&\text{if}\;2^{i-1}\leq j< 2^i.\end{cases}$$
\end{proof}

\begin{mylem}
\label{lemtemlemtembem}
Let $\mathfrak{X}=(X_i)_{0\leq i<l}$ be a sequence of length $l>0$ such that $X_i\in\mathcal{A}^{i\ast}$ for every $0\leq i<l$. There exist $2^l-1$ sets $D_1,\ldots,D_{2^l-1}\in\mathcal{A}$ such that for every $x\in\mathcal{X}$, we have
$$x\ast \mathfrak{X} =\bigg\{g_{\ast}\big((x,\mathbf{x})\big):\;  \mathbf{x}\in\prod_{i=1}^{2^l-1}D_i\bigg\}=\bigg\{\pi_{\mathbf{x}}(x):\;  \mathbf{x}\in\prod_{i=1}^{2^l-1}D_i\bigg\}.$$
\end{mylem}
\begin{proof}
We will show the lemma by induction on $l>0$. If $l=1$, the lemma is trivial: If we take $D_1=X_0\in\mathcal{A}$, then for every $x\in\mathcal{X}$ we have
$$x\ast \mathfrak{X}=\{x\ast x_0: x_0\in X_0\}=\big\{g_{\ast}\big((x,x_0)\big): x_0\in D_1\big\}=\bigg\{g_{\ast}\big((x,\mathbf{x})\big):\;  \mathbf{x}\in\prod_{i=1}^{2^1-1}D_i\bigg\}.$$
Now let $l>1$ and suppose that the lemma is true for $l-1$. Let $\mathfrak{X}=(X_i)_{0\leq i<l}$ and define the sequence $\mathfrak{X}'=(X_i)_{0\leq i<l-1}$. The induction hypothesis implies the existence of $2^{l-1}-1$ sets $D_1',\ldots,D_{2^{l-1}-1}'\in\mathcal{A}$ such that for every $x\in\mathcal{X}$ we have
$$x\ast \mathfrak{X}'=\bigg\{g_{\ast}\big((x,\mathbf{x}')\big):\;  \mathbf{x}'\in\prod_{i=1}^{2^{l-1}-1}D_i'\bigg\}.$$
On the other hand, since $X_{l-1}\in\mathcal{A}^{(l-1)\ast}$, Lemma \ref{lemlemtemyem} shows the existence of $2^{l-1}$ sets $D_0'',\ldots,D_{2^{l-1}-1}''\in\mathcal{A}$ such that
$$\displaystyle X_{l-1}=\bigg\{g_\ast(\mathbf{x}''):\; \mathbf{x}''\in\prod_{i=0}^{2^{l-1}-1}D_i''\bigg\}.$$
Define the $2^l-1$ sets $D_1,\ldots,D_{2^l-1}\in\mathcal{A}$ as follows:
$$D_i=\begin{cases}D_i'\;&\text{if}\; 1\leq i<2^{l-1},\\D_{i-2^{l-1}}''\;&\text{if}\; 2^{l-1}\leq i<2^l.\end{cases}$$
For every $x\in\mathcal{X}$ we have:
\begin{align*}
x\ast\mathfrak{X} &= (x\ast\mathfrak{X}')\ast X_{l-1}=\bigg\{g_{\ast}\big((x,\mathbf{x}')\big):\;  \mathbf{x}'\in\prod_{i=1}^{2^{l-1}-1}D_i'\bigg\}\ast \bigg\{g_\ast(\mathbf{x}''):\; \mathbf{x}''\in\prod_{i=0}^{2^{l-1}-1}D_i''\bigg\}\\
&=\bigg\{g_\ast\big((x,\mathbf{x}')\big)\ast g_\ast(\mathbf{x}''):\; \mathbf{x}'\in\prod_{i=1}^{2^{l-1}-1}D_i,\; \mathbf{x}''\in\prod_{i=2^{l-1}}^{2^l-1}D_i\bigg\}= \bigg\{g_{\ast}\big((x,\mathbf{x})\big):\;  \mathbf{x}\in\prod_{i=1}^{2^l-1}D_i\bigg\}.
\end{align*}
\end{proof}

\begin{mylem}
\label{lemMainlem}
We have the following:
\begin{enumerate}
\item $\mathcal{A}$ is a stable partition of $(\mathcal{X},\ast)$.
\item If $y_0\in\mathcal{C}_{0}$ and $A_{y_0}\in\mathcal{A}$ then $y_0\in\mathcal{Y}_{\mathcal{A},\gamma}(X_0,Y_0)$.
\end{enumerate}
\end{mylem}
\begin{proof}
1) Let $D_0\in\mathcal{A}$. By Lemma \ref{lemtemlemtem1}, there exists $y_0\in \mathcal{C}_0$ such that $D_0=A_{y_0}$. Let $\displaystyle a_{y_0}=\operatorname*{arg\,max}_{x\in\mathcal{X}} p_{y_0}(x)$. Clearly, $p_{y_0}(a_{y_0})\geq \frac{1}{|\mathcal{X}|}>\gamma'$ and so $a_{y_0}\in A_{y_0}=D_0$.

Since $\mathcal{A}$ is an $\mathcal{X}$-cover (Lemma \ref{lemtemlemtem2}), Theorem \ref{theinduced} of Part I \cite{RajErgI} implies the existence of an integer $n$ satisfying $0\leq n<2^{2^{|\mathcal{X}|}}$ and $\mathcal{A}^{n\ast}=\langle\mathcal{A}\rangle$. Moreover, Lemma \ref{propinduced1} shows the existence of a sequence $\mathfrak{X}=(X_i)_{0\leq i<k}$ such that $X_i\in\mathcal{A}^{i\ast}$ for all $0\leq i< k$ and $a_{y_0}\ast\mathfrak{X}\in\mathcal{A}^{k\ast}= \langle\mathcal{A}\rangle^{(k-n)\ast}$. Let
\begin{equation}
\label{eqteqleq123r1}
B=a_{y_0}\ast\mathfrak{X}\in\mathcal{A}^{k\ast}=\langle\mathcal{A}\rangle ^{(k-n)\ast}.
\end{equation}

Lemma \ref{lemtemlemtembem} shows the existence of $2^k-1$ sets $D_1,\ldots,D_{2^k-1}\in\mathcal{A}$ such that
\begin{equation}
B=\Big\{g_{\ast}\big((a_{y_0},\mathbf{x})\big):\;  \mathbf{x}\in\prod_{i=1}^{2^k-1}D_i\Big\}=\Big\{\pi_{\mathbf{x}}(a_{y_0}):\;  \mathbf{x}\in\prod_{i=1}^{2^k-1}D_i\Big\}.
\label{equBd}
\end{equation}

Define $$\mathcal{C}_{y_0}'=\big\{y_1^{2^k-1} \in \mathcal{Y}^{2^k-1}:\; \forall 1\leq i< 2^k,\;A_{y_i}=D_i\big\}=\prod_{i=1}^{2^k-1}\mathcal{Y}_{D_i}.$$

Since $D_1,\ldots,D_{2^k-1}\in\mathcal{A}$, we have $$\mathbb{P}_{Y_1^{2^k-1}}(\mathcal{C}_{y_0}')=\prod_{i=1}^{2^k-1}\mathbb{P}_{Y_i}(\mathcal{Y}_{D_i})=\prod_{i=1}^{2^k-1}\mathbb{P}_{Y_0}(\mathcal{Y}_{D_i})\geq \gamma'^{2^k-1}.$$

On the other hand, since $y_0\in\mathcal{C}_0$, we have $\mathbb{P}_{Y_1^{2^k-1}}(\mathcal{C}_{y_0})>1-\gamma'^{2^k-1}$ from the definition of $\mathcal{C}_0$. Therefore, $\mathbb{P}_{Y_1^{2^k-1}}(\mathcal{C}_{y_0}\cap \mathcal{C}_{y_0}')>0$ which implies that $\mathcal{C}_{y_0}\cap \mathcal{C}_{y_0}'\neq \o$. Hence, there exists a sequence $(y_1,\ldots,y_{2^k-1})\in\mathcal{C}_{y_0}$ such that $A_{y_i}=D_i$ for all $1\leq i<2^k$.

Now fix a sequence
\begin{equation}
\label{eqteqleq123r2}
\mathbf{x}'=(x_i')_{1\leq i<2^k}\text{\;such\;that\;}x_i'\in D_i\text{\;for\;all\;}1\leq i<2^k.
\end{equation}

Let $x\in \pi_{\mathbf{x}'}^{-1}(B)$, then there exists $x'\in B$ such that $x'=\pi_{\mathbf{x}'}(x)$. Now from \eqref{equBd}, since $x'\in B$, there exists a sequence $\mathbf{x}=(x_i)_{1\leq i<2^k}$ such that $x_i\in D_i$ for all $1\leq i<2^k$ and $x'=\pi_{\mathbf{x}}(a_{y_0})$. We have:
\begin{itemize}
\item $(y_i)_{0\leq i<2^k}\in\mathcal{C}$ since $(y_1,\ldots,y_{2^k-1})\in\mathcal{C}_{y_0}$.
\item For every $1\leq i< 2^k$, we have $p_{y_i}(x_i)\geq \gamma'$ and $p_{y_i}(x_i')\geq \gamma'$ since $x_i,x_i'\in D_i=A_{y_i}$. Therefore, $\displaystyle p_{y_1^{2^k-1}}(\mathbf{x})=\prod_{i=1}^{2^k-1}p_{y_i}(x_i)\geq\gamma'^{2^k-1}$ and similarly $p_{y_1^{2^k-1}}(\mathbf{x}')\geq \gamma'^{2^k-1}$.
\end{itemize} 
From the definition of $\mathcal{C}$, we get $\|p_{y_0,\mathbf{x}}-p_{y_0,\mathbf{x}'}\|_{\infty}<\gamma'$ which implies that $|p_{y_0,\mathbf{x}}(x')-p_{y_0,\mathbf{x}'}(x')|<\gamma'$. Therefore, $$|p_{y_0}(a_{y_0})-p_{y_0}(x)|=\big|p_{y_0}\big(\pi_{\mathbf{x}}^{-1}(x')\big)-p_{y_0}\big(\pi_{\mathbf{x}'}^{-1}(x')\big)\big|=|p_{y_0,\mathbf{x}}(x')-p_{y_0,\mathbf{x}'}(x')|<\gamma'.$$
We conclude that
\begin{equation}
\forall x\in\pi_{\mathbf{x}'}^{-1}(B),\; |p_{y_0}(a_{y_0})-p_{y_0}(x)|<\gamma',
\label{eqlalaeqlala}
\end{equation}
and so $p_{y_0}(x)>p_{y_0}(a_{y_0})-\gamma' \geq \frac{1}{|\mathcal{X}|}-\gamma' \geq \frac{1}{|\mathcal{X}|}-\frac{1}{|\mathcal{X}|(2^{|\mathcal{X}|}+2)}> \frac{1}{|\mathcal{X}|}-\frac{1}{2|\mathcal{X}|} =\frac{1}{2|\mathcal{X}|}>\frac{1}{(2^{|\mathcal{X}|}+2)|\mathcal{X}|}\geq \gamma'$	 which implies that $x\in A_{y_0}=D_0$. But this is true for every $x\in \pi_{\mathbf{x}'}^{-1}(B)$. We conclude that $\pi_{\mathbf{x}'}^{-1}(B)\subset D_0$. On the other hand, since $D_0\in\mathcal{A}\preceq\langle\mathcal{A}\rangle$, there exists $C\in \langle\mathcal{A}\rangle$ such that $D_0\subset C$. Therefore, $\pi_{\mathbf{x}'}^{-1}(B)\subset D_0\subset C$ and 
$$\|\langle\mathcal{A}\rangle\|=\|\langle\mathcal{A}\rangle^{(k-n)\ast}\|=|B|\stackrel{(a)}{=}|\pi_{\mathbf{x}'}^{-1}(B)|\leq |D_0|\leq |C|=\|\langle\mathcal{A}\rangle\|,$$ where (a) follows from the fact that $\pi_{\mathbf{x}'}$ is a bijection. We conclude that $\|\langle\mathcal{A}\rangle\|=|\pi_{\mathbf{x}'}^{-1}(B)|=|D_0|=|C|$. But  $\pi_{\mathbf{x}'}^{-1}(B)\subset D_0\subset C$, so we must have

\begin{equation}
\label{eqteqleq123r3}
\pi_{\mathbf{x}'}^{-1}(B)=D_0=C\in \langle\mathcal{A}\rangle.
\end{equation}

Now since this is true for every $D_0\in\mathcal{A}$, we conclude that $\mathcal{A}\subset\langle\mathcal{A}\rangle$. But $\mathcal{A}$ is an $\mathcal{X}$-cover and $\langle\mathcal{A}\rangle$ is a partition, so we must have $\mathcal{A}=\langle\mathcal{A}\rangle$. We conclude that $\mathcal{A}$ is a stable partition.

\vspace*{3mm}
2) Let $y_0\in\mathcal{C}_{y_0}$ and suppose that $D_0=A_{y_0}\in\mathcal{A}$. Define $\displaystyle a_{y_0}=\operatorname*{arg\,max}_{x\in\mathcal{X}} p_{y_0}(x)$. Let $B\in \mathcal{A}^{k\ast}$ and $\mathbf{x}'\in\mathcal{X}^{2^k-1}$ be defined as in equations \eqref{eqteqleq123r1} and \eqref{eqteqleq123r2} respectively. Equation \eqref{eqteqleq123r3} shows that $D_0=\pi_{\mathbf{x}'}^{-1}(B)$. By replacing $\pi_{\mathbf{x}'}^{-1}(B)$ by $D_0$ in equation \eqref{eqlalaeqlala}, we conclude that for every $x\in D_0$ we have $|p_{y_0}(a_{y_0})-p_{y_0}(x)|<\gamma'$, which means that 
\begin{equation}
p_{y_0}(a_{y_0})-\gamma'< p_{y_0}(x)< p_{y_0}(a_{y_0})+\gamma'.
\label{eqlala1}
\end{equation}
On the other hand, for every $x\in\mathcal{X}\setminus D_0=\mathcal{X}\setminus A_{y_0}$, we have 
\begin{equation}
0\leq p_{y_0}(x)< \gamma'.
\label{eqlala2}
\end{equation}
By adding up the inequalities \eqref{eqlala1} for all $x\in D_0$ with the inequalities \eqref{eqlala2} for all $x\in \mathcal{X}\setminus D_0$, we get $|D_0|\cdot p_{y_0}(a_{y_0})-|D_0|\cdot \gamma'< 1 <|D_0|\cdot p_{y_0}(a_{y_0})+|\mathcal{X}|\cdot \gamma'$, from which we get $|p_{y_0}(a_{y_0})-\frac{1}{|D_0|}|< \frac{|\mathcal{X}|}{|D_0|}\gamma'\leq |\mathcal{X}|\gamma'$. We conclude that for every $x\in D_0$, we have $$\left|p_{y_0}(x)-\frac{1}{|D_0|}\right|\leq |p_{y_0}(x)-p_{y_0}(a_{y_0})|+\left|p_{y_0}(a_{y_0})-\frac{1}{|D_0|}\right|< \gamma'+|\mathcal{X}|\gamma'< (2^{|\mathcal{X}|}+1)\gamma'\leq\gamma,$$ and for every $x\in\mathcal{X}\setminus D_0=\mathcal{X}\setminus A_{y_0}$, we have $p_{y_0}(x)<\gamma'< \gamma$. Therefore, $\|p_{y_0}-\mathbb{I}_{D_0}\|_{\infty}\leq \gamma$ and so $y_0\in\mathcal{Y}_{\mathcal{A},\gamma}(X_0,Y_0)$.
\end{proof}

\vspace*{3mm}
Now we are ready to prove Proposition \ref{MainProp}:
\begin{proof}[Proof of Proposition \ref{MainProp}]
According to Lemma \ref{lemMainlem}, $\mathcal{A}$ is a stable partition. Moreover, for every $y_0\in\mathcal{C}_0$ satisfying $A_{y_0}\in\mathcal{A}$, we have $y_0\in\mathcal{Y}_{\mathcal{A},\gamma}(X_0,Y_0)$. Therefore, if we define $$\mathcal{Y}_{\mathcal{A}}'=\big\{y\in\mathcal{Y}:\; A_y\in \mathcal{A}\big\},$$ then $\mathcal{Y}_{\mathcal{A}}'\cap\mathcal{C}_0\subset \mathcal{Y}_{\mathcal{A},\gamma}(X_0,Y_0)$.

We have $\displaystyle \mathcal{Y}_{\mathcal{A}}'^c=\bigcup_{\substack{D\subset \mathcal{X}\\D\notin\mathcal{A}}}\mathcal{Y}_{D}$. Now since $\mathbb{P}_{Y_0}(\mathcal{Y}_{D})<\gamma'$ for every $D\notin\mathcal{A}$, we have:

$$\mathbb{P}_{Y_0}(\mathcal{Y}_{\mathcal{A}}'^c)\leq \sum_{\substack{D\subset \mathcal{X}\\D\notin\mathcal{A}}}\mathbb{P}_{Y_0} (\mathcal{Y}_{D})< 2^{|\mathcal{X}|}\gamma'.$$

But $\mathbb{P}_{Y_0}(\mathcal{C}_0)>1-\gamma'$ by Lemma \ref{lemtemlemtem1}, so we have $\mathbb{P}_{Y_0}(\mathcal{Y}_{\mathcal{A}}'\cap\mathcal{C}_0)>1-(2^{|\mathcal{X}|}+1)\gamma'\geq 1-\gamma$, which implies that $\mathbb{P}_{Y_0}(\mathcal{Y}_{\mathcal{A},\gamma}(X_0,Y_0))>1-\gamma$ since $\mathcal{Y}_{\mathcal{A}}'\cap\mathcal{C}_0\subset \mathcal{Y}_{\mathcal{A},\gamma}(X_0,Y_0)$. By letting $\mathcal{H}=\mathcal{A}$, which is a stable partition, we get $\mathcal{P}_{\mathcal{H},\gamma}(X_0,Y_0)=\mathbb{P}_{Y_0}(\mathcal{Y}_{\mathcal{H},\gamma}(X_0,Y_0))>1-\gamma$.
\end{proof}

\section{Proof of Proposition \ref{propBhat}}
\label{appB}
Inequalities 1) and 2) are proved in Proposition 3.3 of \cite{SasogluThesis}, and the upper bound of 3) is shown in Proposition 3.2 of \cite{SasogluThesis}. It remains to show the lower bound of 3).

Let $D^{\mathrm{ML}}_{W}:\mathcal{Y}\rightarrow\mathcal{X}$ be the ML decoder of the channel $W:\mathcal{X}\longrightarrow\mathcal{Y}$. I.e., for every $y\in\mathcal{Y}$, $\displaystyle D^{\mathrm{ML}}_{W}(y)=\operatorname*{arg\,max}_{x\in\mathcal{X}} W(y|x)$. For every $x\in\mathcal{X}$, let $\mathbb{P}_{e,x}(W)$ be the probability of error of $D^{\mathrm{ML}}_{W}$ given that $x$ was sent through $W$. Clearly, $\displaystyle\mathbb{P}_e(W)=\frac{1}{|\mathcal{X}|}\sum_{x\in\mathcal{X}}\mathbb{P}_{e,x}(W)$.

Now fix $x,x'\in\mathcal{X}$ such that $x\neq x'$ and define $\mathbb{P}_{e,x,x'}(W):=\frac{1}{2}\mathbb{P}_{e,x}(W)+\frac{1}{2}\mathbb{P}_{e,x'}(W)$. Consider the channel $W_{x,x'}:\{0,1\}\longrightarrow\mathcal{Y}$. We can use $D^{\mathrm{ML}}_{W}$ to construct a decoder for $W_{x,x'}$ as follows:
\begin{itemize}
\item If $D^{\mathrm{ML}}_{W}(y)=x$, the decoder output is $0$.
\item If $D^{\mathrm{ML}}_{W}(y)=x'$, the decoder output is $1$.
\item If $D^{\mathrm{ML}}_{W}(y)\notin\{x,x'\}$ for $y\in\mathcal{Y}$, we consider that an error has occurred.
\end{itemize}
It is easy to see that the probability of error of the constructed decoder (assuming uniform binary input to the channel $W_{x,x'}$) is equal to $\frac{1}{2}\mathbb{P}_{e,x}(W)+\frac{1}{2}\mathbb{P}_{e,x'}(W)=\mathbb{P}_{e,x,x'}(W)$. But since the ML decoder of $W_{x,x'}$ has the minimal probability of error among all decoders, we conclude that:

\begin{equation}
\label{eqe123}
\mathbb{P}_{e,x,x'}(W)\geq \mathbb{P}_e(W_{x,x'})= \frac{1}{2}\sum_{y\in\mathcal{Y}}\min\big\{W_{x,x'}(y|0),W_{x,x'}(y|1)\big\}=\frac{1}{2}\sum_{y\in\mathcal{Y}}\min\big\{W(y|x),W(y|x')\big\}.
\end{equation}
On the other hand, we have:
\begin{align*}
Z(W_{x,x'})&=\sum_{y\in\mathcal{Y}}\sqrt{W(y|x)W(y|x')}=\sum_{y\in\mathcal{Y}}\sqrt{\Big(\min\big\{W(y|x),W(y|x')\big\}\Big)\Big(\max\big\{W(y|x),W(y|x')\big\}\Big)}\\
&\stackrel{(a)}{\leq} \Big(\sum_{y\in\mathcal{Y}}\min\big\{W(y|x),W(y|x')\big\}\Big)^{1/2}\Big(\sum_{y\in\mathcal{Y}}\max\big\{W(y|x),W(y|x')\big\}\Big)^{1/2}\\
&\stackrel{(b)}{\leq} \sqrt{2\mathbb{P}_{e,x,x'}(W)}\Big(\sum_{y\in\mathcal{Y}}W(y|x)+W(y|x')\Big)^{1/2}= \sqrt{2\mathbb{P}_{e,x,x'}(W)}.\sqrt{2}=2\sqrt{\mathbb{P}_{e,x,x'}(W)},
\end{align*}
where (a) follows from the Cauchy-Schwartz inequality. (b) follows from \eqref{eqe123} and from the fact that $\max\big\{W(y|x),W(y|x')\big\}\leq W(y|x)+W(y|x')$. We conclude that:
\begin{equation}
\mathbb{P}_{e,x,x'}(W)\geq \frac{1}{4}Z(W_{x,x'})^2.
\label{eqe432}
\end{equation}

Now since $\displaystyle \mathbb{P}_{e}(W)=\frac{1}{|\mathcal{X}|}\sum_{x\in\mathcal{X}}\mathbb{P}_{e,x}(W)$, we have:
\begin{align*}
\sum_{\substack{x,x'\in\mathcal{X}\\x\neq x'}}\mathbb{P}_{e,x,x'}(W)&=\frac{1}{2}\Big(\sum_{\substack{x,x'\in\mathcal{X}\\x\neq x'}}\mathbb{P}_{e,x}(W)\Big)+\frac{1}{2}\Big(\sum_{\substack{x,x'\in\mathcal{X}\\x\neq x'}}\mathbb{P}_{e,x'}(W)\Big)\\
&=\frac{1}{2}\Big(\sum_{x\in\mathcal{X}}(|\mathcal{X}|-1)\mathbb{P}_{e,x}(W)\Big)+\frac{1}{2}\Big(\sum_{x'\in\mathcal{X}}(|\mathcal{X}|-1)\mathbb{P}_{e,x'}(W)\Big)\\
&=\frac{1}{2}(|\mathcal{X}|-1)|\mathcal{X}|\mathbb{P}_{e}(W)+\frac{1}{2}(|\mathcal{X}|-1)|\mathcal{X}|\mathbb{P}_{e}(W)=(|\mathcal{X}|-1)|\mathcal{X}|\mathbb{P}_{e}(W).
\end{align*}

Therefore,
\begin{align*}
\mathbb{P}_{e}(W)&=\frac{1}{(|\mathcal{X}|-1)|\mathcal{X}|}\sum_{\substack{x,x'\in\mathcal{X}\\x\neq x'}}\mathbb{P}_{e,x,x'}(W)\stackrel{(a)}{\geq} \frac{1}{(|\mathcal{X}|-1)|\mathcal{X}|} \sum_{\substack{x,x'\in\mathcal{X}\\x\neq x'}} \frac{1}{4} Z(W_{x,x'})^2 
\\
&\stackrel{(b)}{\geq} \frac{1}{4} \Big(\frac{1}{(|\mathcal{X}|-1)|\mathcal{X}|} \sum_{\substack{x,x'\in\mathcal{X}\\x\neq x'}} Z(W_{x,x'})\Big)^2=\frac{1}{4} Z(W)^2,
\end{align*}
where (a) follows from \eqref{eqe432} and (b) follows from the convexity of the mapping $t\rightarrow t^2$.

\section*{Acknowledgment}
I would like to thank Emre Telatar for enlightening discussions and for his helpful feedback on the paper. 

\bibliographystyle{IEEEtran}
\bibliography{bibliofile}
\end{document}